\documentclass[11pt]{amsart}

\usepackage{tikz}
\usetikzlibrary{arrows,decorations.markings,shapes.arrows,patterns}

\usepackage{amsmath,amsthm}
\pdfmapfile{+mathpple.map}
\usepackage{mathrsfs}
\usepackage{amscd,amssymb, amsfonts, verbatim,subfigure, enumerate}
\usepackage[mathcal]{eucal}
\usepackage[super]{nth}

\usepackage{mathpazo}

\usepackage{color,slashed}
\usepackage[colorlinks=true,linkcolor=blue,citecolor=red]{hyperref}

\newcommand{\wbar}{\br{w}}

\newcommand{\zbar}{\br{z}}

\newcommand{\PV}{\op{PV}}

\newcommand{\dpa}[1]{\frac{\partial}{\partial #1}}

\newcommand{\eps}{\epsilon}
\newcommand{\g}{\mathfrak{g}}

\newcommand{\what}{\widehat}

\newcommand{\til}{\widetilde}
\newcommand{\mscr}{\mathscr}

\newcommand{\br}{\overline}

\newcommand{\iso}{\cong}
\newcommand{\C}{\mathbb C}

\newcommand{\norm}[1]{\left\| #1 \right\|}
\newcommand{\Oo}{\mscr O}
\newcommand{\Z}{\mathbb Z}

\newcommand{\into}{\hookrightarrow}

\newcommand{\op}{\operatorname}

\newcommand{\mbb}{\mathbb}
\newcommand{\mf}{\mathfrak}
\newcommand{\mc}{\mathcal}

\newcommand{\ip}[1]{\left\langle #1 \right\rangle}
\newcommand{\abs}[1]{\left| #1 \right|}

\newcommand{\R}{\mbb R}
\renewcommand{\d}{\mathrm{d}}

\newcommand{\dbar}{\br{\partial}}

\DeclareMathOperator{\Sym}{Sym} \DeclareMathOperator{\Hom}{Hom}

\newtheoremstyle{thm}% name
  {7pt}%      Space above
  {7pt}%      Space below
  {\itshape}%         Body font
  {}%         Indent amount (empty = no indent, \parindent = para indent)
  {\bf}% Thm head font
  {.}%        Punctuation after thm head
  {5pt}%     Space after thm head: " " = normal interword space;
  {\thmnumber{#2 }\thmname{#1}\thmnote{ (#3)}}%         Thm head spec (can be left empty, meaning `normal')

\newtheoremstyle{def}% name
  {7pt}%      Space above
  {10pt}%      Space below
  {\itshape}%         Body font
  {}%         Indent amount (empty = no indent, \parindent = para indent)
  {\bf}% Thm head font
  {.}%        Punctuation after thm head
  {5pt}%     Space after thm head: " " = normal interword space;
  {\thmnumber{#2} \thmname{#1}\thmnote{ (#3)}}%         Thm head spec (can be left empty, meaning `normal')

\newtheoremstyle{rem}% name
  {4pt}%      Space above
  {10pt}%      Space below
  {}%         Body font
  {}%         Indent amount (empty = no indent, \parindent = para indent)
  {\itshape}% Thm head font
  {:}%        Punctuation after thm head
  {3pt}%     Space after thm head: " " = normal interword space;
  {}%         Thm head spec (can be left empty, meaning `normal')

\newtheoremstyle{texttheorem}% name
  {8pt}%      Space above
  {8pt}%      Space below
  {\itshape}%         Body font
  {}%         Indent amount (empty = no indent, \parindent = para indent)
  {\bf}% Thm head font
  {. \hspace{5pt}}%        Punctuation after thm head
  {3pt}%     Space after thm head: " " = normal interword space;
  {}%         Thm head spec (can be left empty, meaning `normal')

\newtheorem*{claim}{Claim}
\newtheorem*{theorem*}{Theorem}
\newtheorem*{lemma*}{Lemma}
\newtheorem*{corollary*}{Corollary}
\newtheorem*{proposition*}{Proposition}
\newtheorem*{definition*}{Definition}

\newtheorem*{conjecture}{Conjecture}

\newtheorem{theorem}{Theorem}[subsection]
\newtheorem{thm-def}{Theorem/Definition}[theorem]
\newtheorem{proposition}[theorem]{Proposition}

\newtheorem{lemma}[theorem]{Lemma}

\numberwithin{equation}{subsection}

\theoremstyle{def}
\newtheorem{definition}[theorem]{Definition}

\theoremstyle{rem}
\newtheorem*{remark}{Remark}

\newcommand{\cinfty}{C^{\infty}}

\usepackage{stmaryrd}
\date{}
\newcommand{\gl}{\mf{gl}}

\begin{document}

\title[$M$-theory in the $\Omega$-background]{$M$-theory in the $\Omega$-background and $5$-dimensional non-commutative gauge theory}

\author{Kevin Costello}

\address{Perimeter Institue for Theoretical Physics}
\email{kcostello@perimeterinstitute.ca}

\begin{abstract}
The  $\Omega$-background is defined for supergravity, and a very general class of such backgrounds is constructed for $11$-dimensional supergravity.  $11$-dimensional supergravity in a certain $\Omega$-background is shown to be equivalent to a $5$-dimensional non-commutative gauge theory of Chern-Simons type. $M2$ and $M5$ branes are expressed as $1$ and $2$-dimensional extended objects in the $5$-dimensional gauge theory.  This $5$-dimensional gauge theory is shown to admit a consistent quantization with two coupling constants, despite being formally non-renormalizable.

A check of a twisted version of AdS/CFT is performed relating this $5$-dimensional non-commutative gauge theory to the theory on $N$ $M5$ branes, wrapping an $A_{k-1}$ singularity and placed in an $\Omega$-background.   The operators on the $M5$ branes, in the $\Omega$-background, are described by a certain chiral algebra which in the large $N$ limit becomes a $W_{k+\infty}$ algebra. This chiral algebra is recovered from the $5$-dimensional gauge theory.
  
This argument also provides a holographic explanation of the result of Maulik-Okounkov and Schiffmann-Vasserot that the $W_{k+\infty}$ algebra acts on the equivariant cohomology of the moduli of instantons on an $A_{k-1}$ singularity.   
\end{abstract}

\maketitle

\section{Introduction}
Nekrasov's $\Omega$-background \cite{Nek03, NekWit10} has played a central role in recent investigations of supersymmetric gauge theory.  The essential features of an $\Omega$-background are the following.   If one has a supersymmetric field theory on a space-time $X$ with an $S^1$ action, one can try to deform the theory so that it acquires a supersymmetry whose square is the generator of rotation.  One then shows that objects (such as local operators) in the theory which are fixed by this supercharge can be described by an effective lower-dimensional field  theory living at the fixed points of the $S^1$ action on $X$ \cite{NekWit10, Yag14, BulDimGai15}.  Thus, the $\Omega$-background provides a field-theoretic version of the phenomenon of localization in equivariant cohomology.

In this paper I investigate the $\Omega$-background in supergravity.  Given a space-time manifold $X$ which has an $S^1$ isometry, we can look for a supergravity background whose underlying manifold is $X$ and which is equipped with a generalized Killing spinor whose square is the generator of rotation.   Putting a field theory in such a supergravity background has the effect of putting the field theory in an $\Omega$-background.

The following theorem gives a broad class of $11$-dimensional supergravity backgrounds equipped with such a generalized Killing spinor. 
\begin{theorem} 
Let $M$ be a manifold with holonomy in $G2$, and let $Z$ be a hyper-K\"ahler $4$-manifold.  Suppose that $M$ has an $S^1$ action.  

Then, there exists a complex $3$-form $\eps C$ on $M \times Z$ such that $(M \times Z,\eps C)$ satisfies the equations of motion of $11$-dimensional supergravity.  There also exists a generalized Killing spinor $\Psi(\eps)$ on $(M \times Z, \eps C)$ whose square is $\eps V$, where $V$ is the vector field generating rotation.   
\end{theorem} 

When $\eps = 0$, $\Psi(\eps)$ is one of the four covariant constant spinors one has on a manifold of $G2 \times SU(2)$ holonomy.  The $3$-form field $C$ is 
$$
C = V^{\flat} \wedge \omega^{0,2}
$$
where $\omega^{0,2}$ is the covariant constant $(0,2)$ form on $Z$ (in a particular complex structure), and $V^{\flat}$ is the $1$-form on $M$ constructed from the vector field $V$ using the metric. 

\subsection{}
One reason that one might be interested in supergravity backgrounds of this nature is that putting a $D$- or $M$-brane in such a background will automatically give an $\Omega$-background construction on the brane.   But one might also have a more ambitious goal: to understand the AdS dual of field theories in the $\Omega$-background.  This paper is a  first step in a project to investigate AdS/CFT for $M$-theory in the $\Omega$-background. 

To be able to achieve this, we would need to see if there is a localization phenomenon in supergravity in the $\Omega$-background similar to the one present in field theory.  For this to happen, it is essential to incorporate the generalized Killing spinor which squares to rotation.

In \cite{CosLi16}, we gave a general  definition of twisting supergravity when one has a supergravity background with a Killing spinor which squares to zero.   We define a \emph{twisted background} to be a supergravity background where the bosonic ghost is set to equal the square-zero Killing spinor.  This is the analog of twisting in supersymmetric field theory. Putting a supersymmetric theory in a twisted supergravity background has the effect of adding a supersymmetry of square zero to the BRST operator of the supersymmetric field theory, which is the essential part of procedure of twisting the theory. We define twisted supergravity to be supergravity in perturbation theory around a twisted background. 

I will give a similar definition for supergravity in the twisted $\Omega$ background. Given a supergravity background with a generalized Killing spinor which squares to a rotation, the twisted $\Omega$ background is defined to be the supergravity background where the bosonic ghost takes value the chosen generalized Killing spinor\footnote{There are some subtleties with this definition which we will discuss in the  body of the paper.}.  We can then define supergravity in the twisted $\Omega$ background to be supergravity in perturbation theory around such a background.   I will argue that the localization phenomenon familiar from the $\Omega$-background in field theory holds for supergravity in the twisted $\Omega$ background as well.  
 
\subsection{}
The main statements of this paper concern $11$-dimensional supergravity and $M$-theory in a twisted $\Omega$ background.  Let $TN_k$ denote the Taub-NUT manifold with an $A_k$ singularity at the origin.  Consider $M$-theory on $TN_k \times \R^3 \times \C^2$.  We consider $TN_k \times \R^3$ to be a manifold with holonomy inside $G2$, and $\C^2$ to be a manifold with holonomy inside $SU(2)$, so that our general construction of $\Omega$-backgrounds applies. 

There is an $S^1 \times S^1$ action on $TN_k \times \R^2$ which preserves the holomorphic volume form.  One $S^1$ rotates the circle fibre of the Taub-NUT and the other simultaneously rotates the base of the Taub-NUT fibration and the plane $\R^2$.  This action preserves the $G2$ structure on $TN_k \times \R^3$. Our general construction gives a twisted $\Omega$-background for $11$-dimensional supergravity, with two parameters $\eps,\delta$, where $\eps$ corresponds to rotation of the Taub-NUT circle fibre and $\delta$ to the other rotation. 

Our main result is the following.
\begin{proposition}
$M$-theory in this twisted $\Omega$-background is equivalent to a $5$-dimensional non-commutative Chern-Simons gauge theory on $\R \times \C^2$, with complex gauge group $GL(k+1,\C)$.

This gauge theory has fundamental field given by a $3$-component partial connection
$$
A = A_t \d t + A_{\zbar_1}\d \zbar_1  + A_{\zbar_2} \d \zbar_2 
$$
where each component is a smooth map from $\R \times \C^2$ to $\mf{gl}(k+1,\C)$.

The action functional is the non-commutative Chern-Simons action  
\begin{multline*}
S(A) = \tfrac{1}{2\delta } \int \op{Tr} \left( A \d A\right) \d z_1 \d z_2 + \tfrac{1}{3\delta } \int \op{Tr}\left( A^3 \right)\d z_1 \d z_2  \\
+  \int \op{Tr} \left\{ 
\sum_{n \ge 1} \tfrac{\eps^n}{3 \delta 2^n n!} \varepsilon_{i_1j_1} \dots \varepsilon_{i_n j_n} A \left( \dpa{z_{i_1}} \dots \dpa{z_{i_n}} A \right) \wedge \left( \dpa{z_{j_1}} \dots \dpa{z_{j_n}} A \right)\right\} \d z_1 \d z_2.
\end{multline*}
Here $\varepsilon_{ij}$ is the alternating symbol, and in the second line the sum over indices $i_r$, $j_s$ in the set  $\{1,2\}$ is understood.  
\end{proposition}
If $\ast_\eps$ indicates the Moyal product defined using the Poisson tensor $\dpa{z_1} \dpa{z_2}$, we can write the action more succinctly as
$$
\tfrac{1}{2\delta } \int \op{Tr} A \d A \d z_1 \d z_2 + \tfrac{1}{3 \delta} \int A\ast_\eps A \ast_\eps A \d z_1 \d z_2.  
$$
 The gauge transformations in this theory are modified to incorporate the Moyal product as well.

\subsection{}
This theory is highly non-renormalizable, as the action functional contains infinitely many derivatives.   Nevertheless, we have the following result. 
\begin{theorem*}
This $5$-dimensional gauge theory exists at the quantum level, in perturbation theory in the regime where $\abs{\delta} \ll \abs{\eps}$. Even though the theory is formally non-renormalizable,  all counter-terms are fixed uniquely (up to a redefinition of the coupling constants $\eps$ and $\delta$).   
\end{theorem*}
Roughly, the reason that this theorem holds is because the very large gauge symmetry fixes the counter-terms uniquely. The theorem is proved by a cohomological analysis of the possible counter-terms.

\subsection{} 
This paper also gives a description of $M2$ and $M5$ branes.  I argue that an $M2$ brane wrapping $\R^3$ in $TN_k \times \R^3 \times \C^2$ becomes the instanton particle of the non-commutative $5$-dimensional gauge theory.   Further, the theory living on the $M2$ brane becomes holomorphic quantum mechanics on the moduli of non-commutative instantons on $\C^2$.  

An $M5$ brane wrapping $TN_k$ and a holomorphic curve $\Sigma \subset \C^2$ becomes a t'Hooft surface operator in the $5$-dimensional gauge theory. 

\subsection{}
In broad outlines, the argument I present which justifies this description of $M$-theory in the $\Omega$-background goes as follows.  The most important step in the argument is the following.
\begin{proposition*} 
In $M$-theory  in the $\Omega$ background on $TN_k \times \R^2 \times \R \times \C^2$, varying the radius of the Taub-NUT is $Q$-exact for the supersymmetry we are using. 
\end{proposition*} 
This means that we will get the same answer if the Taub-NUT radius is small or large.   

If we take the radius of the Taub-NUT to be small, then of course everything simplifies greatly.  In this limit we can approximate $M$-theory on $TN_k \times \R^3 \times \C^2$ by type IIA superstring theory, with $k+1$ $D6$ branes. When $\eps = \delta = 0$, we can give a description of the twist of the theory on the $D6$ branes and a (conjectural) description of the twist of type IIA supergravity using the results of \cite{CosLi16}.  The twist of the theory on the $D6$ branes is a gauge theory on $\R^3 \times \C^2$ which is holomorphic on $\C^2$ and topological on $\R^3$.  

Next, let us take $\eps$ to be non-zero but $\delta = 0$.  We find that the $C$-field we need to introduce in $11$ dimensions becomes a $B$-field in the $\C^2$ direction, which makes the theory on the $D6$ brane non-commutative in those directions \cite{SeiWit99}. We also argue, using the description of type IIA supergravity given in \cite{CosLi16}, that after introducing this $B$-field type IIA supergravity becomes essentially trivial and we don't need to include it. 

Thus, when $\delta = 0$, we have a description in terms of a $7$-dimensional non-commutative gauge theory.   To figure out what happens when we turn on $\delta$, we use a fairly standard analysis of gauge theories in the $\Omega$ background to give a description in terms of a $5$-dimensional theory.

This shows that when the Taub-NUT radius is small, we can give an effective description of a BPS sector of $11$-dimensional supergravity in terms of the $5$-dimensional non-commutative gauge theory.  However, since the  variation of the Taub-NUT radius is $Q$-exact, this description should also holds at large radius. This completes the proof.

Some aspects of the arguments sketched above are a little more indirect than would be ideal. It would be desirable to, for example, explicitly relate the solutions to the equations of motion of the $5$-dimensional gauge theory (which are non-commutative instantons) to the supersymmetric solutions to the equations of motion of $11$-dimensional supergravity. Unfortunately I was unable to do this.

\subsection{}
Let me try to compare this work to existing work \cite{HelOrlRef12} on $M$-theory in $\Omega$-backgrounds.   These authors construct backgrounds (called fluxtrap or fluxbrane backgrounds) for string and $M$-theory by starting with geometries consisting of an $\R^8$ bundle over a torus and applying various dualities.  The relationship between the $M$-theory backgrounds considered here and those considered in \cite{HelOrlRef12} is not at all clear to me; it would be very interesting to clarify this point. 

Unfortunately, I was unable to use the fluxtrap backgrounds of \cite{HelOrlRef12} for the purposes of the present paper. This was my motivation for developing the approach used here. Another advantage of the class of $M$-theory backgrounds discussed here has the advantage that it is very general: given any $G2$ manifold with an $S^1$ action we can construct an $\Omega$-background for $M$-theory.  A disadvantage of the backgrounds considered here is that the $C$-field is necessarily complex.

\subsection{}
As a test of this picture, I investigate holography for $M5$ branes in the $\Omega$-background, using a proposal  for twisted holography developed with Si Li.  The holographic statement  I prove  relates the algebra of operators on $N$ $M5$ branes wrapping $TN_k$, in the $\Omega$ background and the large $N$ limit,  with the $5$-dimensional non-commutative gauge theory in the presence  of the t'Hooft operator given by the $M5$ branes.    

To prove this statement, we first need to understand the algebra of operators on $N$ $M5$ branes in the $\Omega$-background.  For formal supersymmetry reasons this will  be a chiral algebra, that is, it will have the structure of the algebra of operators of a two-dimensional chiral conformal field theory.  By an argument involving reduction to type IIA, I show that this chiral algebra can  be identified with the algebra of operators living on a surface defect in $GL(N)$ Chern-Simons theory,  where on the surface  defect we couple Chern-Simons theory to $k+1$ fundamental and $k+1$ anti-fundamental chiral free fermions.  

This description of the operators on an $M5$ brane is of interest independent of its application to holography. Indeed, whatever chiral algebra lives on the $M5$ branes in the $\Omega$-background will act on the equivariant cohomology of the moduli of instantons on an $A_k$ singularity.  In the case $k = 0$, I expect the chiral algebra to be the $W_N$ algebra, in line with the AGT correspondence \cite{AldGaiTac10}.  In the general case, Davide Gaiotto provided an argument which strongly suggests that the chiral algebra considered here is isomorphic to a parafermionic algebra considered in \cite{NisTac11} in their work on generalizing AGT to include $A_k$ singularities.   
 
We will work with the specialization where $\delta = 0$, so that  the equivariant parameters corresponding to the rotation of the two planes in the $M5$ brane world-volume sum to zero.  In this specialization, the Chern-Simons theory is treated classically, and the only effect that Chern-Simons theory has on the operators living on the surface defect is that it forces us to take the $GL(N)$ invariant operators.  Taking the large $N$ limit, I show that the algebra of operators on the surface defect becomes a $W_{k+1+\infty}$ algebra with central extension set to $N$.  

The currents of the $W_{k+1+\infty}$ algebra form an associative algebra which is the universal enveloping algebra of a certain infinite-dimensional Lie algebra. This Lie algebra is a central extension of the Lie  algebra  $\C[z,z^{-1},\partial_z] \otimes \mf{gl}_{k+1}$ of matrices whose entries are polynomials in $z,z^{-1}$ and $\partial_z$.  Inside this Lie algebra is the usual  Kac-Moody current algebra $\C[z,z^{-1}] \otimes \mf{gl}_{k+1}$; the central extension is determined by its behaviour on this subalgebra.   When we set the central parameter to take value $N$, the $W_{k+1+\infty}$ algebra maps to the chiral algebra living on $N$ $M5$ branes.   

In the case $k  = 0$, this is consistent with a  number of results in the literature.  Yagi \cite{Yag12} has argued directly that the algebra of operators on $N$ $M5$ branes in the $\Omega$ background is the $W_N$-algebra. This statement is of course closely related to the work of Alday, Gaiotto and Tachikawaa \cite{AldGaiTac10}.  Further, a result of Beem, Rastelli and van Rees \cite{BeeRasvan15} shows that a construction similar in spirit (but different in details) to the $\Omega$-background construction produces the $W_N$ algebra from the theory on a stack of $N$ $M5$ branes.  When $N \to \infty$, the $W_N$ algebra becomes the $W_{1+\infty}$ algebra.  

When $\delta \neq 0$, the algebra of operators on the $M5$ branes becomes a deformation of the $W_{k+1+\infty}$ algebra. It is natural to conjecture that this deformation is the same as the one analyzed by Maulik and Okounkov \cite{MauOko12}: they call it the Yangian for  the affine Lie algebra $\mf{gl}_{k+1}$.  Maulik-Okounkov showed that this Yangian acts on the equivariant cohomology of the moduli of instantons on an $A_k$ singularity.  The chiral  algebra living on $N$ $M5$ branes will also act on the equivariant cohomology of the moduli of instantons of rank $N$ on the $A_k$ singularity (one can see this by relating this chiral algebra to the instanton operators in $5$-dimensional maximally supersymmetric gauge theory).  This is why I expect to find the Yangian considered by Maulik-Okounkov as the algebra of operators on $N$ $M5$ branes when $N \to \infty$.   
\subsection{}

The specialization $\delta = 0$ we are considering is such that the gravitational dual theory (which is our $5$-dimensional non-commutative gauge theory) is treated at tree level.   To prove holography, we need to identify the $W_{k+1+\infty}$ algebra in the gravitational theory.

At a rather crude level,  this is almost immediate. As we have seen, the $W_{k+1+\infty}$ algebra is the universal enveloping algebra of a  central extension of the Lie algebra $\mf{gl}_{k+1}[z,z^{-1}, \partial_z]$.  This Lie algebra is nothing but the gauge symmetries preserving the trivial field configuration of the $5d$ gauge theory on $\R \times \C^\times \times \C$.  

Indeed, the equations of motion of the $5$-dimensional gauge theory describe a bundle on $\R \times \C^\times \times \C$, which is flat on $\R$ and holomorphic on $\C^\times \times \C$ in a compatible way, and where $\C^\times \times \C$ is non-commutative.  Without the non-commutativity, the gauge symmetries preserving the trivial such bundle would the group of holomorphic  maps from $\C^\times \times \C$ to $GL(k+1)$. The Lie algebra of the group of gauge symmetries is $\C[z,z^{-1},w] \otimes \mf{gl}_{k+1}$.   Making $\C^\times \times \C$ makes $w$ and $z$ obey canonical  commutation relations $[w,z]  = \eps$, and the parameter $\eps$ can be absorbed into a redefinition  of $w$.   This turns the Lie algebra of gauge symmetries into the Lie algebra $\C[z,z^{-1},\partial_z] \otimes \mf{gl}_{k+1}$.

This shows that this infinite-dimensional Lie algebra is the symmetries of the background for the $5d$ gauge theory which is holographically dual to the $M5$ branes in an $\Omega$-background.

\subsection{}
One has to work quite a bit harder to see the central extension of this Lie algebra.  We will find that this central extension arises from the flux sourced by the $M5$ branes.   To see this central extension, it is convenient to reduce the $5$-dimensional theory to $3$-dimensions in such a way that the $M5$ branes live on the boundary.

Let us give coordinates $t,z,w$ to our $5$-dimensional space-time $\R \times \C\times \C$, and suppose that the $M5$ branes live at $ t = w = 0$.   We can form the holographic dual of the  theory living on the $M5$ branes by removing the locus $t = w = 0$ and introducing the flux sourced by the $M5$ branes.  (Since the theory does not have a metric, passing to the near horizon limit will play no role). 

Once we remove the locus where the $M5$ branes live, our space-time becomes
$$
(\R_t \times \C_w \setminus (0,0) )  \times \C_z. 
$$
We can identify
\begin{align*} 
\R_t \times \C_w \setminus (0,0)  &\iso  \R^+_r \times S^2\\
(t,w) & \mapsto \left( r = (t^2 + w \br{w})^{1/2}, \frac{(t,w)}{(t^2 + w \br{w} )^{1/2} } \right) 
\end{align*}
Using this isomorphism, we can identify the $5$-dimensional space obtained by removing the location of the $M5$ branes with 
$$
\R^+_r \times \C_z \times  S^2
$$
(we should think of this space as  an analog in our context of $AdS_3 \times S^2$). 

Reducing the $5$-dimensional theory along $S^2$ gives a $3$-dimensional theory with an infinite tower of Kaluza-Klein modes.     The lowest-lying KK modes give us a gauge field $A$ for the Lie algebra $\mf{gl}_{k+1}$, with Chern-Simons action at level $N$.

The higher KK modes are bit more complicated, but let me sketch a description.  The theory will not be invariant under the full $3$-dimensional symmetries of $\R^+_r \times \C_z$.  Instead it will be invariant under diffeomorphisms of $\R^+_r$ and isometries of $\C_z$. 

 We have fields $\alpha_i$, for $i > 0$, and $\beta_{-i}$, for $i > 1$.  The field $\alpha_i$ is a sum of two tensors 
$$
\alpha_i = \alpha_i^r \partial_z^i \d r + \alpha_i^{\zbar} \partial_z^{i} \d \zbar
$$ 
(where $\partial_z^{i}$ is included to indicate that each  term transforms as a section of the $i$th power of the holomorphic tangent bundle of $\C_z$).

The field $\beta_{-i}$ is a section
$$
\beta_{-i} \in \cinfty(\R^+_r \times \C_z, ((T^\ast\C_z)^{\otimes i})
$$
so that $\beta_{-i}$ is a scalar times $\d z^{i}$. 

We can write down the action functional in the specialization where $\eps$, the parameter of non-commutatavity, is zero: 
$$
 N \int_{\R^+_r \times \C_z} CS(A) + \frac{1}{\delta}\sum_{i > 0} \int_{\R^+_r \times \C_z}\op{Tr} \beta_{-i-1} \left( \partial_{\zbar}^A  \alpha_i^r   -  \partial_r^A \alpha_i^{\zbar}\right) + \frac{1}{2 \delta} \sum_{i,j > 0}\op{Tr} \beta_{-i-j-1}\alpha_i \alpha_j + O(\eps).  
$$
Here $\partial_{\zbar}^A$ indicates the covariant derivative for the connection $A$. (The action including all terms in the $\eps$-expansion is written down in section \ref{section_dimensional_reduction_3d}).  

In the case $k = 0$, the $3$-dimensional theory with these fields and action should thus be holographically dual  to the chiral Toda theory, whose algebra of operators is the $W_N$ algebra.   (In \cite{GabGop12}, it is argued that the holographic dual of the full (i.e.\ non-chiral) $W_N$ theory is a Vasiliev higher-spin theory in three dimensions.  I could not see  a direct connection between the holographic dual of the chiral $W_N$ theory constructed here and that for the full $W_N$ theory considered in \cite{GabGop12}.)

One expects from holography that the algebra of operators on $N$ $M5$ branes, as $N \to \infty$, can be described as the operators on the boundary of this $3$-dimensional reduction of our $5$-dimensional theory, for a certain natural boundary condition at $r = \infty$. The boundary condition is the one where the fields $\alpha_i$ are zero and the $(0,1)$ component $A^{0,1}$ of the connection is zero.  The boundary operators will then be built from functions of $A^{1,0}$ and  of the fields $\beta_{-i}$. 

In the limit in which our non-commutative gauge theory is  treated at the tree level, this algebra of boundary operators is calculated explicitly (using Feynman diagrams).  We find that it is indeed the $W_{k+1+\infty}$ algebra with the correct value of the central extension, thus verifying holography.  We find, for instance, that the boundary operators built from the connection $A$ form a copy of the Kac-Moody algebra at level $N$ (this is nothing but the well-known fact that the Kac-Moody algebra is the algebra of operators on the boundary of Chern-Simons theory).   The operators built from the fields $\beta_{-i}$ provide the spin $2$ and higher generators of the $W_{k+1+\infty}$ algebra.

\subsection{}

It is much harder to analyze what happens  when we include quantum corrections on the gravitational side.   To do this, we need to work directly in $5$ dimensions, as the $3$-dimensional description is not adequate for performing loop-level computations in the theory.   I describe an abstract construction of a chiral algebra directly from the $5$-dimensional theory, which at tree level reproduces the construction given above in $3$-dimensional language. I conjecture that this chiral algebra that arises is the Yangian for affine $\mf{gl}_{k+1}$ studied by Maulik and Okounkov \cite{MauOko12}, which deforms the $W_{k+1+\infty}$ algebra.  

\subsection{}
In a sequel to this paper, holography for $M2$ branes is analyzed. There, stronger results can be obtained. The algebra of operators on $N$ $M2$ branes is a deformation quantization of the algebra of holomorphic functions on the moduli of instantons on $\C^2$ of rank $k+1$ and charge $N$.   I will show that, in the large $N$ limit, this algebra can be recovered from the $5$-dimensional gauge theory. In contrast to the situation for $M5$ branes, in this example the correspondence can be checked to all orders in the loop expansion of the $5$-dimensional gauge theory.

\subsection{Acknowledgements}
I'm very grateful to Davide Gaiotto, Jaume Gomes and Edward Witten for patiently answering my many questions about string theory.  I'd particularly like to thank Davide Gaiotto for contributing the arguments in section \ref{section_M5} relating the constructions discussed here to the AGT correspondence.  

\tableofcontents

\section{Generalities on the $\Omega$-background}
Before giving an analysis $11$-dimensional supergravity in the $\Omega$-background, I want to make some general comments about gauge theory and supergravity in the $\Omega$-background, expanding in the brief discussion in the introduction. 

We will start by disucssing, in general terms, how to put a field theory in the $\Omega$-background. The goal will be to get enough of an understanding of this construction that we can import it into supergravity.

Suppose we have a  manifold $X$ with an $S^1$ action. Let $M \subset X$ be the set of fixed points, which we assume is a submanifold. Suppose that, away from $M$, the $S^1$ action is free.   Suppose further that the quotient $Y = X/S^1$ is a manifold, and that $M \subset Y$ is a submanifold.  

Next, suppose we have a field theory on $X$, which we call $T_0$, which is invariant under the $S^1$ action.  Further, suppose that the field theory has a fermionic symmetry $Q_0$ with $Q_0^2 = 0$.   This allows us to twist the theory, by adding $Q_0$ to the BRST operator.    

This happens, for instance, if  $X$ is a Riemannian $4$-manifold with an $S^1$ action by isometries, and we put the Donaldson-Witten twist of $N=2$ Yang-Mills on $X$.  The supercharge $Q_0$ is the globalization of the supercharge invariant under the twisted action of $\op{Spin}(4)$ on the space of fields on flat space.

To put the theory in the $\Omega$-background, we need further data.  We need the following.
\begin{enumerate} 
\item A deformation of our theory $T_0$ into a one-parameter family of theories $T_\eps$ on $X$, each invariant under the $S^1$ action.  
\item A family $Q_\eps$ of fermionic symmetries of the theory $T_\eps$.
\item Let $V_\eps$ be the infinitesimal symmetry of the theory $T_\eps$ arising from the $S^1$ action. We require that 
$$
Q_\eps^2 = \eps V_\eps. 
$$  
\end{enumerate} 
In typical examples, to construct this data, one takes a supersymmetric theory, puts it on a manifold $X$ possibly in a twisted way (i.e.\ working with a non-trivial background gauge field for the $R$-symmetry group) so that the supercharge $Q_0$ is preserved. This produces the theory $T_0$. The family of theories $T_\eps$ is produced by working in some non-trivial supergravity background, e.g.\ with some fluxes.  The spinor $Q_\eps$ is then a generalized Killing spinor for this supergravity background.  I will not give any details right now: later we will discuss in detail a class of $11$-dimensional supergravity backgrounds which have the desired features. Our aim here is rather to elucidate some formal aspects of the $\Omega$-background construction, for which specific examples are not needed.

To produce the $\Omega$ background, we need to ``twist''\footnote{Some authors do not emphasize the importance of this twisting in the $\Omega$ background construction, although in essentially all applications, e.g.\ instanton counting, one uses this twisting} the family of theories $T_\eps$ by the family of fermionic symmetries $Q_\eps$.  Normally, twisting would involve adding a square-zero supercharge to the BRST operator.  In this situation, this does not work, because $Q_\eps^2 = \eps V_\eps$.  If we added $Q_\eps$ to the BRST operator, we would be in a situation akin to that of a theory with a gauge anomaly: the new BRST operator would not square to zero.  

The solution to this difficulty is obtained by mimicking the Cartan approach to equivariant cohomology.  Let us view the theories $T_\eps$ as being a family of theories on $Y$ which we call $\pi_\ast T_\eps$  by compactifying along the map $\pi : X \to Y= X/S^1$, and including all KK modes. The theory $\pi_\ast T_\eps$ has a defect on the submanifold $M \subset Y$ which is the image of the $S^1$ fixed points of $X$. Each $\pi_\ast T_\eps$ has an $S^1$ global symmetry, coming from the $S^1$ action on $X$. We can define a new family of theories $\pi_\ast^{S^1} T_\eps$  on $Y$ by including only $S^1$-invariant objects of the theory $\pi_\ast T_\eps$. E.g. the space of local operators of $\pi_\ast^{S^1} T_\eps$ consists of those operators of $T_\eps$ which are local from the point of view of $Y$ and are $S^1$-equivariant.  

Now, on $\pi_\ast^{S^1} T_\eps$, the operator $Q_\eps$ has square zero, since we are including only $S^1$-invariant objects.  We define the theory in the $\Omega$ background to be the theory $\pi_\ast^{S^1} T_\eps$ but where we twist by $Q_\eps$, i.e.\ we add $Q_\eps$ to the BRST operator.

\subsection{}
One aspect of this construction which might appear a little unsatisfactory is that, at first sight, it does not appear that the theory in the $\Omega$-background is a continuous deformation of the twist of the theory $T_0$ (or rather, of its twist when reduced to $Y$).  To twist $\pi_\ast T_\eps$, we needed to take $S^1$-invariants, but to twist $T_0$ we do not need to do this. 

However, the family of twisted theories is continuous.  The point is that the $S^1$-action on the twist of the theory $\pi_\ast T_0$ is trivial.  To see this, let $\d_\eps$ denote the BRST operator of the family of theories $T_\eps$. We can expand
\begin{align*} 
\d_\eps &= \d_0 + \eps \d_1 + \dots\\
Q_\eps &= Q_0 + \eps Q_1 + \dots 
\end{align*}
and view $\d_i$, $Q_i$ as acting on the theory $T_0$. Then, 
$$
[\d_0 + Q_0,\d_1 + Q_1] = V  
$$
is the vector field generating circle rotation.  This vector field is therefore exact for the sum of the BRST operator and the supercharge $Q_0$ in the theory $T_0$. Taking $S^1$ invariants therefore has no effect, so that we have a continuous family of theories.

\section{Supergravity in the $\Omega$-background}
\label{supergravity_omega}
Our main concern is supergravity in the $\Omega$-background.  Before we discuss this, let us remind the reader of the concept of twisting supergravity.

In supergravity, we gauge local supersymmetries, so that there are bosonic spinorial ghost fields corresponding to these local supersymmetries.  We can consider a supergravity background where some of the ordinary bosonic fields (e.g.\ the metric, various fluxes) have a non-zero value, but where also the bosonic ghost field has a non-zero value. To satisfy the equations of motion, these fields must satisfy the following:
\begin{enumerate} 
  \item Ordinary bosonic fields must satisfy the usual supergravity equations of motion.
\item  The bosonic ghost field must correspond to a generalized Killing spinor for this supergravity background.
\item The square of this generalized Killing spinor must be zero.  
\end{enumerate}
In \cite{CosLi16} such a background is called a twisted background. 
 
The importance of this concept is that it is the supergravity analog of the familiar procedure of twisting a supersymmetric field theory by adding some supersymmetry to the BRST operator.    Indeed, if we place a supersymmetric field theory in a twisted supergravity background, then the bosonic ghost field of the supergravity theory has precisely the effect of adding the corresponding supersymmetry to the BRST operator of the supersymmetric field theory.

We would like to generalize this to give a version of the $\Omega$-background for supergravity.  For now, we will work in a general setting without focusing on the details of the supergravity theory.  The example of $11$-dimensional supergravity will be presented shortly. 

Suppose we have a $d$-dimensional manifold $X_0$ with an $S^1$ action.  Suppose that $X_0$, with a metric and possibly some other non-zero fields, is a solution to the equations of motion for some supergravity theory. Suppose further that $X_0$ is equipped with a generalized Killing spinor $\Psi_0$ of square zero. Finally, we require that all of this data is preserved  by the $S^1$ action on $X_0$.  

As in the case of a field theory, we need to have extra data in order to consider an $\Omega$-background for supergravity.  We need:  
\begin{enumerate} 
\item A family $X_\eps$ of supergravity backgrounds deforming $X_0$, each with an $S^1$ action deforming that on $X_0$. 
\item Each $X_\eps$ is equipped with a generalized Killing spinor $\Psi_\eps$.
\item We require that $\Psi_\eps^2  = \eps V$, where $V$ is the vector field generating the $S^1$ action.  
\end{enumerate}
We let $Y_\eps = X_\eps/S^1$, and consider reducing the theory on $X_\eps$ to one on $Y_\eps$.
 
We would like to consider the theory on $Y_\eps$ in the background where the bosonic ghost is set to $\Psi_\eps$. Unfortunately, this does not satisfy the equations of motion, because the equations of motion require that the spinor corresponding to the bosonic ghost squares to zero.

If one just directly reduces the fields of supergravity from $X_\eps$ to $Y_\eps$, one would see a ghost for the $S^1$ gauge symmetry, including a ghost for the gauge transformation which is constant on $Y_\eps$.  This constant ghost is what causes problems for us: if $V$ denotes this constant ghost, then $\Psi(\eps)^2 = \eps V$.  

The solution is simple. When we consider the supergravity theory on $Y_\eps$, we remove the constant ghost for the $S^1$ gauge symmetry, and require that all operators we consider are strictly invariant under the $S^1$ action.  Indeed, this is the natural approach to treating gauge theories with compact gauge group when one wants to take account of ``large'' gauge transformations.   If we do this, we no longer have the constant ghost field $V$, and wehave $\Psi_\eps^2 = 0$.  

Thus, we have the following definition of supergravity in the $\Omega$-background.
\begin{definition} 
Suppose we have a supergravity background $X_\eps$ as above. We define supergravity in the twisted $\Omega$-background to be the supergravity theory on $Y_\eps = X_\eps/S^1$, obtained in the following way:
\begin{enumerate} 
\item We reduce the theory from $X_\eps$ to $Y_\eps$, including all KK modes.
\item We remove the ghost for the constant $S^1$ gauge symmetry from the fields of the theory on $Y$.  Instead, when discussing operators or other quantities, we require them to be invariant under constant gauge transformations. (Invariance under non-constant gauge transformations is enforced by ghosts for non-constant gauge fields). 
\item Then, we consider the theory in the background where the bosonic ghost is set to the spinor $\Psi_\eps$.  
\end{enumerate} 
\end{definition}

\section{$11$-dimensional supergravity in the $\Omega$-background} 
\label{section:11dsugra_omega}
In this section we will see how to place $M$-theory in an $\Omega$-background.  We will give a very general construction of certain $11$-dimensional supergravity backgrounds equipped with a generalized Killing spinor which squares to the generator of a rotation.  We will use  a special case of this general construction (built from the Taub-NUT manifold) to derive the relationship between $M$-theory and $5$- and $7$-dimensional non-commutative gauge theories. 

\subsection{}
Consider $M$-theory on an $11$-manifold of the form $M^7 \times Z^4$ where $M$ is a $G2$ manifold\footnote{By which we mean it has a $3$-form satisfying the properties necessary to define a $G2$ structure. This is the same as having a reduction of structure group of the frame bundle to $G2$ such that the holonomy of the corresponding metric is contained in $G2$. } and $Z$ is hyper-K\"ahler. 

Choosing one of the complex structures on $Z$ picks out a covariant constant spinor $\psi_Z$ on $Z$, and a holomorphic symplectic form $\omega^{2,0}_Z$. Let $\omega^{0,2}_Z$ be the complex conjugate of $\omega^{2,0}_Z$, normalized so that $\omega^{2,0}_Z \wedge \omega^{0,2}_Z$ is the Riemannian volume form on $Z$.  The spinor $\psi_Z$ is related to the complex structure by the fact that if $\lambda$ is a $(1,0)$ form in this complex structure then the Clifford multiplication $\lambda \cdot \psi_Z$ is zero.

The $G2$ structure on $M$ also picks out a covariant constant spinor $\psi_M$ on $M$, associated to the element of the $8$-dimensional spin representation of $\op{Spin}(7)$ which is invariant under $G2$.

Suppose that $M$ has an action of the group $U(1)$ by isometries which also preserve the $3$-form defining the $G2$ structure on $M$. Let $V$ denote the vector field which is the generator of this action, and let $V^{\flat}$ denote the $1$-form dual to $V$ using the Riemannian metric. 

Since $V$ is a Killing vector field, the equiation
$$
g( \nabla_\alpha V, \beta ) + g ( \nabla_\beta V, \alpha) = 0
$$ 
holds for all vector fields $\alpha,\beta$ on $M$. Since $g(\nabla_\alpha V, \beta)$ is anti-symmetric in $\alpha$ and $\beta$, it defines a $2$-form $\omega_V$ on $M$. We have
$$
\omega_V = \d V^{\flat}. 
$$ 
Thus, $\omega_V$ is closed. It is also co-closed, $\d^\ast \omega_V = 0$.  This is equivalent (using a Weitzenbock identity) to the statement that the $1$-form $V^{\flat}$ is harmonic, and the $1$-form associated to any Killing vector field on a Ricci flat manifold is always harmonic.

In the appendix, we will show the following.   
\begin{theorem}
\label{theorem_supergravity_solution}
Consider the $11$-manifold $M \times Z$ equipped with the $4$-form
$$
F =  \omega_V \wedge \omega^{0,2}_Z
$$
where $\omega^{0,2}_Z$ is the complex conjugate of the holomorphic symplectic form on $Z$. We will view $F$ as a flux defining an $M$-theory background. Then, for any pair of constants $\eps,\alpha$, 
\begin{enumerate} 
\item $M \times Z$ with the flux $-\eps\alpha^{-2} F$ satisfies the equations of motion of $11$-dimensional supergravity (where $\eps$ is a constant). 
\item The spinor 
$$\Psi(\eps)= \alpha  \psi_M \otimes \psi_Z + \alpha^{-1} \eps (V \cdot \psi_M) \otimes (\omega^{0,2}_Z \cdot \psi_Z)$$
is a generalized Killing spinor for the supergravity background with flux $-\eps\alpha^{-2} F$.  Here, $\cdot$ indicates Clifford multiplication. 
\item The square of $\Psi(\eps,\alpha)$ is the Killing vector field $\eps V$ on $M \times Z$.  
\end{enumerate}
\end{theorem}
Now, we are in the situation discussed in section \ref{supergravity_omega}.   Let us follow the construction there to put $11$-dimensional supergravity in the $\Omega$-background. Let 
$$
X_\eps = (M \times Z,-\eps \alpha^{-2} F)  
$$
be the $11$-dimensional supergravity background, and let $Y_\eps = X_\eps / S^1$ where as before $S^1$ acts on $M$.  Let us assume, as before, that the quotient is reasonably well behaved.    Then, as we explained in section \ref{supergravity_omega}, $11$-dimensional supergravity in the $\Omega$-background is obtained by compactifying the theory on $X_\eps$ down to $Y_\eps$, while including all KK modes.   Then, the theory on $Y_\eps$ includes an $S^1$ gauge field. We discard the ghost for constant gauge transformations.  If we do this, the spinor $\Psi(\eps,\alpha)$ on $Y_\eps$ squares to zero, and we can consider the supergravity theory on $Y_\eps$ in the background where the constant ghost is set to $\Psi(\eps,\alpha)$.  

Because we have included all KK modes in this construction,  we have lost very little information about the full $11$-dimensional theory.  On the other hand, this makes the construction hard to understand.  We will state a theorem which will allow us to gain some understanding of the dependence of the construction on the $M$-theory circle.  To state the theorem, we need to introduce a slightly different supergravity solution than the one we have considered so far.

Let $A_{G2}$ be the $G2$ $3$-form on $M$.  Then, $M \times Z$ equipped with the $3$-form $-\eps C + A_{G2}$ is again a solution of the supergravity equations of motion, simply because $A_{G2}$ is closed.  The spinor $\Psi(\eps,\alpha)$ remains a generalized Killing spinor, as the generalized Killing spinor equations only depend on the $4$-form which is the field-strength of the $C$-field.     

In the appendix, we show that a certain family of supergravity solutions where we scale the metric on $M$ is exact for the supercharge we are using. Let us write down the family explicitly. 

Let $M_r$ be the manifold $M$ but with metric $r^2 g_M$.  We can write down the $C$-field and spinor we described above using $M_r \times Z$.  If we identify spinors on $M_r$ with those on $M_{r = 1}$ in a natural way (see appendix \ref{appendix:sugra} for the precise conventions), we can write down explicitly a three parameter family of supergravity solutions, with parameters $\eps,\alpha,r$.  The metric, spinor and $C$-field are given by
\begin{align*} 
g &= r^2 g_M + g_Z\\
C &= -\eps r^2 \alpha^{-2} V^{\flat} \wedge \omega_Z^{0,2} + c r^3 A_{G2}\\
\Psi &=  \alpha r^{-1/2}  \psi_M \otimes \psi_Z + \alpha^{-1} r^{1/2}\eps (V \cdot \psi_M) \otimes (\omega^{0,2}_Z \cdot \psi_Z)
\end{align*}
The constant $c$ accompanying the $G2$ $3$-form is some non-zero complex number. The normalization of the spinors is such that in this family, the spinor $ r^{-1/2}  \psi_M \otimes \psi_Z $ has norm $1$ in the metric $r^2 g_M + g_Z$.
  
The theorem, proved in the appendix, is the following.
\begin{theorem} 
In this $3$-parameter family, the action of the vector field $r \dpa{r} + \tfrac{1}{2} \alpha \dpa{\alpha}$ is exact for the supersymmetry $\Psi$ we're using.    
\label{theorem_exactness}
\end{theorem}

In particular, if we set $\alpha = r^{1/2}$, then the family is independent of $r$ once we twist. This will allow us to exchange the dependence of the twisted theory on the radius of the $M$-theory circle with its dependence on $\eps$ and on the $G2$ $3$-form $A_{G2}$.   

Unfortunately, I was unable to calculate the value of the constant $c$, except to determine that it must  be purely imaginary. However, for our purposes, the $G2$ three-form will play almost no role.

\section{The $5$-dimensional gauge theory in detail}
\label{section_5d_gauge_theory}

The main result of this paper is that $M$-theory compactified in the $\Omega$-background on $TN_k \times \R^2$ is a $5$-dimensional non-commutative gauge theory.  I wrote down the action functional for this gauge theory in the introduction. Let me now discuss the theory in more detail.

Let's describe the theory first in the traditional formulation, with fields, an action functional, and gauge  symmetry. Then we will explain how to rewrite it in the BV formalism.   We will start by describing the theory on $\R \times \C^2$, and then explain how to write it on $\R \times X$ for a class of complex surfaces $X$.

The theory is a non-commutative gauge theory, meaning that the complex plane $\C^2$ is non-commutative.  If $\Oo(\C^2)$ denotes the ring of holomorphic functions on $\C^2$, recall that the Moyal product on $\Oo(\C^2)$ is a non-commutative product of the form
$$
f \ast_\eps g = f g + \eps  \tfrac{1}{2} \varepsilon_{ij} \dpa{z_i} f \dpa{z_j} g + \eps^2 \tfrac{1}{2^2 \cdot 2!} \varepsilon_{i_1 j_1} \eps_{i_2 j_2} \left( \dpa{z_{i_1}} \dpa{z_{i_2}} f\right) \left( \dpa{z_{j_1}} \dpa{z_{j_2}} g \right) + \dots
$$
where $\eps$ is a formal parameter, $\varepsilon_{ij}$ is the alternating symbol, and we have used the summation convention.    The coefficient of $\eps^n$ in the expansion is 
$$
  \frac{1}{2^n n!} \sum \varepsilon_{i_1j_1} \dots \varepsilon_{i_n j_n} \left( \dpa{z_{i_1}} \dots \dpa{z_{i_n}} f \right) \wedge \left( \dpa{z_{j_1}} \dots \dpa{z_{j_n}} g \right)  
$$
where the sum is taken over indices $i_1,\dots,i_n,j_1,\dots,j_n \in\{1,2\}$. 
This product can be extended to a product on the Dolbeault complex $\Omega^{0,\ast}(\C^2)$, by the same formula except that product of holomorphic functions is replaced by wedge product of Dolbeault forms.  The $\dbar$ operator on $\Omega^{0,\ast}(\C^2)$ is a derivation for the Moyal product, so that $\Omega^{0,\ast}(\C^2)[[\eps]]$ becomes a non-commutative differential graded algebra.

The fundamental field of our $5$-dimensional theory is a partial connection
$$
A \in \left(  \Omega^1(\R \times \C^2) / \left(\d z_1 \Omega^0 \oplus \d z_2 \Omega^0 \right) \right) \otimes \gl_N. 
$$
Thus, $A$ has $3$ components,
$$
A = A_0 \d t + A_{1} \d \zbar_1 + A_{2} \d \zbar_2
$$
where $A_0$, $A_1$ and $A_2$ are $\gl_N$-valued smooth functions on $\R \times \C^2$.

The action functional is 
\begin{align*} 
S(A) &= \int_{\R \times \C^2} \d z_1 \d z_2 \op{Tr} \left\{ \tfrac{1}{2} A \d A + \tfrac{1}{3} A (A \ast_\eps A) \right\}  \\
     &=  \int_{\R \times \C^2} \d z_1 \d z_2 \op{Tr} \left\{ \tfrac{1}{2} A \ast_\eps \d A + \tfrac{1}{3} A \ast_\eps A \ast_\eps A) \right\} 
\end{align*}
where $\eps$ is a coupling constant, and is treated as a formal parameter.   The expressions on the first and second lines are equivalent because the difference between them is a total derivative. 

Note that if we try to treat $\eps$ as a non-zero number, we find a non-local action functional.  While each term in the expansion of $\eps$ is local, the coefficient of $\eps^n$ has $2n$ derivatives, so that summing up the coefficients of $\eps^n$ will lead to a non-local expression.  Also, note that this theory is very non-renormalizable, because the classical action functional contains arbitrarily high derivatives.  Nevertheless, we will prove in section \ref{section:quantization} that the theory can be quantized (using techniques developed in \cite{CosLi15}).

The Lie algebra of infinitesimal gauge transformations of the theory is the space $\Omega^0(\R \times \C^2) \otimes \gl_N$, but equipped with the Lie bracket
$$
[f,g] = f \ast_\eps g - g \ast_\eps f
$$
which is the commutator in the tensor product of $\Omega^0(\R \times \C^2)$ with it's $\eps$-dependent Moyal product with the associative algebra $\gl_N$.  An infinitesimal gauge transformation $f$ acts on a field $A$ by 
$$
A \mapsto A + \eps \left( \d f + [f,A] \right) 
$$
where, again, the commutator $[f,A]$ uses the Moyal product.  Here by $\d f$ we mean the image of the de Rham differential in $f$ in the quotient of $\Omega^1(\R \times \C^2)$ by the subspace spanned by $\d z_1$ and $\d z_2$.

Let us now explain how one can put this theory in the BV formalism.  Let us introduce the space
$$
\mc{A} = \Omega^\ast(\R \times \C^2) / \langle \d z_1, \d z_2 \rangle
$$
which is the quotient of the de Rham complex by the differential ideal generated by the $1$-forms $\d z_1$, $\d z_2$.  Thus, 
 $$
\mc{A} = \cinfty(\R \times \C^2) [ \d t, \d \zbar_1, \d \zbar_2] 
 $$
 where the parameters $\d t$, $\d \zbar_1$ and $\d \zbar_2$ are of degree $1$. 

 Evidently, $\mc{A}$ is a commutative algebra and the de Rham operator $\d_{dR}$ descends to a differential $\d_{\mc{A}}$ on $\mc{A}$.  Explicitly, 
 $$
\d_{\mc{A}} = \d t \dpa{t} + \sum  \d \zbar_i \dpa{\zbar_i}. 
 $$
 Further, the Moyal product on $\C^2$ gives a map
 \begin{align*} 
 \mc{A} \otimes \mc{A} &\to \mc{A}[[\eps]]\\
     \alpha \otimes \beta & \mapsto  \alpha \ast_\eps \beta
 \end{align*}
 defined by the same formula we gave earlier, but where product of holomorphic functions on $\C^2$ is replaced by the wedge-product in $\mc{A}$. This Moyal product makes $\mc{A}[[\eps]]$ into a differential-graded commutative algebra over $\C[[\eps]]$.

 Let $\mc{A}_c$ be the quotient of the space $\Omega^\ast_c(\R \times \C^2)$ of compactly supported forms by the ideal generated by $\d z_1$ and $\d z_2$.  Then, $\mc{A}_c[[\eps]]$ is equipped with an integration map
 \begin{align*} 
 \int : \mc{A}_c[[\eps]] &\to \C [[\eps]]\\ 
     \alpha & \mapsto  \int_{\R \times \C^2} \d z_1 \d z_2 \alpha.
 \end{align*}
 Note that $\int \d_{\mc{A}}\alpha = 0$ and 
 $$
\int \alpha \ast_\eps \beta = \pm \int \beta \ast_\eps \alpha
 $$
 for $\alpha,\beta \in \mc{A}$.  

 Thus, the space $\mc{A}_c$ has all the structure needed to define the kind of Chern-Simons action functional that appears in open-string field theory \cite{Wit86}. The space of fields for this action functional is $\mc{A}_c \otimes \gl_N[1]$, and the action functional is
 \begin{align*} 
 S(\alpha) & = \tfrac{1}{2} \int \alpha \ast_\eps \d \alpha + \tfrac{1}{3} \int \alpha \ast_\eps \alpha \ast_\eps \alpha\\
               &= \tfrac{1}{2} \int \alpha  \d \alpha + \tfrac{1}{3} \int \alpha ( \alpha \ast_\eps \alpha).
 \end{align*}
The first and second lines are the same because the difference between them is a total derivative.

Then, $\mc{A}_c \otimes \gl_N[1]$ is the space of fields of our $5$-dimensional gauge theory in the BV formalism.  The BV action functional is the functional $S$ above. The odd symplectic structure is given by the formula
$$
\ip{\alpha,\beta} = \int \alpha \beta
$$
for $\alpha,\beta \in \mc{A}_c \otimes \gl_N[1]$.

\subsection{The gauge theory on more general manifolds}
Let $X$ be a holomorphic symplectic complex surface, and suppose that $X$ is equipped with a $\C^\times$ action which scales the holomorphic symplectic form.  (Such an $X$ is called \emph{conical}).  For example, $X$ could be the cotangent bundle of a Riemann surface, or the resolution of an ADE singularity. In this section, we will explain how to put our theory on $\R \times X$.  
\begin{definition}
A $\ast$-product  on the sheaf $\Oo_X$ of holomorphic functions on $X$ is a map of sheaves $\Oo_X \otimes_{\C} \Oo_X \to \Oo_X[[\eps]]$ which satisfies the following properties.
\begin{enumerate} 
 \item It makes $\Oo_X[[\eps]]$ into a sheaf of associative algebras, quantizing the sheaf of Poisson algebras $\Oo_X$.
 \item The coefficient of $\eps^n$ in the product is given by a holomorphic bi-differential operator of finite order.
\item Suppose the $\C^\times$ action on $X$ scales the holomorphic symplectic form by some non-zero weight $k$. Let us give $\Oo_X[[\eps]]$ an action of $\C^\times$ by combining the given action on $X$ with the action on $\C[[\eps]]$ which scales $\eps$ with weight $k$.  We require that the associative product on $\Oo_X[[\eps]]$ must be compatible with this $\C^\times$-action.   
\end{enumerate}
\end{definition}
The third condition severely restricts the moduli of $\ast$-products on $\Oo_X$. If $X$ is the cotangent bundle of a Riemann surface or a resolution of an ADE singularity, then this condition implies that the space of $\ast$-products is finite-dimensional and is a torsor for $H^2(X)$ (meaning that non-canonically, the space of such $\ast$-products is isomorphic to $H^2(X)$). 

In this situation, we can define our $5$-dimensional gauge theory on $\R \times X$, just as before.  We let $\mc{A}^{X}$ denote the quotient of $\Omega^\ast(\R \times X)$ by the differential ideal generated by $\Omega^{1,0}(X)$, and $\mc{A}_c^{X}$ be the corresponding quotient of $\Omega^\ast_c(R \times X)$. Thus,
   $$
   \mc{A}^{X} = \Omega^\ast(\R) \what{\otimes} \Omega^{0,\ast}(X)
   $$
where $\what{\otimes}$ is the completed projective tensor product.

The fact that the $\ast$-product on $\Oo_X$ is implemented by holomorphic bi-differential operators means that it extends in a natural way to a $\ast$-product on $\mc{A}^{X}$, making $\mc{A}^{X}[[\eps]]$ into a differential graded associative algebra.  As before, there is an integration map on $\mc{A}_c^{X}$ defined by the formula
$$ 
\int \alpha = \int_{\R \times X} \omega \alpha 
$$
where $\omega$ is the holomorphic volume form on $X$. 

Thus, we can define a field theory where the fields are $\mc{A}_c^{X} \otimes \gl_N[1]$, and the action functional as before is
$$
S(\alpha)  = \tfrac{1}{2} \int \alpha \ast_\eps \d \alpha + \tfrac{1}{3} \int \alpha \ast_\eps \alpha \ast_\eps \alpha.
$$
This describes the theory on $X$ in the BV formalism.

\section{$M$-theory on Taub-NUT manifolds}
\label{section:Mtheory_TN}
We will now turn to applying our general construction of $\Omega$-background for $M$-theory to the special case when the $G2$ manifold is the product of a Taub-NUT manifold with flat space.

Let us first recall how the relationship between the Taub-NUT manifold and $D6$ branes in type IIA string theory.  There are several versions of the Taub-NUT manifold. The basic version is a complete Riemannian $4$-manifold with $SU(2)$ holonomy and which is equipped with a map to $\R^3$. Away from the origin in $\R^3$, the fibre of Taub-NUT is a circle, but at the origin the circle shrinks to a point.  There is a $U(1)$ action on Taub-NUT which preserves the map to $\R^3$ and rotates the circle fibres, and is holomorphic in each complex structure (and also preserves the holomorphic symplectic form in each complex structure).  

Multi-centered Taub-NUT is a variant of this construction which depends on a choice of points $x_1,\dots, x_k \in \R^3$.  The multi-centered Taub-NUT maps to $\R^3$ in such a way that the generic fibre is a circle, but that the circle shrinks to a point at each $x_i$.  Multi-centered Taub-NUT again has $SU(2)$ holonomy and a $U(1)$ action rotating the circle fibres. As we bring all the points $x_i \in \R^3$ together at the origin, the Taub-NUT manifold with centers $x_i$ acquires an $A_{k-1}$ singularity. 

Let $TN_{x_1,\dots,x_k}$ denote the Taub-NUT manifold with centers $x_i \in \R^3$, and let $TN_{k-1}$ denote the specialization of this when all $x_i = 0$. This is the Taub-NUT manifold with an $A_{k-1}$ singularity at the origin.  It is a standard statement that $M$-theory on $TN_{x_1,\dots,x_k} \times \R^7$ becomes type IIA string theory on $\R^3 \times \R^7$ with $k$ $D6$ branes located on the submanifolds $x_i \times \R^7$.   As we bring the $x_i$ together at the origin, the Taub-NUT manifold acquires a singularity and we get a stack of $k$ $D6$ branes at the origin in $\R^3$.   The fact that the theory on a stack of $k$ $D6$ branes is maximally supersymmetric gauge theory with group $U(k)$ will provide the connection with the $5$-dimensional gauge theory studied in this paper.  

Now, $TN_{x_1,\dots,x_k} \times \R^3$ has the structure of a $G2$ manifold in a natural way, by choosing an embedding of $SU(2)$ inside $G2$.  Thus, $TN_{x_1,\dots,x_k} \times \R^3 \times \C^2$ is the product of a $G2$ manifold with a hyper-K\"ahler manifold.  The isometry of $TN_{x_1,\dots,x_k}$ which rotates the circle fibre places us in the situation discussed in the theorem. Thus, we can introduce a flux $-\eps F$ and a family $\Psi(\eps)$ of spinors such that $TN_{x_1,\dots,x_k} \times \R^3 \times \C^2$ with the flux $-\eps F$ satisfies the equations of motion of $11$-dimensional supergravity, and $\Psi(\eps)$ is a generalized Killing spinor whose square is $\eps V$ where $V$ is the Killing vector field on $TN_{x_1,\dots,x_k}$ which rotates the circle fibres.   The same statements hold if we use the singular space  $TN_{k-1}$ instead.
 
Thus, we are in the situation we discussed above, and can consider putting $M$-theory in the $\Omega$-background on the Taub-NUT manifold.

\begin{proposition}
Consider $M$-theory on $TN_{k-1} \times \R^3 \times \C^2$, with the $C$-field $-\eps V^{\flat} \d \zbar_1 \d \zbar_2$,  where as above $V$ is the vector field generating the $S^1$ action on $TN_{k-1}$ and $V^{\flat}$ is the $1$-form dual to $V$ using the metric.  Suppose the radius at $\infty$ of the Taub-NUT circle is $r$. 
    
Upon reduction to $10$ dimensions, this theory becomes type IIA string theory on $\R^3 \times \R^3 \times \C^2$ with $k$ $D6$ branes supported on $0 \times \R^3 \times \C^2$ and with a constant $B$-field $-r^2 \eps \d \zbar_1 \d \zbar_2$.  
\end{proposition}

\begin{proof}
The content of the statement is that our $3$-form flux in $M$-theory on the Taub-NUT reduces to the constant $B$-field.  To do this, we will use the discussion in \cite{Wit95a} for conventions regarding the reduction of $11$-dimensional supergravity to $10$ dimensions.  We will do the calculation when the radius at $\infty$ is $1$.  Witten in this paper explains that we should think of the reduction of $M$-theory on to $10$ dimensions along a circle as follows.  Suppose the $11$-dimensional metric is of the form
$$
G_{mn} \d x^m \d x^n + e^{2 \gamma} ( \d \theta - A_m \d x^m ) ^2 
$$
where $\theta$ is a coordinate on the $M$-theory circle, and $\gamma$ is a $10$-dimensional scalar which will be related to the dilaton.  We define a $10$-dimensional metric by $g_{mn} = e^{\gamma} G_{mn}$. Reducing to $10$ dimensions produces type IIA with metric $g_{mn}$ and with dilaton $\phi$ satisfying
$$
e^{-2 \phi} = e^{-3 \gamma}. 
$$
The dilaton is related to the string coupling constant by $\lambda = e^{\phi}$.  The radius of the $M$-theory circle is $e^{\gamma} = \lambda^{2/3}$. The masses of the Kaluza-Klein modes are, when measred in this metric $g$, of order $\lambda^{-1}$.

Let us now recall the metric on Taub-NUT, and see how to use Witten's discussion to produce the correct scaling on various factors. Let $x_i$ denote coordinates on the $\R^3$ which is the base of the Taub-NUT circle fibration, and let $\theta$ be a fibre coordinate.  Note that since $TN_k \to \R^3$ is a topologically non-trivial $S^1$-bundle away from the origin, the coordinate $\theta$ is only locally defined.   Let $y_i$ be coordinates on the other $\R^3$, and $z_i$ those on $\C^2$.

Define a function of $x \in \R^3$ by 
$$
f = 1 + \tfrac{k} { \abs{x}} .
$$
Then, the metric on $TN_k \times \R^3 \times \C^2$ metric is of the form
$$
\sum \d y_i^2 + \sum \d z_i \d \zbar_i +  f \sum \d x_i^2 + f^{-1} \left( \d \theta - \sum A_i \d x_i \right)^2 
$$
where $A$ is a (locally-defined) one-form on $\R^3$ satisfying the Bogomolny equation
$$
-\ast \d A = \d f. 
$$
The more invariant way to think of $A$ is that, away from $0$, it is a connection on the principal $S^1$ bundle defined by the Taub-NUT over $\R^3 \setminus 0$. Then, $\d A$ is the curvature of this connection.

An explicit calculation tells us that
$$
\d A = k \abs{x}^{-3} \left(x_1 \d x_2 \d x_3 - x_2 \d x_1 \d x_3 + x_3 \d x_1 \d x_2 \right).
$$
Note that $\int_{S^2} \d A =  4k\pi $.
 
The $C$-field for our $\Omega$-background is obtained (up to a factor of $\eps$)  by contracting the metric with the vector field $\dpa{\theta}$, yielding 
$$
 -\eps   f^{-1} \d \theta \d \zbar_1 \d \zbar_2 + \eps f^{-1} A \d\zbar_1 \d \zbar_2.  
$$

Now let's see what we find in $10$ dimensions, using Witten's conventions.  We find that the $10$-dimensional metric is
$$
   f^{-1/2}\sum \d y_i^2 + f^{-1/2}  \sum \d z_i \d \zbar_i +  f^{1/2}  \sum \d x_i^2. 
$$
The $B$-field is
$$
B = - \eps f^{-1} \d \zbar_1 \d \zbar_2.
$$
The RR $3$-form is 
$$
C_{3} = \eps f^{-1} A \d\zbar_1 \d \zbar_2. 
$$
(where recall that $A$ is only locally defined; the $4$-form $\d C_3$ is globally well-defined). Finally, the RR $1$-form is $A$.

The metric, dilaton and RR $1$-form are those induced by $k$ D6 branes in flat space.  I claim that the $B$-field and $C_3$-field are those described by $k$ $D6$ branes in the presence of the constant $B$-field $-\eps \d \zbar_1 \d \zbar_2$.  

We will verify this by showing that, asymptotically at infinity, we find a flat type IIA background with constant $B$-field. Far away from the $D6$ branes, the field $-\eps f^{-1} \d \zbar_1 \d \zbar_2$ approaches $-\eps \d \zbar_1 \d \zbar_2$, and similarly the field $C_3$ approaches $\eps A \d \zbar_1 \d \zbar_2$.  The integral of $\d A$ the $2$-sphere $S^2_r$ at distance $r$ from the $D6$ branes is $4\pi$, whereas the Riemannian volume of this sphere tends to infinity as $r$ goes to infinity. Thus, $\norm{\d A}^2$ tends to zero at infinity.  

This tells us that, as desired, asymptotically at infinity we are in type IIA with constant $B$-field.  Since the fields other than the $B$ and $C_3$ field describe type IIA with $k$ D6 branes, the $C_3$-field and the non-constant $B$-field must be obtained from the back-reaction of the $D6$ brane on a constant $B$-field.

As a consistency check, note that a $D6$ brane in a constant $B$-field is a source not just for the RR $1$-form field, but also for the $3$-form field.  We should have the identities
\begin{align*} 
\d F_2&= \delta_{D6} \\
\d F_4 &= \delta_{D6} \wedge B 
\end{align*}
where $\delta_{D6}$ is the $\delta$-current for the location of the $D6$ brane; it is a $3$-form with distributional coefficients.   This equation is solved if we take the RR $1$-form $A$ to be a monopole in the directions transverse to the $D6$ brane, and $C_3 = B \wedge A$.  This is what we find above (except that the $B$-field is the not the constant one but incorporates the back-reaction from the $D6$ brane).    
 
Finally, suppose the radius at $\infty$ in $TN_k$ is $r$. Then, in a coordinate $\theta$ on this circle at $\infty$, the metric is $r^2 \d \theta^2$ and the vector field $V$ is $\dpa{\theta}$.  Thus, the $1$-form $V^{\flat}$ is $r^2 \d \theta$, whose integral is $r^2$. This is why we find the constant $B$-field $-\eps r^2 \d \zbar_1 \d \zbar_2$ in this context.

\end{proof}

\subsection{}
Let us now analyze what the supercharge we use in $M$-theory  looks like in terms of type IIA string theory.    We will use the conjectural description of twists of type IIA string theory in terms of the topological $A$- and $B$-models developed in \cite{CosLi16}.  Although, at the level of the closed string theory, our description is still conjectural, we proved the open-string version of the conjecture by analyzing twists of theories living on $D$-branes.   

Let us first discuss the case that $\eps = 0$. In this case we are twisting $11$-dimensional supergravity and then reducing to $10$ dimensions along the circle fibres of the Taub-NUT.  We need to describe the supercharge we find that twists type IIA supergravity.   We will work near infinity, where the Taub-NUT manifold is asymptotically flat.

There is a decomposition
$$
S_{11d} = S_{7d} \otimes S_{4d}
$$
of $11$-dimensional spinors into a tensor product of an $8$-dimensional spin representation of $\op{Spin}(7)$ and a $4$-dimensional spin representation of $\op{Spin}(4)$.  Let us write $\op{Spin}(4) = SU(2)_+ \times SU(2)_-$, where $S_+$ and $S_-$, the positive and negative chirality spin representations of $\op{Spin}(4)$ are the fundamental representations of $SU(2)_+$ and $SU(2)_-$ respetive.y

The $11$-dimensional spinor is of the form
$$
\Psi = \Psi_{G2} \otimes \Psi_{SU(2)_+}
$$
where $\Psi_{G2} \in S_{7d}$ is invariant under $G2$, and $\Psi_{SU(2)_+} \in S_{4d}$ is one of the two spinors invariant under $SU(2)_+$. 

To describe what happens in $10$ dimensions, let us decompose the spinors as a representation of $\op{Spin}(6) \times \op{Spin}(4)$. We have
$$
S_{10d} = S_{6d} \otimes S_{4d}
$$
where $S_{6d}$ is an $8$-dimensional spin representation of $\op{Spin}(6)$ which is a sum of the chiral and antichiral $4$-dimensional representations.

We can use the isomorphism $\op{Spin}(6) \iso SU(4)$. Let $V$ denote the fundamental representation of $SU(4)$. Then, 
$$S_{6d} = S_{6d}^+ \oplus S_{6d}^- = V \oplus V^\ast.$$ 

The $G2$ spinor $\Psi_{G2}$ is a sum
$$
\Psi_{G2} =  \Psi_{SU(3)}^+ +  \Psi_{SU(3)}^-
$$
where $\Psi_{SU(3)}^{\pm} \in S_{6d}^{\pm}$ are the two spinors invariant under $SU(3) \subset G2$.  Further, $\ip{\Psi_{SU(3)}^+, \Psi_{SU(3)}^-} = 1$. These properties fix $\Psi_{G2}$ up to a normalization factor.

Thus, we find that the $10$-dimensional spinor obtained by reduction from $11$ dimensions is 
$$
\Psi_{IIA}= \Psi_{SU(3)}^+ \otimes \Psi_{SU(2)_+} + \Psi_{SU(3)}^- \otimes \Psi_{SU(2)_+}.  
$$

\subsection{}

Next, let us analyze what happens when we turn on $\eps$.  Recall that when we turn on $\eps$, we also introduce a $B$-field $-\eps \d \zbar_1 \d \zbar_2$.  

A  constant $B$-field can be viewed as an element of the Lie algebra $\mf{so}(10,\C)$. As such, it gives rise to an endomorphism of the space of spinors, which can be realized concretely by Clifford multiplication.   In string theory in the presence of a constant $B$-field, this action plays an important role.  A spinor $\psi$ preserves an object of the theory (such as a $D$-brane) in the presence of a constant $B$-field if $e^{-B} \psi$ preserves the same object when we don't have a $B$-field.     

Let $\Psi_{IIA}(\eps)$ denote the spinor we get in $10$ dimensions once we turn on $\eps$.  Then, 
$$
e^{\eps \d \zbar_1 \d \zbar_2} \Psi_{IIA}(\eps)
$$
must preserve the $D6$ brane, since $\Psi_{IIA}(\eps)$ preserves the $D6$ brane in the presence of a $B$-field. In fact, we find the following.   
\begin{lemma}
$$
e^{\eps \d \zbar_1 \d \zbar_2} \Psi_{IIA}(\eps) = \Psi_{IIA}(\eps = 0). 
$$
\end{lemma} 
\begin{proof}
By the definition of the $11$-dimensional spinor $\Psi(\eps)$, we have
\begin{equation}
\Psi^{IIA}(\eps) =  \Psi^{IIA}(0) +  \eps \left(\dpa{\theta} \cdot \Psi_{G2}\right) \otimes (\omega^{0,2}  \cdot \Psi_{SU(2)_+}) \tag{$\dagger$}. 
\end{equation}
Here $\cdot$ indicates Clifford multiplication and $\dpa{\theta}$ is the vector field which rotates the $11$-dimensional circle.  

We have seen that
$$
\Psi_{G2} = \Psi_{SU(3)}^+ + \Psi_{SU(3)}^-
$$
so that $\Psi_{G2}$ is a sum of the two $6$-dimensional spinors invariant under $SU(3)$.  Since $SU(3) \subset G2$ fixes the vector $\dpa{\theta}$, the spinor $\dpa{\theta}\cdot \Psi_{G2}$ must also be a sum of the two $SU(3)$ invariant spinors, and it must be linearly independent from $\Psi_{G2}$.

There is a $4$-dimensional space of $SU(3) \times SU(2)_+$ invariant spinors for type $IIA$, and there is a two-dimensional subspace  of this that preserves the $D6$ brane.  This two-dimensional subspace is invariant under the other copy $SU(2)_-$ of $SU(2)$ inside of $\op{Spin}(4)$.  The two  $SU(2)_+$-invariant spinors form a copy of the fundamental representation of $SU(2)_-$.  Thus, any two-dimensional subspace of the $SU(3) \times SU(2)_+$ invariant spinors which is a representation of $SU(2)_-$ must consist of the tensor product of a single $SU(3)$ invariant spinor with both $SU(2) = SU(2)_+$ invariant spinors.  The single $SU(3)$ invariant spinor must be $\Psi_{G2}$.  It follows that
$$
\left( \dpa{\theta} \cdot \Psi_{G2} \right) \otimes (\d \zbar_1 \d \zbar_2\cdot  \Psi_{SU(2)} )
$$ 
can not preserve the $D6$ brane. 

However, we know that our spinor $\Psi_{IIA}(\eps)$ preserves the $D6$ brane in the presence of the $B$-field $-\eps \d \zbar_1 \d \zbar_2$.  This means that $e^{\eps \d \zbar_1 \d \zbar_2} \Psi_{IIA}(\eps)$ preserves the $D6$ brane. Further, since $\d \zbar_1 \d \zbar_2$ is nilpotent, $e^{\eps \d \zbar_1 \d \zbar_2}$ only contains constant and linear terms in $\eps$.  

Since the coefficient of $\eps$ in $\Psi_{IIA}(\eps)$ does not preserve the $D6$ brane, we deduce that $e^{\eps \d \zbar_1 \d \zbar_2} \Psi_{IIA}(\eps)$ is independent of $\eps$, so that 
$$
e^{\eps \d \zbar_1 \d \zbar_2} \Psi_{IIA}(\eps) = \Psi_{IIA}(\eps = 0) 
$$
as desired. 
\end{proof}

\subsection{The twist of the theory on a $D6$ brane} 
The twist of type IIA superstring theory by the supercharge $\Psi_{IIA}(\eps = 0)$ was analyzed in \cite{CosLi16}.  There, we showed that this supercharge has the feature that translation in the $6$ directions rotated by $SU(3)$ is $Q$-exact, as are the two  translations in the directions $\dpa{\zbar_i}$ in the complex plane $\C^2$.  Writing $\R^{10} = \R^6 \times \C^2$, we find the supergravity theory is topological on $\R^6$ and holomorphic on $\C^2$. A conjectural description of this supergravity theory was presented in \cite{CosLi16}: it was argued that this supergravity theory arises from the topological string theory which is the topological $A$-model on $\R^6$ and the topological $B$-model on $\C^2$.

An analysis of $D$-branes was also performed in \cite{CosLi16}. It was shown that the theory living on a $D$-brane, after twisting, is equivalent to the theory living on a brane in the appropriate topological string theory. For the $D6$ brane that arises from putting $M$-theory on a Taub-NUT, we find a brane on $\R^3 \times \R^3 \times \C^2$ which lives on $0 \times \R^3 \times \C^2$.  It is the product of an $A$-brane on the Lagrangian submanifold $0 \times \R^3 \subset \R^3$, and a space-filling $B$-brane on $\C^2$. 

Following \cite{CosLi16}, we can write down the fields and action functional of this theory in the BV formalism as follows.  It is a theory of Chern-Simons type associated to a differential graded algebra with a trace (as in Witten's work \cite{Wit86} on open string theory). The dg algebra is
$$
\Omega^\ast(\R^3) \what{\otimes} \Omega^{0,\ast}(\C^2) \otimes \mf{gl}_N
$$ 
where $\what{\otimes}$ indicates an appropriate completed tensor product. The differential is the sum of the de Rham operator on $\Omega^\ast(\R^3)$ and the Dolbeault operator on $\Omega^{0,\ast}(\C^2)$.  The trace is given by 
$$
\op{Tr} (\alpha) = \int_{\R^3 \times \C^2} \d z_1 \d z_2 \op{Tr}_{\mf{gl}_N} \alpha.
$$
Of course, as in our previous discussion of the five-dimensional gauge theory, for this to converge we need $\alpha$ to have compact support. 

We are interested in the deformation of this where we turn on the parameter $\eps$.  This has the effect, in twisted type IIA, of introduce a $B$-field $-\eps \d \zbar_1 \d \zbar_2$.  

This deformation of twisted type IIA will correspond to a deformation of the topological string on $\R^6 \times \C^2$.  Let us give an argument for what this must be.  According to Seiberg and Witten \cite{SeiWit99}, the theory living on a $D$-brane in the presence of a constant $B$-field is a non-commutative field theory.  We thus expect that this deformation has the effect of turning $B$-branes living on $\C^2$ into $B$-branes living on a non-commutative $\C^2$.  In the $B$-model this is realized by introducing a constant polyvector field $-\eps \dpa{z_1} \dpa{z_2}$ into the fields of BCOV theory, which is the space-time ``gravity'' theory associated to the topological $B$-model.   
  
For the theory on the $D6$ brane in the twisted theory, introducing the constant polyvector field $-\eps \dpa{z_1} \dpa{z_2}$ has the effect of replacing the differential-graded algebra $\Omega^\ast(\R^3) \what{\otimes} \Omega^{0,\ast}(\C^2) \otimes \mf{gl}_k$ be a non-commutative version.  In the non-commutative version, the product is replaced by the Moyal product
$$
\alpha \ast_{-\eps} \beta = \sum_{n \ge 0} (-\eps)^n \frac{1}{2^n n!} \varepsilon_{i_1j_1} \dots \varepsilon_{i_n j_n} \left( \dpa{z_{i_1}} \dots \dpa{z_{i_n}} \alpha \right) \wedge \left( \dpa{z_{j_1}} \dots \dpa{z_{j_n}} \beta \right) 
$$
where $\varepsilon_{ij}$ is the alternating symbol, and $\alpha,\beta $ indicate elements of $\Omega^\ast(\R^3) \what{\otimes} \Omega^{0,\ast}(\C^2)$.   

This deformation has a rather drastic effect on the closed string sector of twisted type IIA.  Before turning on the $B$-field, this twist of type IIA is conjecturally described by the $A$-model on $\R^6$ and the $B$-model on $\C^2$.  However, as is well-known, the $B$-model on a non-commutative space is the same as the $A$-model equipped with a $B$-field.  If we throw away instanton effects in the $A$-model (as we do in this paper, since we work in the supergravity limit), the $B$-field has no effect.  Thus, after turning on this $B$-field, the twist of type IIA theory becomes the topological $A$-model on $\R^{10}$.  The closed-string field theory associated to this has fields the de Rham complex $\Omega^\ast(\R^{10})$, and is both topological and free.    The fact that this theory is so trivial means that we can essentially disregard it from our analysis.

\section{Dependence on the $M$-theory radius}
So far, we have seen how $M$-theory in the twisted $\Omega$-background can be described, upon reduction to $10$ dimensions, in terms of type IIA  string theory in the presence of a $B$-field, some $D6$ branes, and a certain supercharge. Further, we have argued that when $\eps = 0$, the twist of type IIA can be described by a topological string on $\R^6 \times \C^2$ which is the $A$-model on $\R^6$ and the $B$-model on $\C^2$. The $D6$ branes live on $\R^3 \times \C^2$.  Turning on $\eps$ takes us to a non-commutative $B$-model on $\C^2$, which is equivalent to the $A$-model in the presence of a $B$-field.

We would like to understand in what regime this will give a good picture of what happens in $11$ dimensions. The key thing to understand is to what extend everything depends on the radius of the $M$-theory circle.

Let $TN_k$ denote the Taub-NUT manifold with an $A_k$ singularity at the origin. We have already stated theorem \ref{theorem_exact} which shows that, with an appropriate choice of scale for the supercharge and the $C$-field, we can construct a family of supergravity backgrounds which is independent of the radius. Let us apply this construction when the G2 manifold is $TN_k \times \R^3$ and the manifold of $SU(2)$ holonomy is $\C^2$. 

Explicitly, we have a family of spinors, metric, and $C$-fields given by
\begin{align*} 
g &= r^2 g_{TN_k} + r^2 g_{\R^3} + g_{\C^2} \\
C &= -\eps r V^{\flat}\omega_Z^{0,2} + c r^3 A_{G2}\\
\Psi &=   \psi_M \otimes \psi_Z + \eps (V \cdot \psi_M) \otimes (\omega^{0,2}_Z \cdot \psi_Z)
\end{align*}
The normalization of the spinors is such that in this family, the spinor $ \psi_M \otimes \psi_Z $ has norm $r$.

Theorem \ref{theorem_exactness} shows that $r$-derivative of this family of supergravity backgrounds is exact with respect to the supercharge associated to the family of generalized Killing spinors $\Psi_r(\eps)$.    This tells us that after we pass to the $\Omega$-background, $M$-theory in this family of backgrounds is independent of $r$.

Let's describe what this family of $M$-theory backgrounds looks like when we pass to $IIA$.  As we have seen, the field $-r \eps C$ becomes the $B$-field $-\eps r \d \zbar_1 \d \zbar_2$.  

To describe the reduction of the $G2$ $3$-form, let us identify $\R^6$ with $\C^3$ and we coordinates $w_i$.  Then, the field $r^3A_{G2}$ gives rise to the $B$-field $r^3\sum \d w_i \d \wbar_i$ and the Ramond-Ramond $3$-form $r^3\op{Re} \d w_1 \d w_2 \d w_3$.

The metric on $\R^6 \times \C^2$ acquires an extra factor of $r$ if we work in units where $\alpha' =1$.  The metric is $r^3 g_{\R^6} + r g_{\C^2}$. 

Finally, the string coupling constant is $r^{3/2}$. 

We also, of course, have $k$ $D6$ branes which wrap $\C^2$ and are located on the locus in $\R^6$ where the $w_i$ are real.

 The statement which follows formally from our supersymmetry analysis is that type IIA string theory in this background is, after twisting, independent of $r$.

\subsection{}
Let us summarize our conclusions about this so far in this section:
\begin{enumerate} 
\item For $r$ small, the theory has a description as a twist of type IIA in the presence of $k$ $D6$ branes and a $B$-field. 
\item Once we have turned on the $B$-field, the theory on the $D6$ branes is a non-commutative $GL(k)$ gauge theory on $\R^3 \times \C^2$, where the coordinates $z_1,z_2$ on $\C^2$  are non-commutative. 
\item In the presence of the $B$-field, the contribution of the twisted closed sting theory is irrelevant because almost all closed string fields become cohomologically exact.
\item If we introduce an additional term in the $B$-field and a Ramond-Ramond $3$-form to the type IIA background, we can make everything independent of $r$.
\end{enumerate}
From the point of view of the $D6$ brane, the additional terms in the type IIA  background coming from the $G2$ $3$-form do not play any role\footnote{They will not play any role when the field theory on the $D6$ brane is treated in perturbation theory. In a more general  situation, these additional terms may play a role when we include stringy effects: that is, they will effect the weight of counts of BPS $D2$ branes and fundamental strings.  BPS $D2$ branes and fundamental strings come from special Lagrangians and  holomorphic curves in the Calabi-Yau $\R^6$ with coordinates $w_i$.  Since there are no compact special Lagrangians,  or holomorphic curves bounding our Lagrangian $\op{Im} w_i = 0$, these additional terms do not contribute at all. }.   We will thus ignore them. 

As we have seen, the $7$-dimensional gauge theory can be written as an open-string field theory built from the dg algebra
$$
\mc{A}_N = \Omega^\ast(\R^3) \what{\otimes} \Omega^{0,\ast}(\C^2) \otimes \mf{gl}_N
$$
equipped with the Moyal product discussed in section \ref{section:Mtheory_TN}, with parameter of non-commutativity $c$.    Let us denote the coupling of the $7$-dimensional gauge theory by $\hbar$, so that with this coupling constant the action functional is 
$$
\frac{1}{\hbar} \int_{\R^3 \times \C^2} \d z_1 \d z_2 \op{Tr} \left\{ \tfrac{1}{2} \alpha \ast_{c} \d \alpha + \op{Tr} \tfrac{1}{3}\alpha \ast_{c} \alpha \ast_{c} \alpha \right\}.  
$$ 
Since everything is independent of the $M$-theory radius $r$, the only possibility is that $c = -\eps$ and that $\hbar$ is the $11$-dimensional Planck length. We will ignore this Planck length parameter in what follows.  

To summarize, we have shown that $11$-dimensional supergravity in an $\Omega$-background where the Taub-NUT fibre  is rotated can be described in terms of a $7$-dimensional non-commutative gauge  theory which is very close to the $5$-dimensional gauge theory we are interested in.

\section{From $7$ to $5$ dimensions}

So far, we have argued that a twisted version of $M$-theory on a Taub-NUT $\Omega$-background leads to a non-commutative $7$-dimensional gauge theory which is very similar to the $5$-dimensional gauge theory we introduced in section \ref{section_5d_gauge_theory}. Now we will see how performing a further $\Omega$-background construction will lead to precisely the $5$-dimensional gauge theory we are interested in. 

Let $TN_k$ denote the Taub-NUT manifold with an $A_k$ singularity at the origin. Consider $M$-theory on $TN_k \times \R^2 \times \R \times \C^2$. The factor $TN_k \times \R^2$ has the structure of a Calabi-Yau $3$-fold, using the obvious embedding of $SU(2)$ into $SU(3)$.  We want to perform an $\Omega$-background construction using a rotation which acts on both $TN_k$ and on $\R^2$.  Recall that the isometries of $TN_k$ include the $U(1)$ we considered before, which is tri-holomorphic, and an $SO(3)$ which rotates the $3$ different complex structures present on a hyper-K\"ahler manifold.  We can choose a $U(1)$ inside the $SO(3)$ which fixes one of the three complex structures and which scales the corresponding holomorphic symplectic form.     We can make this $U(1)$ also rotate $\R^2$ so that it scales the one-form $\d (y_1 + i y_2)$ with weight opposite to the way it scales the holomorphic symplectic form on $TN_k$.  In this way, the $U(1)$ action on $TN_k \times \R^2$ is by an isometry which preserves the $SU(3)$ structure.

Together with the $U(1)$ action which rotates $TN_k$ by a tri-holomorphic isometry, we find we have a $U(1) \times U(1)$ action on $TN_k \times \R^2$ which preserves the $SU(3)$ structure.  As before, $V$ let refer to the vector field on $TN_k$ generating the tri-holomorphic $U(1)$ action, and let $\til{V}$ denote the vector field generating the other $U(1)$ action.  Let $\til{V}^{\flat}$ denote the $1$-form dual to $\til{V}$, and define a $3$-form by$$
\til{C} = \til{V}^{\flat} \otimes \omega^{0,2}_{\C^2} 
$$
where $\omega^{0,2}_{\C^2} = \d \zbar_1 \d \zbar_2$.
 
As in theorem \ref{theorem_supergravity_solution}, let $\psi_{TN_k \times \R^2 \times \R}$ denote the spinor associated to the $G2$ structure on $TN_k \times \R^2 \times \R$, and let $\psi_{\C^2}$ denote the spinor associated to the $SU(2)$ structure on $\C^2$.  Define a family of spinors
$$
\Psi(\eps,\delta) = \psi_{TN_k \times \R^3} \otimes \psi_{\C^2} + \eps(V \cdot \psi_{TN_k \times \R^3}) \otimes (\omega^{0,2}_{\C^2}\cdot \psi_{\C^2}) + \delta(\til{V} \cdot \psi_{TN_k \times \R^3}) \otimes (\omega^{0,2}_{\C^2} \cdot \psi_{\C^2}). 
$$ 
Then the family of spinors $\Psi(\eps,\delta)$ are generalized Killing spinors for $11$-dimensional supergravity with the $3$-form $-\eps C - \delta \til{C}$, and 
$$
\Psi(\eps,\delta)^2 = \eps V + \delta \til{V}. 
$$

We can then put $M$-theory in the $\Omega$ background associated to sum $\eps V + \delta\til{V}$ of vector fields, just as before. 

Let us think of $\delta$ as being small compared to $\eps$.  If we do this, then we can imagine first doing the $\Omega$-background construction with respect to the vector field $V$ which rotates $TN_k$, and then doing a further $\Omega$-background construction. The first $\Omega$-background construction yields, as we have seen, a $D6$ brane in type IIA in the presence of a constant $B$-field.  It is natural to expect that the second $\Omega$-background construction will have the effect, on the locus of the $D6$ brane, of putting the $7$-dimensional gauge theory in an $\Omega$-background. We will not argue this point in detail. Instead, we will show that putting this $7$-dimensional gauge theory in an $\Omega$-background will yield the $5$-dimensional non-commutative gauge theory we discussed earlier.

\subsection{$7$-dimensional gauge theory in an $\Omega$-background}
In this section, we will analyze the $7$-dimensional non-commutative gauge theory in an $\Omega$-background, and argue that when we do this we find the $5$-dimensional non-commutative gauge theory we introduced in section \ref{section_5d_gauge_theory}.

Consider $7$-dimensional supersymmetric pure gauge theory with gauge group $G$ and corresponding Lie algebra $\mf{g}$.   Let us consider the theory on $\R^3 \times X$, where $X$ is a manifold with $SU(2)$ holonomy.   If we take the size of $X$ to be small, then the theory has an effective description as a $3$-dimensional $N=4$ $\sigma$-model with target the hyper-K\"ahler manifold of $G$-instantons on $X$.  

 We can give a similar description of the full $7$-dimensional theory on $\R^3 \times X$, as a $3$-dimensional $N=4$ theory with an infinite-dimensional target.  Recall that to describe a $3$-dimensional $N=4$ gauge theory, we need to give a hyper-K\"ahler vector space $V$ and a group $G$ which acts on $V$ by hyper-K\"ahler isometries.  The Higgs branch of the theory is then the hyper-K\"ahler quotient of $V$ by $G$.  

If $X$ is a $4$-manifold with $SU(2)$ holonomy, then the space $\Omega^1(X,\g)$ of connections on the trivial $G$-bundle on $X$ is an infinite-dimensional hyper-K\"ahler vector space.  The group $\cinfty(X,G)$ of smooth maps from $X$ to $G$ acts on $\Omega^1(X,\g)$ by gauge transformations. The three components of the hyper-K\"ahler moment map combine into the map
\begin{align*} 
 \Omega^1(X,\mf{g}) & \mapsto \Omega^2_+(X,\mf{g}) \\
A & \mapsto F(A)_+
\end{align*}
sending a connection to the self-dual part of its curvature.  Thus, the hyper-K\"ahler quotient of $\Omega^1(X,\mf{g})$ by $\cinfty(X,G)$ is the space of (topologically trivial) instantons on $X$. (The construction can be modified in an evident way to describe instantons of a given topological type).

There is a $3$d $N=4$ theory built from the hyper-K\"ahler vector space $\Omega^1(X,\mf{g})$ acted on by $\cinfty(X,G)$.  Since the fields of this $3$d theory are maps from $\R^3$ to the space of sections of a bundle on $X$, the result can be described as a theory on $\R^3 \times X$.  The resulting theory is the supersymmetric pure gauge theory on $\R^3 \times X$.  This description holds in the case that $X = \R^4$ as well (although one needs a modicum of care to deal with the fact that $\Omega^1(\R^4,\mf{g})$ is not strictly  hyper-K\"ahler because the integrals defining the hyper-K\"ahler forms may not converge).   A disadvantange of this presentation of $7$-dimensional gauge theory is that the $\op{Spin}(7)$-action is not evident.  However, this will not be important for us.

\subsection{}
As is well-known,  $3d$ $N=4$ theories admit two classes of topological twists.  We are interested in the topological twist which has the feature that local operators in this twist are holomorphic functions on the Higgs branch, in a particular complex structure.  We will call this the Rozansky-Witten twist, since at low energies it can be described by the Rozansky-Witten $\sigma$-model with target the Higgs branch.  

 Let us discuss this twist in a little more detail for a $3d$ $N=4$ theory built from a vector space $V$ and a group action $G$.  Let us choose one of the complex structures on $V$, and let $V_{\C}$ denote the corresponding complex vector space. Then, $V_{\C}$ is a holomorphic-symplectic vector space. Let $G_{\C}$ denote the complexification of $G$. The hyper-K\"ahler quotient of $V$ by $G$ is isomorphic to the holomorphic symplectic reduction of $V_{\C}$ by $G_{\C}$.

The Rozansky-Witten twist of the $N=4$ theory built from $V$ and $G$ is a gauged Rozansky-Witten theory built from the holomorphic symplectic vector space $V_{\C}$ with the group action $G_{\C}$.  We will not explain in detail how to describe gauged Rozansky-Witten theory, because the story becomes simpler after performing an $\Omega$-background construction.

A $3d$ $N=4$ theory on $\R^3$ can, after the Rozansky-Witten twist, be put in an $\Omega$-background.  This involves working equivariantly with respect to an $S^1$ which rotates a plane in $\R^3$.  We can denote the $\Omega$-background by $\R^2_{\eps} \times \R$.  When we put the theory in this background, it becomes effectively a quantum-mechanical system on the line $\R$ which is fixed by the rotation.  The algebra of local operators in this quantum mechanical system is then a deformation-quantization of the algebra of holomorphic functions on the Higgs branch of the theory \cite{Yag14,BulDimGai15}.  

In the case that our theory is represented in the UV as a gauge theory built from a hyper-K\"ahler vector space $V$ acted on by a group $G$, this algebra of operators has the following representation. As before, we let $V_{\C}$ denote the complex vector space $V$ with one of its complex structures, and let $G_{\C}$ denote the complexification of $G$, which acts on $V_{\C}$ by holomorphic symplectic symmetries. Let $\omega$ denote the holomorphic symplectic form on $V_{\C}$, which we think of as an anti-symmetric pairing on $V_{\C}$.  Then, we can construct a deformation quantization of the holomorphic symplectic quotient $V_{\C} \sslash G_{\C}$ as the quantum Hamiltonian reduction of the Weyl algebra $W(V_{\C})$ generated by the dual $V^\ast_{\C}$ by $G_{\C}$.  

We can also realize this quantum-mechanical system as a one-dimensional gauge theory.  The fields of this gauge theory are elements $\phi \in \cinfty(\R,V_{\C})$ and $A \in \Omega^1(\R, \mf{g}_{\C})$, consisting of a map to $V_{\C}$ and a $G_{\C}$ connection.  The action functional is 
$$
\tfrac{1}{\hbar}S(A,\phi) = \tfrac{1}{\hbar} \int_{\R} \omega( \phi, \d_{A} \phi ). 
$$ 
This theory has gauge symmetry given by maps from $\R$ to the complex group $G_{\C}$, acting on fields in the evident way.  The space of fields of this theory is a complex manifold, and the action functional is holomorphic.  We consider the theory as a ``complex'' theory, meaning that operators are holomorphic functions of the fields and that the path integral is performed over a real slice.  The parameter $\hbar$ is the equivariant parameter that appears in the $\Omega$-background construction. 

One can check that the commutative algebra of classical operators in this gauge theory is the algebra of holomorphic functions on the symplectic quotient $V_{\C} \sslash G_{\C}$, and so the algebra of operators at the quantum level is a deformation quantization of this.  Indeed, the algebra of classical operators is the algebra of holomorphic functions on the space of solutions to the equations of motion.   We will verify that this space is $V_{\C} \sslash G_{\C}$. 

Varying the field $\phi$ in the action functional $S$ tells us that $\d_{A} \phi = 0$.  Varifying the field $A$ tells us that for all $X \in \mf{g}$  and all $t \in \R$ we have $\omega(\phi(t) , X (\phi(t) ) = 0$.  This means that each $\phi(t)$ is in the zero locus of the moment map.    We can apply a gauge transformation which makes $A$ trivial, so that $\d \phi = 0$ and $\phi$ is constant.  Thus, $\phi$ becomes a constant map to $V_{\C}$ which is in the zero locus of the moment map.  Finally, quotienting by constant gauge transformations shows us that the space of solutions is the symplectic quotient $V_{\C} \sslash G_{\C}$.

\subsection{}
Let us apply these observations to the case of interest, when we have a $3d$ $N=4$ theory built from the infinite-dimensional hyper-K\"ahler vector space $\Omega^1(\R^4) \otimes \mf{g}$, acted on by the group of smooth maps from $\R^4$ to $G$.  In this case, as we have seen, the hyper-K\"ahler reduction is the moduli of instantons on $\R^4$.  To write down the quantum mechanical system that arises when we consider the $3d$ $N=4$ theory in the Rozansky-Witten twist and the  $\Omega$-background, we need to choose one of the complex structures on $\Omega^1(\R^4) \otimes \mf{g}$ and form the holomorphic symplectic reduction instead.  Let us choose a complex structure on $\R^4$, so that we identify $\R^4 = \C^2$. Then,
$$
\Omega^1(\R^4) \otimes \mf{g} = \Omega^{0,1}(\C^2) \otimes_{\C} \mf{g}_{\C}
$$ 
and the right hand side has an evident complex structure.  We will view $\Omega^{0,1}(\C^2) \otimes \mf{g}_{\C}$ as the space of $(0,1)$-connections on the trivial holomorphic $G_{\C}$-bundle on $\C^2$.  The group of smooth maps from $\C^2$ to the complex group $G_{\C}$ acts on $\Omega^{0,1}(\C^2) \otimes \mf{g}_{\C}$ by gauge transformations of the form
$$
A \mapsto g^{-1} A g +  g^{-1} \dbar g.
$$
These are holomorphic symplectic symmetries.  The holomorphic symplectic reduction is the moduli space of holomorphic $G_{\C}$-bundles on $\C^2$. 

 Thus, we see that the hyper-K\"ahler reduction of $\Omega^1(\R^4) \otimes \mf{g}$ is the moduli of instantons on $\R^4$, whereas the holomorphic-symplectic reduction one gets by choosing a particular complex structure is the moduli of holomorphic bundles on $\C^2$.    
 
In the case of an ordinary finite-dimensional $3d$ $N=4$ gauge theory, we have seen that after performing the Rozansky-Witten twist and putting the theory in an $\Omega$-background, we are left with a quantum-mechanical gauge theory built from $V_{\C}$ and $G_{\C}$.  Applying this argument to the infinite-dimensional sitaution at hand, we find a quantum-mechanical system with an infinite-dimensional space of fields, where the fields are
\begin{align*} 
\phi & \in \Omega^0(\R, \Omega^{0,1}(\C^2)) \otimes \mf{g}_{\C} \\
\alpha & \in \Omega^1(\R, \Omega^{0,0}(\C^2)) \otimes \mf{g}_{\C}. 
\end{align*}
We can combine these two fields into a $3$-component partial connection
$$
A = A_t \d t + A_{\zbar_1} \d \zbar_1 + A_{\zbar_2} \d \zbar_2
$$
where each component is a smooth function on $\R \times \C^2$. Then, the action functional for the quantum-mechanical system becomes the Chern-Simons action
$$
S(A)  = \int \d z_1 \d z_2 CS(A) = \int \d z_1 \d z_2 \left( \tfrac{1}{2} \ip{A, \d A}_{\mf{g}_\C} + \tfrac{1}{6} \ip{A,[A,A]}_{\mf{g}_\C}\right) 
$$
where $\ip{-,-}_{\mf{g}_{\C}}$ is the invariant pairing on the Lie algebra.   The gauge transformations are by smooth maps from $\R \times \C^2$ to $G_{\C}$ acting on the partial connection $A$.  

The resulting $5$-dimensional theory is (in the case $\mf{g}  = \mf{gl}_n$) the $5d$ gauge theory we introduced earlier in the limit where the plane $\C^2$ becomes commutative. 

\subsection{}
To sum up our discussion of $7$-dimensional gauge theory,  we have shown that a twist of the $7d$ maximally supersymmetric gauge theory, placed in an $\Omega$-background in one plane,  becomes our $5d$ Chern-Simons type theory. Further, if we introduce the loop expansion parameter $\hbar$ into the $5d$ gauge theory,  so that our action becomes $\tfrac{1}{\hbar} S(A)$, then we can identify $\hbar$ with the equivariant parameter in the $\Omega$-background construction in the plane. 

Combined with our arguments about $M$-theory, this discussion tells us that the $5$-dimensional gauge theory we have introduced arises from $M$-theory on $TN_k \times \R^2 \times \R \times \C^2$ in an $\Omega$-background with parameters $\eps$ and $\delta$, where $\delta$ matches the parameter $\hbar$ in our $5$-dimensional theory and $\eps$ matches the non-commutativity parameter $c$.  

We can sum up our results as follows. 
\begin{theorem}
 Let $TN_k^{\eps,r}$ denote the Taub-NUT manifold with radius $r$, an $A_{k-1}$ singularity at the origin, and $\Omega$-background parameter $\eps$.

Then, $M$-theory on $TN_k^{\eps,r} \times \R^2_\delta \times \R \times \C^2$ is equivalent to $GL(k)$ non-commutative gauge theory on $\R \times \C^2$ with action 
     $$S(A) = \frac{1}{\delta} \int \tfrac{1}{2} \op{Tr} A \ast_\eps \d A + \tfrac{1}{3} \op{Tr} A \ast_\eps A \ast_\eps A$$ 

\end{theorem}

\section{Quantization of the $5d$ gauge theory}
\label{section:quantization}
In this section we will prove a theorem about when this $5$-dimensional theory can be quantized.  As we mentioned above, this theory is highly non-renormalizable, as the classical action functional has infinitely many terms with more and more derivatives. Even so, we find that consistency of the quantization -- the BV quantum master equation -- constrains the problem of quantization so tightly that the only free parameters are $\eps$ and $\delta$.  

The parameter $\delta$ plays the role of the loop expansion parameter, and we will quantize order by order in $\delta$.  In order to cancel certain potential higher loop anomalies, we will also need to allow negative powers of $\eps$.  This means that the quantum-corrected action functional is a series in $\eps$ and $\delta$ where each positive power of $\delta$ can be accompanied by some finite number of negative powers of $\eps$. More formally, this means we have a theory over the sub-ring of the base ring $\C((\eps))[[\delta]]$ consisting of series $\sum_{i \ge 0} \delta^i f_i(\eps)$ where $f_i(\eps) \in \C((\eps))$ and $f_0(\eps) \in \C[[\eps]]$.

There are two versions of our quantization theorem, depending on whether we work on a general conical surface $X$ or on $\C^2$. In the case that we work on $\C^2$, it is natural to ask that we quantize the theory in a way compatible with all the symmetries of $\C^2$, whereas on a general conical $X$ we can only ask that the quantization is compatible with the $\C^\times$ action on $X$.  Asking that our quantization is compatible with these extra symmetries on $\C^2$ means that the theorem is slightly different in this case. 

Let us first state the version of our theorem that applies for a general conical symplectic surface $X$.  To state the theorem, we should recall that every class $w \in H^2(X)$ gives rise to a first-order deformation of the $\ast$-product on the sheaf $\Oo_X$ of holomorphic functions on $X$, of the form
$$
\alpha \ast_\eps \beta + \gamma \alpha \ast_{\eps,w} \beta
$$
for a square-zero parameter $\gamma$.  
This first-order deformation is compatible with the $\C^\times$-action on $X$ in the sense we discussed above. Further, every $\C^\times$-equivariant deformation of the $\ast$-product on $X$ is of this form.  
\begin{theorem}
Let $X$ be a conical complex symplectic surface, and consider our gauge theory on $\R \times X$ with gauge group $\mf{gl}_N$. Let us consider quantizing the theory in perturbation theory in the loop expansion parameter $\delta$, and let us work ``uniformly in $N$'' as discussed in detail in \cite{CosLi15}.  Let us also ask that the quantization is compatible with the $\C^\times$-action on $X$. Finally, let us allow negative powers of $\eps$ as long as they are accompanied by positive powers of $\delta$. Heuristically, this means that we should treat $\eps$ and $\delta$ as both being small but with $\delta$ much smaller than any positive power $\eps^n$ of $\eps$.  

Then, there are no obstructions (i.e.\ anomalies) to producing a consistent quantization.  At each order in $\delta$, we are free to add $\op{dim} H^2(X) + 1$ independent terms to the action functional, leading to an ambiguity in the quantization.   Explicitly, the deformations of the action functional we can incorporate at $k$ loops are 
 \begin{align*} 
 S & \mapsto S  + \delta^{k}\eps^{-k}\left( \tfrac{1}{2} \int \alpha \d \alpha + \tfrac{1}{3} \int \alpha \ast_\eps \alpha \ast_\eps \alpha\right)\\
 S & \mapsto S +   \delta^k \eps^{-k}\int \alpha \ast_{\eps} \alpha \ast_{\eps,w} \alpha 
\end{align*}
where $\ast_{\eps,w}$ indicates the first-order deformation of the $\ast$-product on $X$ associated to a class $w \in H^2(X)$. 

\end{theorem}

One way to summarize the result is the following. Recall that to define the action functional we need to give a $\ast$-product on the sheaf $\Oo_X$ of holomorphic functions on $X$, together with a holomorphic volume form $\omega$ on $X$.  The theorem states that to specify a quantization, we need to specify $\ast$-product and holomorphic volume form which both depend on $\delta$. 
 
More formally,the moduli space of quantizations consists of pairs consisting of:
\begin{enumerate} 
\item  A $\delta$-dependent associative $\ast$-product
$$
\alpha \ast_{\eps,\delta} \beta = \sum_{k \ge 0} \delta^k \alpha \ast_{\eps,k} \beta
$$
on $\Oo_X$, which at $\delta = 0$ are the original $\ast$-product on $X$.
\item  A $\delta$-dependent holomorphic volume $\omega_{\delta} = \sum \delta^k \omega_k$ on $X$, where $\omega_0$ is the original holomorphic volume form. 
\end{enumerate}
Further, the product $\ast_{\eps,\delta}$ must be $\C^\times$-invariant, where both $\eps$ and $\delta$ have weight $k$.  The holomorphic volume form $\omega_{\delta}$ must be of weight $k$ under the $\C^\times$-action on $X$. This data is taken up to gauge equivalence.

The theorem is proved in the appendix. The proof consists of calculating the cohomology groups describing obstructions and deformations to quantization.

Let us now describe the version of the theorem that applies to $\C^2$.  We will consider quantizations that are invariant under $\mf{sl}(2)$, compatible with the $\C^\times$ action scaling $\C^2$, and also holomorphically translation invariant. The latter means that they are translation invariant but that the translations $\dpa{\zbar_i}$ act in a homotopically trivial way (that is, in a way exact for the BRST operator). 

\begin{theorem}
Let us consider quantizing the theory on $\R \times \C^2$ in perturbation theory in the loop expansion parameter $\delta$, and let us work ``uniformly in $N$'' as discussed in detail in \cite{CosLi15}.  Let us also ask that the quantization is compatible with the $\C^\times$-action on $\C^2$, is $\mf{sl}(2)$-invariant, and is holomorphically translation invariant.  Also, let us allow negative powers of $\eps$ of the form $\eps^{-r} \delta^n$ where $r \le n$ (heuristically this means that $\eps$ and $\delta$ are both small and $\delta$ is much less than $\eps$).  

Then, there are no obstructions (i.e.\ anomalies) to producing a consistent quantization.  At each order in $\delta$, we are free to add $2$ independent terms to the action functional, leading to an ambiguity in the quantization.   Explicitly, the deformations of the action functional we can incorporate at $k$ loops are 
\begin{align*} 
S &\mapsto S + \delta^{k}\eps^{-k}\left( \tfrac{1}{2} \int \alpha \d \alpha + \tfrac{1}{3} \int \alpha \ast_\eps \alpha \ast_\eps \alpha\right)\\ 
S &\mapsto S + \delta^k \eps^{-k} \eps \dpa{\eps} \int \alpha \ast_\eps \alpha \ast_\eps \alpha.
\end{align*}
These two deformations of the action can be absorbed by a change of coordinates in $\eps$ and $\delta$.
\end{theorem}
It is important to note that when we work on $\C^2$, we only need to involve the parameters $\eps$, $\delta$ and $\eps^{-1} \delta$.  We can think of this as saying that the quantum theory is defined in the regime when $\abs{\delta} \ll \abs{\eps} \ll 1$. 

As in any quantum field theory, there is no canonical bijection between the coupling constants one writes down at the classical level and those that parametrize the quantum theory.  We quantize the theory order by order in powers of $\delta$, which is the loop expansion parameter. At each order in $\delta$, we are free to the two terms we mentioned Lagrangian.  In general, there is no canonical way to fix this choice. Different quantization schemes will lead to different choices.

The theorem says that, up to changes of coordinates on the space of fields, the only terms we can add at $n$ loops that do not generate an anomaly are those generated by applying the change of coordinates 
\begin{align*} 
\eps & \mapsto \eps \left(1 + \lambda_1 \delta^n \eps^{-n} \right)\\
\delta & \mapsto \delta \left(1 + \lambda_2 \delta^n \eps^{-n} \right)  
\end{align*} 
to the classical action, where $\lambda_1,\lambda_2$ are non-zero constants. 

What this tells us is that the quantum theory is parameterized by a point in a two-dimensional (formal) complex manifold which we can coordinatize by $\eps$ and $\delta$.  The choice of coordinates is not, however, canonical. Different quantization schemes will give different sets of coordinates, related by changes of coordinates of the form
\begin{align*} 
\eps & \mapsto \eps \left(1 + \lambda_1 \delta \eps^{-1} + \lambda_2 \delta^2 \eps^{-2} + \dots \right)\\
\delta & \mapsto \delta \left(1 + \mu_1 \delta \eps^{-1 } + \mu_2 \delta^2 \eps^{-2} + \dots  \right)  
\end{align*} 
\subsection{}
The proof of these results concerning quantization is rather technical, and is placed in the appendix. However, I will explain a little about the argument here so that the reader can get some feeling for the proof. 

The obstruction-deformation complex controlling quantizations of the theory is built from possible terms one can add to the Lagrangian.  For example, anomalies are governed by terms of ghost number $1$ that we can add.  Every first-order deformation arises from the integral of a local functional. 

Locally, every solution to the equations of motion of the  gauge theory is trivial. However, the trivial solution has an infinite-dimensional symmetry group whose Lie algebra is $\mf{gl}_N[[z_1,z_2]]$, where as usual $[z_1,z_2] = \eps$. It follows that local operators are functions of the ghost field, and form the  exterior algebra $\wedge^\ast (\gl_N[[z_1,z_2]])^\vee$, where the superscript $\vee$ indicates the linear dual. Concretely, if $\psi$ denotes the ghost field, these local operators are obtained by evaluating a product of $\psi$ and its derivatives in $z_1,z_2$ at a point in space-time. 

The BRST operator introduces a non-trivial differential, which can be identified with the Chevalley-Eilenberg differential.  Thus, the cochain complex of local operators is the Chevalley-Eilenberg cochain complex $C^\ast(\mf{gl}_N[[z_1,z_2]])$.    This has the same cohomology as the subcomplex of local operators invariant under the group $GL(N,\C)$. So, we will assume that the operators we will consider are all $GL(N,\C)$ invariant.

Invariant theory tells us that (classically) the algebra of local operators is generated by expressions like
$$
\Phi_{D_1,\dots,D_m}(\psi) = \op{Tr} \left(D_1(\psi) \dots D_m(\psi) \right)(z_i = 0)
$$
where the $D_i$ are elements of $\C[\dpa{z_1}, \dpa{z_2}]$.

For finite $N$, there are relations among these local operators. However, we have stated that we wish to quantize the theory for all $N$ simultaneously, in a compatible way.  This allows us to discard these relations; so that we find the operators $\Phi_{D_1,\dots,D_m}$ freely generate the commutative algebra of classical local operators.  

The BRST differential acts on these generators. The action was computed in classic work by Loday-Quillen \cite{LodQui84} and Tsygan \cite{Tsy83}.  These authors showed (in a more general situation) that we can identify the single-trace local operators with the linear dual of the cyclic homology of the non-commutative algebra $\C[[z_1,z_2]]$.  

When $\eps = 0$, the cyclic homology groups are infinite-dimensional.  However, when we make the algebra non-commutative, the cyclic homology groups collapse drastically.  This is similar to the way the algebra of bulk local operators in the topological $B$-model on a non-commutative space is trivial. As, bulk local operators in the $B$-model are the Hochschild cohomology of boundary local operators, which is a non-commutative algebra of the kind we are considering.

Ultimately, we find that the cohomology of the algebra of local operators is freely generted by a sequence of operators ghost numbers $3,5,7,\dots$. Modulo $\eps$, these operators can be represented by the expressions $\op{Tr} \psi^{2n+1} \dpa{z_1} \psi \dpa{z_2} \psi$, for $n \ge 0$. The precise formula for the lift to $\eps \neq 0$ is not relevant.   

We can convert local operators into Lagrangians by a descent procedure which shifts the ghost number.  For this particular field theory, the action of the translation vector fields $\dpa{z_i}$ can be implemented by a gauge symmetry.  This means that on the algebra of local operators, $\dpa{z_i}$ acts trivially. This is a situation in which the descent procedure is particularly simple.  We find that, if our space-time is $\R \times X$, we can couple a clas in $H^{5-k}(X)$ with a local operator of degree $k$ to get a deformation of the theory, and a class in $H^{6-k}$ with a local operator of degree $k+1$ to get a potential anomaly. 

The only local operators that can potentially contribute are those in degrees $3$ and $5$. We thus find that deformations come from classes in $H^2(X)$ and $H^0(X)$. These are the deformations discussed earlier. Anomalies would come from $H^3(X)$ and $H^1(X)$, but if we assume these groups vanish, there are no anomalies.  A separate argument shows that if $X$ is the cotangent bundle of a Riemann surface the potential anomalies coming from $H^1(\Sigma)$ can not arise.

\section{$M2$ branes}
In this section we will analyze the $M2$ in $M$-theory in the $\Omega$-background, and see what it becomes in $7$ and $5$ dimensions.  

Suppose that $M^7$ is a $G2$ manifold and $Y^4$ is a manifold of $SU(2)$ holonomy. Standard arguments  tell us that there are BPS $M2$ branes on submanifolds of $M^7 \times Y^4$ which are an associative submanifold of $M^7$ times a point in $Y^4$; and that there are BPS $M5$ branes on submanifolds which are the product of a coassociative submanifold of $M^7$ with a holomorphic curve in $Y^4$.  We will always consider $M2$ and $M5$ branes of this type.

Let us start by considering the $\Omega$-background where the parameter $\delta$ is zero. Thus, we rotate the circle fibre of Taub-NUT with speed $\eps$, and at the tip of Taub-NUT we find a $7$-dimensional non-commutative gauge theory as discussed above.  

Further, let us start by considering an $M2$ brane wrapping $\R^3$ in $TN_k \times \R^3 \times \C^2$. A standard argument  tells us that, without the $\Omega$-background construction, the $M2$ brane is described by an instanton in $7$-dimensional gauge theory.  It is natural to think that, after twisting and using the $\Omega$-background construction, it is still described by an instanton of charge $1$ in the gauge theory on $\R^3 \times \C^2$. But now, since the gauge theory is non-commutative, it should be described by a non-commutative instanton on $\C^2$ living in a neighbourhood of the support of the $M2$ brane.  Similarly, a stack of $k$ $M2$ branes will be described by an instanton  of charge $k$. 

Of course, after performing the additional $\Omega$-background construction which brings us to $5$ dimensions, we expect that the $M2$ brane becomes an instanton in the non-commutative gauge theory in $\R \times \C^2$, supported in a neighbourhood of $\R \times x$ for some $x \in \C^2$ which corresponds to the location of the $M2$ brane.

In other words:

\fbox{ \parbox{13cm}{$M2$ branes $\implies$ instanton particles in the $5$-dimensional gauge theory} } 

\subsection{}
Next, let us consider the theory living on an $M2$ brane. After putting $M$-theory on a Taub-NUT, the $M2$ brane becomes a $D2$ brane in type IIA on $\R^3 \times \R^3 \times \C^2$. If the first $\R^3$ is the base of the Taub-NUT circle fibration, then the $D2$ brane is situated on $0 \times \R^3 \times 0$.  Thus, the $D2$ branes lies entirely inside the $D6$ brane living at the singular point of the Taub-NUT. 

The theory on a stack of $N$ $D2$ branes by itself is a $3d$ $N=8$ gauge theory, which is equivalently $3d$ $N=4$ theory with adjoint matter.  This can be viewed as the $3d$ $N=4$ quiver gauge theory coming from the quiver with one node, labelled by $N$, connected to itself by an edge.  The inclusion of the $D6$ brane means that we should also consider fields associated to open strings stretching between the $D2$ and $D6$ branes. This leads to bifundamental matter, so that if we have a stack of $k$ $D6$ branes we find the $3d$ $N=4$ quiver gauge theory for the quiver
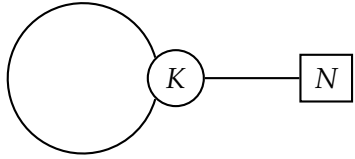
\begin{center}
\begin{tikzpicture}
  \node[circle,draw,thick](NL) at (0,0){$K$};
  \node[rectangle,draw,thick, inner sep=5](NR) at (2,0){$N$};
  \draw[thick](NL.east)--(NR.west);
  
  \draw[thick](NL.north west) arc (16:344:1) ;  

\end{tikzpicture}
\end{center}
This, of course, is the ADHM quiver describing instantons on $\C^2$ of rank $k$ and charge $N$.

From this, we deduce that (when the equivariant parameters $\eps,\delta$ are zero) the theory on the $M2$ brane is a twist of the $3d$ $N=4$ quiver gauge theory for the ADHM quiver.  One might worry that we have lost some information in the passage to IIA. However, this not the case, as we have seen that up to terms exact for the supersymmetery we are using everything is independent of the $M$-theory radius.

\subsection{}
So far, we have argued that the theory on a stack of $N$ $M2$ branes on $TN_k \times \R^3 \times \C^2$ can be seen as the Rozansky-Witten twist of the $3d$ $N=4$ quiver gauge theory for the ADHM quiver. Next, let's analyze what happens to the $M2$ brane when $\eps \neq 0$ and $\delta = 0$.  We will find that turning on $\eps$ leads to an FI term which has the effect of setting the moment map to $\eps$ times the identity of $\mf{gl}(k+1)$.      

In the presence of this FI term, the Higgs branch of the $3d$ $N=4$ theory is the variety consisting of matrices 
\begin{align*} 
 I &\in \Hom(\C^N,\C^{k+1}) \\
J &\in  \Hom(\C^{k+1},\C^N) \\ 
X,Y&\in  \gl(N) \\
 A & \in  \gl(N)).  
\end{align*}
satisfying the moment map equation
$$
[X,Y] + IJ = \eps \op{Id}  
$$
and taken modulo the action of $GL(N)$.  This variety is the space of instantons on the non-commutative deformation of $\C^2$.

There are two ways to see why we must have such an FI term when we turn on $\eps$.  One is to reduce to type IIA and use the results of \cite{SeiWit99}. There, the theory on a stack of $N$ $D0$ branes in the presence of $k+1$ $D4$ branes and a constant $B$-field was analyzed.  Without the $B$-field, the theory on the $D0$ branes is described by $N=8$ quiver quantum mechanics for the ADHM quiver.  If we turn on the $B$-field, Seiberg and Witten argued that this induces the FI term discussed above.

The $D0-D4$ system considered by Seiberg and Witten is related to the $D2-D6$ system considered here by applying $T$-duality to two of the directions common to the $D2$ and $D6$ branes.  This $T$-duality does not effect the $B$-field, since the $B$-field only involves the $4$ directions on the $D6$ brane which are orthogonal to the $D2$ branes.   From this we see that the $B$-field gives rise to an FI term on the $N=4$ theory living on the $D2$ branes.  The Higgs branch of the moduli of vacua of this theory is then the moduli of non-commutative instantons. 
\begin{comment}
An alternative way to derive the FI term in the theory on the $M2$ brane is to use mirror symmetry for $3d$ $N=4$ theories.  The Coulomb branch of the theory on a stack of $N$  $M2$ branes on $TN_k \times \R^3 \times \C^2$ is the moduli of rank $1$, charge $N$ instanton son $TN_k$. This moduli space  a certain quiver variety for the generalization of the ADHM quiver describing instantons on ALE spaces.  Thus, we can give a mirror description of the theory on the $M2$ branes as being the $3d$ $N=4$ theory built from the quiver describing rank $1$, charge $N$ instantons on an $A_k$ singularity. From this point of view, the $U(1)$ action on $TN_k$ is part of the flavour symmetry group.  With $\eps \neq 0$, the theory we are interested in is the TRW twist of the $3d$ $N=4$ theory for this new quiver, but made equivariant using this $U(1)$ in the flavour symmetry group. 

It is a standard part of the $3d$ mirror symmetry story that the equivariant parameters associated to the flavour symmetry, in the TRW twist, become FI parameters in the RW twist of the mirror theory.  In this way we deduce another derivation of the statement that turning on $\eps$ introduces an FI parameter giving us a $3d$ theory whose moduli of vacua is the moduli of non-commutative instantons on $\C^2$.
\end{comment}
\subsection{}
The relationship we derived between $5$-dimensional non-commutative gauge theory and $M$-theory is a little indirect, involving a reduction to type IIA and a supersymmetry argument to show independence of the $M$-theory radius.  The analysis of $M2$ branes suggests a more direct argument.  We have argued that, when $\eps \neq 0$, the vacua of the theory on an $M2$ brane which are preserved by the supersymmetry $\Psi(\eps)$ can be identified with the moduli of non-commutative instantons on $\C^2$.  Every supersymmetric $M2$ brane configuration leads to a supersymmetric $11$-dimensional supergravity solution, where the $M2$ branes have become black holes.    

This leads to the following conjecture.
\begin{conjecture}
For every non-commutative instanton on $\C^2 = \R^4$ of rank $1$ and charge $N$, there exists a solution to the equations of motion of $11$-dimensional supergravity obtained from putting $N$ $M2$ branes wrapping $\R^3$ in $TN \times \R^3 \times \C^2$ in the presence of the $C$-field $C(\eps) =-\eps V^{\flat} \d \zbar_1 \d \zbar_2 $ discussed earlier.  This supergravity background should have an $S^1$ action coming from the $S^1$ action on $TN$, and a  generalized Killing spinor $\Psi(\eps)$ whose square is $\eps V$ where $V$ is the vector field generating rotation.  
\end{conjecture}
Unfortunately, I was not able to see directly how to write down explicitly a supergravity solution associated to a non-commutative instanton.

\subsection{}
After introducing the further $\Omega$-background which brings us to the $5$-dimensional non-commutative gauge theory, this discussion tells us that the theory on the $M2$ brane is a certain topological quiver quantum mechanics obtained by putting $3d$ $N=4$ theory in an $\Omega$-background compatible with the Rozansky-Witten twist.  

We can write out the fields and the action explicitly.  The fundamental fields of the theory on a stack of $m$ $M2$ brane consist of smooth functions
\begin{align*} 
I &\in \cinfty(\R, \Hom(\C^m,\C^{k+1}) \\
J &\in \cinfty(\R, \Hom(\C^{k+1},\C^m) \\ 
X,Y&\in \cinfty(\R, \gl(m) \\
 A & \in \Omega^1(\R, \gl(m)).  
\end{align*}
The action is
$$
S = \frac{1}{\delta} \left\{  \int \op{Tr} ( X \d_\alpha Y) + \int \op{Tr}( I \d_\alpha J)   - \eps \int \op{Tr} \alpha   \right\} .
$$
The gauge symmetry of the theory is the group of maps from $\R$ to $GL(m,\C)$, which acts in the evident way. 

The equations of motion are as follows. If we vary $\alpha$, we find that the fields $X,Y,I,J$ satisfy the moment map equation
$$
[X,Y] + IJ = \eps \op{Id}.  
$$
The other equations of motion are that the fields $X,Y,I,J$ are covariant constant. We can use gauge symmetry to make the connection $\alpha$ trivial.  We are left, then, with constant fields $X,Y,I,J$, satisfying the moment map equation, modulo constant gauge transformations.  Thus, the space of solutions to the equations of motion is the ADHM moduli space of non-commutative instantons on $\C^2$. 

\section{$M5$ branes}
\label{section_M5}
In this section we will analyze how $M5$ branes behave when put in $M$-theory in the $\Omega$-background. The configuration of $M5$ branes we are interested in is as follows.  Consider $M$-theory on $TN_k \times \R^3 \times X$ where $X$ is a manifold of $SU(2)$ holonomy.  Let $\Sigma \subset X$ be a holomorphic curve in a chosen complex structure.  Then, consider an $M5$ brane on 
$$TN_k \times 0 \times \Sigma \subset TN_k \times \R^3 \times X.$$

Let us first consider what happens when $\eps = 0$.  
\begin{proposition}
An $M5$ brane placed on $TN_k \times 0 \times \C$, is a BPS object in $M$-theory equipped with the supersymmetry $\Psi(\eps = 0)$.  Further, the supersymmetry $\Psi(0)$ on the $M5$ brane is invariant under $\op{Spin}(4) \times \op{Spin}(2)$ once we have twisted the action of this part of the Lorentz group in $6$ dimensions using a certain homomorphism
$$ 
\op{Spin}(4) \times \op{Spin}(2) \to \op{Spin}(5)_R
$$ 
where $\op{Spin}(5)_R$ is the $R$-symmetry group of the theory on an $M5$ brane. 

This implies that the theory on the $M5$ brane with this twist can be put on a $6$-manifold of the form $M \times \Sigma$ where $M$ is a $4$-manifold and $\Sigma$ is a Riemann surface.  If we do this, then the twist of the  $4$-dimensional $N=2$ theory obtained from compactifying on $\Sigma$ is the Donaldson-Witten twist.  
\end{proposition}
Since $TN_k$ is a coassociative cycle in the G2 manifold $TN_k \times\R^3$, it is a standard result that it is supersymmetric for the supercharge we are using at $\eps = 0$ (we will nevertheless verify this directly).  The content of the proposition is the identification of the supercharge on the $M5$ brane with the Donaldson-Witten twist.  

\begin{proof}
Let us write the spin representation of $\op{Spin}(11)$ as
$$
S_{11d} = S_{4d} \otimes S_{3d} \otimes S_{2d} \otimes S_{2d}
$$
as a representation of $\op{Spin}(4) \times \op{Spin}(3) \times \op{Spin}(2) \times \op{Spin}(2)$.  Here, $S_{2d}$ is the $2$-dimensional spin representation of $\op{Spin}(2)$ which is a sum of the representations of weight $1/2$ and $-1/2$.  As above, $S_{4d}$ is $4$-dimensional and $S_{3d}$ is $2$-dimensional.  Let us further decompose
\begin{align*} 
 S_{4d} &= S_{4d}^{+} \oplus S_{4d}^-\\
S_{2d} &= S_{2d}^+ \oplus S_{2d}^- 
\end{align*}
into irreducible representations of $\op{Spin}(4)$ and $\op{Spin}(2)$ respectively. 

Let us consider, as before, the $11$-manifold $TN_k \times \R^3 \times \R^2 \times \R^2$, and consider an $M5$ brane living on $TN_k \times 0 \times \R^2 \times 0$.  The supersymmetries preserving the $M5$ bundle are covariant constant spinors in the sub-bundle  
$$
S_4^+ \otimes S_3 \otimes S_2^+ \otimes S_2 \oplus S_4^- \otimes S_3 \otimes S_2^- \otimes S_2.  
$$
The point is that, under the identification 
$$
S_{6d} = S_{4d} \otimes S_{2d}
$$
(where $S_{6d}$ is the $8$-dimensional representation of $\op{Spin}(6)$ which is a sum of chiral and anti-chiral representations) we have
$$
S_{6d}^+ = S_{4d}^+ \otimes S_{2d}^+ \oplus S_{4d}^{-} \otimes S_{2d}^-.$$
The $M5$ brane is preserved by those spinors which are chiral from the point of view of the $6$ dimensions on when the $M5$ brane lives.  

The spinor $\Psi(0)$ lives in 
$$
S_{4d}^+ \otimes S_{3d} \otimes S_{2d}^+ \otimes S_{2d}^+.  
$$
There is a homomorphism 
$$\rho : \op{Spin}(4) \to \op{Spin}(3)$$
 which under the decomposition $\op{Spin}(4) = SU(2)_+ \times SU(2)_-$ is the constant map on $SU(2)_-$ and is an isomorphism on $SU(2)_+$.  The spinor $\Psi(0)$ is invariant under the copy of $\op{Spin}(4)$ inside $\op{Spin}(4) \times \op{Spin}(3)$ which is embedded by the map $\op{Id} \times \rho$.  

This spinor is also invariant under the anti-diagonal $U(1)$ inside $\op{Spin}(2) \times \op{Spin}(2)$. 

From the point of view of the $M5$ brane, $\op{Spin}(3)$ and one of the two copies of $\op{Spin}(2)$ are part of the $R$-symmetry group $\op{Spin}(5)$.  This tells us that the supercharge $\Psi(0)$ is, from the part of view of the $M5$ brane, the supercharge that twists the theory using the homomorphism described above from
$$
\op{Spin}(6) \supset \op{Spin}(4) \times \op{Spin}(2) \to \op{Spin}(5)_R.  
$$

This twist of the theory on the $M5$ brane can be placed on any $6$-manifold of the form $M \times \Sigma$ where $M$ is a $4$-manifold.    It becomes the Donaldson-Witten twist of the $4d$ $N=2$ theory obtained from  compactifying the $6$-dimensional theory on $\Sigma$.  
\end{proof}

  We would also like to show that the $M5$ brane is also invariant under the supercharge $\Psi(\eps,\delta)$ whose square is the rotation $V(\eps,\delta)$.   This necessitates considering the $M5$ brane in the presence of the $3$-form $C(\eps,\delta)$.  Unfortunately, a direct verification that the $M5$ brane is supersymmetric in this  background is beyond my limited string theory skills. 

However, we can check this indirectly.  If we reduce to type IIA, the $M5$ brane becomes a $D4$ brane wrapping an $\R^3$ inside $\R^6$ and a holomorphic curve in $\C^2$.  The type IIA background is also equipped with a $B$-field proportional to $\d \zbar_1 \d \zbar_2$, and $k$ $D6$ branes. One can readily verify that the $D4$ brane is preserved by the supercharge we use in type IIA.

\subsection{The theory on the $M5$ brane}
We have argued that the $M5$ brane is supersymmetric in our set up.  We would like to understand the theory living on the $M5$ brane.  Of course, this is an ambitious goal, given the mysterious nature of the $(2,0)$ theory.  However, we will be able to make a precise statement once we work in the $\Omega$ background.  

To start with, let's take our parameter $\delta$ to be zero, and $\eps$ to be non-zero. Let's consider a stack of  $M5$ branes which wrap the Taub-NUT manifold and a holomorphic curve $\Sigma \subset \C^2$.   Then, we expect that the algebra of local operators on the $M5$ brane is localized to the fixed points of rotation, which is the curve $\Sigma$ times the tip of the Taub-NUT.

By theorem \ref{theorem_exactness}, everything is independent of the radius at infinity of the Taub-NUT manifold. Therefore, we lose no information by passing to the description in terms of type IIA string theory.

A stack of $N$ $M5$ branes becomes a stack of $N$ $D4$ branes in type IIA, wrapping $\R^3 \times 0 \times \Sigma$ in $\R^3 \times \R^3 \times \C^2$. We also have a stack of $k$ $D6$ branes wrapping $0 \times \R^3 \times \C^2$. 

We would like to understand the theory living on the $D4$ branes.  To do this, we will invoke the language of twisted string theory developed in \cite{CosLi16}. There, we conjectured that we can describe the twist of type IIA we are considering by a topological string theory which is the $A$-model on $\R^6$ and the $B$-model on $\C^2$.  The $B$-field becomes the Poisson tensor $-\eps \dpa{z_1} \dpa{z_2}$, which make the $B$-model on $\C^2$ into a non-commutative $B$-model.  In \cite{CosLi16}, we proved the open-string version of this conjecture, so that our description of the theory living on the twist of the $D4$ brane does not rely on any conjectural statements.

The theory on the  $D4$ brane on $\R^3 \times \Sigma$ is the open-string field theory associated to a differential graded algebra with a trace, which is
$$
\Omega^\ast(\R^3) \what{\otimes}\Omega^{\ast,\ast}(\Sigma)
$$ 
with the differential which is the sum of the de Rham operator on $\R^3$ and the $\dbar$ operator from $\Sigma$.  (The reason we find $\Omega^{\ast,\ast}(\Sigma)$ is that this complex is the Dolbeault model for the endomorphisms of the structure sheaf of $\Sigma$ in the derived category of sheaves on $\C^2$). 

Turning on the $B$-field adds the operator $-\eps \partial_\Sigma$ to the differential.  This is because when we make $\C^2$ non-commutative, and we make the structure sheaf $\Oo_\Sigma$ into a module for the non-commutative deformation of $\Oo_{\C^2}$, the endomorphisms of $\Oo_\Sigma$ in the derived category get deformed by adding the operator $-\eps \partial_\Sigma$ to the differential.   

Let us write down the action explicitly.  The fields are elements 
\begin{align*} 
 \alpha^0 &\in \Omega^\ast(\R^3) \what{\otimes}\Omega^{0,\ast}(\Sigma) \otimes \mf{gl}_N[1]\\
 \alpha^1 &\in \Omega^\ast(\R^3) \what{\otimes}\Omega^{1,\ast}(\Sigma) \otimes \mf{gl}_N.
\end{align*}
The action is
$$
\int\op{Tr} \left(\alpha^1 (\d_{dR}^{\R^3} + \dbar^{\Sigma}) \alpha^0 +\tfrac{1}{2}\alpha^1\alpha^0 \alpha^0 - \frac{\eps}{2}\alpha^0 \partial^\Sigma \alpha^0 \right). 
$$
After the change of coordinates $\alpha^1 \to -\eps \alpha^1$, and letting $\alpha \in \Omega^\ast(\R^3 \times \Sigma)\otimes\mf{gl}_N[1]$, we find the action becomes
$$
-\eps \int \op{Tr}\left(\tfrac{1}{2}\alpha \d \alpha +\tfrac{1}{3}\alpha^3\right). 
$$
In other words, it is the $5$-dimensional analog of Chern-Simons with coupling constant $-1/\eps$. 

\subsection{}
Now let's introduce the further $\Omega$ background, to reduce to $3$ dimensions. The fields become just $\alpha \in \Omega^\ast(\R\times \Sigma)\otimes \mf{gl}_N[1]$.  These are just the fields of Chern-Simons theory: the field of ghost number $0$ is in $\Omega^1(\R\times \Sigma)\otimes \mf{gl}_N$.  The action functional is also just the Chern-Simons functional to.  Indeed, working in the $\Omega$ background has the effect of replacing the de Rham complex of $\R^2$ by the equivariant de Rham complex, which we can replace by its cohomology which is just $\C$.   The action functional is given by the same formula as we wrote above, except that we integrate over $\R^2_\delta$ in equivariant cohomology. Since the integral of $1$ in equivariant cohomology is $\frac{1}{\delta}$, we find that the $3$-dimensional action functional is
$$
-\frac{\eps}{\delta} \int \op{Tr}\left(\tfrac{1}{2}\alpha \d \alpha +\tfrac{1}{3}\alpha^3\right). 
$$
Thus, we find ordinary Chern-Simons theory with coupling constant $-\eps/\delta$.  

This argument reproduces the result in the literature \cite{LuoTanYag14} that $5$-dimensional maximally supersymmetric gauge theory in an $\Omega$ background is Chern-Simons theory with a complex level.  

\subsection{$D4$-$D6$ strings}
The theory on a $D4$ brane becomes, in the $\Omega$ background, Chern-Simons theory.  This is not sufficient to describe the theory on the $M5$ brane, because we have not included the effect of the $D6$ branes. The presence of the $D6$ branes will give a defect in Chern-Simons theory located at the intersection of the $D6$ and $D4$ branes, which is the surface $\Sigma$.  The algebra of operators on this defect is the algebra of operators on the $M5$ brane at the tip of the Taub-NUT in the $\Omega$ background, which is what we are ultimately interested in.   

We can calculate the field theory that lives at the intersection of the $D4$ and $D6$ strings using the language of twisted string theory \cite{CosLi16}. We will reproduce a result derived in \cite{DijHolSulVaf08}. 

Let's first do the calculation before passing to the $\Omega$ background.  We need to calculate the space of open string states for strings stretching from a $D4$ brane to a $D6$ brane.  The space time is $\R^3 \times \R^3 \times \C^2$, and we are considering a topological string theory which is the $A$ model on $\R^6$ and the $B$ model on $\C^2$.  The $D4$ branes lives on $\R^3 \times 0 \times \Sigma$ and the $D6$ branes live on $0 \times \R^3 \times \C^2$.   Open string states between these branes are a tensor product of open string states between the $A$-model branes in $\R^6$ and those between the $B$-model branes in $\C^2$.  Since the branes in the $A$-model are transverse, the open string states are simply $\C$.  For the $B$-model, we find the Dolbeault complex computing
$$
\op{RHom}(\Oo_\Sigma,\Oo_{\C^2}).
$$
This Dolbeault complex is $\Omega^{1,\ast}(\Sigma) [-1]$, where the shift indicates that the fields in $\Omega^{1,0}$ are in odd degrees.  

Similarly, the $D6-D4$ strings give rise to $\Omega^{0,\ast}(\Sigma)[1]$, again with a shift placing it in an odd degree. Putting these two together, and assuming that there are $N$ $D4$ branes and $k$ $D6$ branes, tells us that the space of fields on the defect is 
\begin{align*} 
 \psi_1 &\in \Omega^{0,\ast}(\Sigma)[1] \otimes \op{Hom}(\C^k,\C^N) \\
\psi_2 &\in  \Omega^{1,\ast}(\Sigma)[-1] \otimes \op{Hom}(\C^N,\C^k). 
\end{align*}
The action functional which generates the BRST operator $\dbar$  is $\int \psi_1 \dbar \psi_2$.  

Thus, we find that the theory on the defect consists of $NK$ free chiral fermions, exactly as in \cite{DijHolSulVaf08}.  The way I have written it, the fermions are twisted so that the fermions are not sections of $K^{1/2}_\Sigma$ but of $\Oo_\Sigma$ or $K_\Sigma$.  However, the twist can be by an arbitrary degree $0$ line bundle, which can be absorbed into the gauge field of the Chern-Simons theory.

We arrive at the following claim.
\begin{claim}
The algebra of operators of the $(2,0)$ theory living on a stack of $N$ $M$ branes wrapping an $A_{k-1}$ singularity,  and placed in an $\Omega$-background, is the same as the algebra of operators living on the surface defect in $\mf{gl}_N$ Chern-Simons theory obtained by coupling $Nk$ free chiral fermions.   
\end{claim}
I have not analyzed at all how to quantize this system, preferring to leave this for future investigation.  Note, however, that to produce a consistent quantum system, the level must jump as we cross the surface defect, to cancel an anomaly.  

Disregarding any difficulties that may arise in the analysis of the quantum system, let us denote by $A_{N,k}(\eps,\delta)$ the vertex algebra of operators on the surface defect. A standard procedure produces from every vertex algebra an associative algebra, which is generated by the contour integrals of local operators in the theory along a circle in a punctured disc.  In the mathematics literature \cite{BeiDri04}, this construction is known as chiral homology along a punctured disc, and the result algebra is written $\int_{D^\times} A_{N,k}(\eps,\delta)$.  We will denote the associative algerbra by $\oint A_{N,k}(\eps,\delta)$.  

As an example, this construction produces from the chiral WZW vertex algebra the universal enveloping algebra of the Kac-Moody affine Lie algebra.

Our discussion so far leads to the following conjecture, which is in the same vein as the AGT conjecture \cite{AldGaiTac10}.
\begin{conjecture}
The associative algebra $\oint A_{N,k}(\eps,\delta)$ acts on the equivariant cohomology of the moduli of instantons of rank $N$ on a resolution of the $A_{k-1}$ singularity.  We take the direct sum over all charges of the cohomology of the moduli of instantons.  
\end{conjecture}
The physics derivation of this conjecture is that one can view the algebra $\oint A_{N,k}(\eps,\delta)$ as being the algebra of local operators in the $5$-dimensional maximally supersymmetric gauge theory obtained by reducing the theory on the $M5$ brane on a circle. Since we are in the $\Omega$ background, the local operators sit at the fixed points of the rotation action on the Taub-NUT manifold. The OPE in the remaining direction makes the collection of local operators into a non-commutative associative algebra.  The Hilbert space of this $5d$ theory is the sum over all charges of the equivariant cohomology of the moduli space of instantons,  and local operators act on this space.  

\subsection{}
This discussion suggests that, in the case $k = 1$, the vertex algebra $A_{N,1}(\eps,\delta)$ should be the $W_N$ vertex algebra.  Indeed, it is known \cite{MauOko12,BraFinNak14,SchVas13} that the associative algebra  $\oint W_N$ algebra acts on the cohomology of the moduli of instantons on the plane.  Yagi \cite{Yag12} has argued directly that the algebra of operators on in the $(2,0)$ theory in the $\Omega$ background is the $W$-algebra.  Further, a result of Beem, Rastelli and van Rees \cite{BeeRasvan15} shows that a construction similar in spirit (but different in details) to the $\Omega$-background construction produces the $W_N$ algebra from the theory on a stack of $N$ $M5$ branes.

Is it possible to check this directly? Let's start with the simplest case, when  $N = 1$ and when $\delta = 0$, so that Chern-Simons theory is treated classically.  Then, we are considering the system obtained by coupling a chiral fermion with fields $\psi_1,\psi_2$ and action $\int \psi_1 \dbar \psi_2$ to Abelian Chern-Simons theory.  Since Chern-Simons theory is being treated classically, the only affect the gauge theory has on the free fermion systme is that it forces us to consider $\C^\times$-invariant local operators.

A standard result (part of the boson-fermion correspondence) tells us that the vertex algebra of invariant local operators is precisely the vertex algebra describing the chiral free boson, with a single field $\alpha$ and OPE $\alpha \cdot \alpha \simeq 1/z^2$.  (Note that this is a free boson valued in $\C$ not $\C^\times$, so the twist fields do not play a role). 

This is heartening, as the $W_N$ algebra when $N = 1$ is the chiral free boson, and this is known \cite{Gro95,Nak97} to act on the equivariant cohomology of the moduli of torsion fre
e sheaves on $\C^2$.

\subsection{}
Another accessible case is When $k > 1$ and $N = 1$.  In this case, the the algebra we find is the $\C^\times$ invariants in $k$ pairs of chiral free fermions.   The boson-fermion correspondence tells us that $k$ pairs of chiral free fermions are equivalent to $k$ circle-valued chiral bosons. The bosons being circle valued means that we include the twist fields, so that the Hilbert space decomposes into sectors labelled by the homotopy class of the scalar field on a circle, which is an element of  $\Z^k = \pi_1((S^1)^k)$.   Under the boson-fermion correspondence, this decomposition corresponds to the one by weights of the action of $(\C^\times)^k$ on $k$ pairs of chiral free fermions.  

We are taking the invariants for a diagonal $\C^\times$ inside $(\C^\times)^k$. Under the boson-fermion correspondence, this means that instead of a boson valued on $(S^1)^k$, we find one which is a map to $\R^k / \Z^{k-1}$ where $\Z^{k-1} \subset \Z^k$ is the subgroup consisting of $(v_1,\dots,v_k)$ where $\sum v_i = 0$.    This copy of $\Z^{k-1}$ is the root lattice of type $A_{k-1}$, and can be  identified with the second homology group of the resolved $A_{k-1}$ singularity, equipped with its intersection pairing.  The vertex algebra for the chiral boson  built from $\R^k / \Z^{k-1}$ is thus the tensor product of a single non-compact chiral  free  boson, with the lattice vertex algebra for the root lattice of $A_{k-1}$. 

This vertex algebra is precisely the one that acts naturally on the cohomology of the moduli space of rank-one torsion free sheaves on the resolved $A_{k-1}$ singularity $\til{\C^2/\Z_{k-1}}$ \cite{Gro95, Nak97} The moduli of such torsion free sheaves has a map to $H^2( \til{\C^2/\Z_{k-1}})$ given by the first Chern class.  The fibre over an element of $H^2( \til{\C^2/\Z_{k-1}})$ is isomorphic to the Hilbert scheme of points on the surface $\til{\C^2/\Z_{k-1}}$. Therefore the cohomology of this moduli space of torsion free sheaves is a direct sum over the root lattice of $A_{k-1}$ of a copy of the cohomology of the  Hilbert scheme of points.  This is exactly the structure one finds on the Hilbert space of the chiral free boson built from $\R^k / \Z^{k-1}$.

\subsection{}
It turns out that it is possible to go further, and directly connect the operators living on the surface defect with the $W_N$ algebra.  The following argument was communicated to me by Davide Gaiotto.

Let's consider a more general situation, obtained by inserting an arbitrary two-dimensional chiral theory $T$ into Chern-Simons theory with gauge group $G$.   The WZW currents for $G$ act on the  theory $T$ with some level $k$.   If Chern-Simons theory has level $n$ on one side of the domain wall where we insert $T$, then anomaly cancellation forces us to have Chern-Simons of level $k+n$ on the other side.

Gaiotto has formulated the following general conjecture.
\begin{conjecture}
In this situation, the algebra of operators on the surface defect in Chern-Simons is described by the algebra
$$
\frac{T \times G_n}{G_{n+k}}
$$
where $G_k$ indicates the WZW current algebra at level $k$, $T$ indicates the vertex algebra of operators of the theory we insert, and the quotient means the coset construction.  
\end{conjecture}
This general conjecture implies what we want. As, a system of $N$ pairs of free fermions is equivalent to $U(N)_1$.   If we insert this system in to Chern-Simons theory at level $n$ as a domain wall, Gaiotto's formula tells us that the operators on the domain wall will be
$$
\frac{U(N)_1 \times U(N)_n}{U(N)_{n+1}}.  
$$
This is a standard realization of the $W_N$ algebra.  

For this argument to work as written, we would need Chern-Simons theory to have integral level. In our situation, the level is a function of the equivariant parameters, and it is not natural to restrict the equivariant parameters so that this function is integral.  However, Gaiotto's argument certainly provides very strong evidence for the conjecture that, for general values of the equivariant parameters, the algebra of operators on the surface defect is the $W_N$ algebra at general central charge.

\subsection{Parafermions and the $A_k$ singularity}
Gaiotto suggested a generalization of this argument which will allow us to understand what happens when we wrap  $M5$ branes on an $A_k$ singularity. In this case, if there are $N$ $M5$ branes, we find $N(k+1)$ pairs of free fermions embedded in Chern-Simons theory as a surface defect.   The vertex algebra we will find matches precisely with the generalization of the AGT relation to the case when the $M5$  brane wraps an $A_k$ singularity.

One can describe $Nk$ pairs of free fermions as $S(U(N)_{k+1} \times U(k+1)_N)$ where the $S$ indicates we look at matrices which live in $SU(N+k+1)$ via the block diagonal embedding from $U(N) x U(k+1)\to U(N+k+1)$. Applying Gaiotto's conjecture, we find  that the algebra of operators on the surface defect in Chern-Simons theory where on one side of the surface we have level $m$, is   
$$
\frac{S(U(N)_{k+1} \times U(k+1)_N) \times U(N)_m}{U(N)_{m+k+1}}
$$
This is the parafermionic vertex algebra conjectured in \cite{NisTac11} to replace the $W_N$ algebra in the AGT relation for $M5$ branes on $A_{k}$ singularities. 

\begin{comment}
\subsection{Spectral curves and Hitchin systems}
As has been observed before \cite{DijHolSulVaf08} in a very similar context, the $D4-D6$ strings form a module for the algebra of open-string states on the $D6$ brane. This algebra is, in the case $k = 0$, the non-commutative algebra $\Omega^\ast(\R^3) \what{\otimes} \Omega^{0,\ast}(\C^2)$ where the product on $\Omega^{0,\ast}(\C^2)$ is deformed to the Moyal product. Thus, the cohomology of this algebra is the non-commutative deformation of the algebra of holomorphic functions on $\C^2$.  This is (up to completion) the algebra $\op{Diff}_{\eps}(\C)$ of differential operators on $\C$.  

In this way, we find that the $D4-D6$ strings form a $D$-module on $\C$.  More generally, if we replaced $\C^2$ by $T^\ast \Sigma$, we would find that $D4-D6$ strings form a $D$-module on $\Sigma$.  This $D$-module is the ``quantum spectral curve''.  When $k \neq 0$, we find that instead of having a $D$-module on $\Sigma$, we have a module for the algebra $\op{Diff}_\eps(\Sigma) \otimes \mf{gl}_{k+1}$. 
\end{comment}

\subsection{The $M5$ brane as a t'Hooft operator in the $5d$ gauge theory}
We have argued for a description of the theory on the $M5$ brane in an $\Omega$-background as the theory living on a surface defect in Chern-Simons theory.   We would like to also understand the effect of this theory on the $5$-dimensional non-commutative gauge theory.   In this subsection we will give a heuristic argument that suggests this surface defect should be interpreted as a t'Hooft operator. Later we will analyze carefully the effect of the insertion of a free fermion system into the $5$-dimensional non-commutative a gauge theory, and verify explicitly this prediction.  

In $M$-theory, the $M5$ brane leads to a singularity in the field $H = \d C$ such that $\d H = \Delta_{M5}$ (where we use the upper-case $\Delta$ to indicate the delta-function on the location of the $M5$ brane, to avoid confusion with our parameter $\delta$).   In other words, the $C$-field is magnetically coupled to the $M5$ brane.It is natural to guess that the same thing happens when we consider the $M5$ brane in our non-commutative gauge theory.

Let us first consider the case $\delta = 0$, so that we are dealing with a $7$-dimensional gauge theory.   As we have seen, the fields of this theory, in the BV formalism, are given by
$$
\Omega^\ast(\R^3) \what{\otimes}\Omega^{0,\ast}(\C^2) \otimes \mf{gl}_N[1].
$$
In particular, there are bosonic fields living in 
$$
\Omega^3(\R^3) \what{\otimes}\Omega^{0,0}(\C^2) \oplus \Omega^2(\R^3) \what{\otimes}\Omega^{0,1}(\C^2) \oplus \Omega^1(\R^3) \what{\otimes}\Omega^{0,2}(\C^2). 
$$ 
These fields live in the subspace of $\Omega^3(\R^7)$ consisting of forms which don't involve $\d z_i$.  

It is natural to guess that this $7$-dimensional $3$-form $C_{7d}$ is related to the components of the $11$-dimensional $3$-form $C_{11d}$ which are orthogonal to the Taub-NUT.  In fact, what we will find is that there is a factor of $\eps$ in the relationship:  
$$
C_{11d}\mid_{\R^4 \times \C^2}  = \eps C_{7d}.  
$$  
I will not justify this factor of $\eps$ in detail right now. Instead, we will see how it appears when we relate the free-fermion surface defect discussed above to t'Hooft operators. 
In $11$ dimensions, the $M5$ brane forces the $4$-form field strength to have a singularity such that $\d H$ is a delta-function on the location of the $M5$ brane.  We will guess that a similar thing happens here, with an additional factor of $\eps^{-1}$. Thus, in the presence of an $M5$ brane wrapping the Taub-NUT manifold and a curve $\Sigma \subset \C^2$,  we expect that, if $H_{7d}$ is the field-strength of the $3$-form in the $7$-dimensional gauge theory, then 
$$\d H_{7d} = \eps^{-1} \Delta_{0 \times \Sigma},$$
where as before $\Delta$ means the delta-function. 

Now let's pass to $5$ dimensions by turning on $\delta$.  To understand what condition we find on the field strength of the gauge theory in $5$ dimensions, we need to relate the $5$-dimensional gauge field to the fields of the $7$-dimensional theory.  

The fields in $7$ dimensions can be written -- incorporating all anti-fields, ghost, etc. -- as the cochain complex $\Omega^\ast(\R^3) \what{\otimes} \Omega^{0,\ast}(\C^2) \otimes \mf{gl}_{k+1}[1]$.  The differential (the linearized BRST operator) is the sum of the de Rham operator on $\R^3$ and the Dolbeault operator on $\C^2$.  Passing to the $\Omega$-background means that we replace the de Rham complex on $\R^2$ by the equivariant de Rham complex $\Omega^\ast_{S^1}(\R^2)$.  The underlying graded vector space of the equivariant de Rham complex consists of the $S^1$-invariant forms $\Omega^\ast(\R^2)$, but with differential $\d_{dR} + \delta \iota_{V}$. Here $\delta$ is the equivariant parameter and $\iota_V$ is the operation of contraction with the vector field generating rotation. 

The $5$-dimensional and $7$-dimensional fields are related as follows.  Let $\mbb{1}$ denote any compactly-supported, equivariantly closed, form on $\R^2$ which is cohomologous to the unit in equivariant cohomology.  Then, the $5d$ gauge field $A$ is represented by 
$$
\mbb{1} \boxtimes A \in \Omega^\ast_{S^1}(\R^2) \what{\otimes} \Omega^\ast(\R) \what{\otimes}\Omega^{0,\ast}(\C^2).    
$$   
Since the  integral of $\mbb{1}$ is $1/\delta$, this explains why the $5d$ Chern-Simons type action has a $1/\delta$ not present in the $7d$ action. 

We can choose an explicit representative of the class $\mbb{1}$ as follows.  Let $f(r)$ be a function of the radius $r$ on $\R^2$ which has compact support, is  $1$ near the origin,  and  whose integral agains the volume form $\d x \d y$ is $1$. Then, the representative i 
$$
\mbb{1} = f(r) + \delta^{-1} f'(r) \d r \d \theta.
$$
It is immediate that this form is equivariantly closed.   Further, the integral of this is (if we work in a convention where $\theta$ has period $1$) $\delta^{-1}$.  

Note that, if $\Delta_{r = 0}$ is an equivariant representative of the delta-function at the origin in $\R^2$, then
$$
\mbb{1} = \frac{1}{\delta} \Delta_{r = 0}
$$
in equivariant cohomology.

If $A$ is the gauge-field of the $5$-dimensional gauge theory, it can be represented in $7$ dimensions by the field
$$
\til{A} = \mbb{1} \boxtimes A =  f(r) A + \delta^{-1} f'(r) \d r \d \theta A. 
$$
This is a sum of a $1$-form and a $3$-form in $7$ dimensions.

Similarly, the $5$-dimensional field strength is represented by
$$
\til{F}(A) = \mbb{1} \boxtimes F(A).  
$$
The equation satisfied by the $7$-dimenisonal field strength in the presence of a t'Hooft operator, when applied to $\til{F}(A)$, becomes the equation
$$
\mbb{1} \boxtimes \d F(A) =\eps^{-1} \Delta_{r = 0} \boxtimes \Delta_{0 \times \Sigma}. 
$$
Here the first delta-function is the one for the origin in $\R^2$, but taken in equivariant cohomology. The second one is for the submanifold $0 \times \Sigma \subset \R \times \C^2$,where the surface $\Sigma \subset \C^2$ is the location of the $M5$ brane. 

Since $\mbb{1} = \delta \Delta_{r = 0}$, we find that
$$
\d F(A)= \frac{\delta}{\eps} \Delta_{ 0 \times \Sigma}.  
$$

This means that the gauge field $A$ will look like a monopole in the normal direction to the location of $\Sigma$, but with charge $\tfrac{\delta}{\eps}$.  

To sum up, we find
\begin{proposition}
 An $M5$ brane wrapping the Taub-NUT manifold $TN_k$ gives rise to a  t'Hooft surface operator in the $5$-dimensional gauge theory of charge $\delta/\eps$. 
\end{proposition}

\section{Fermions and t'Hooft operators}
We now have two descriptions of the relationship between the $M5$ brane and the $5$-dimensioanl gauge theory.  On the one hand, we have argued that the $M5$ brane should be viewed as a t'Hooft surface operator. However, we also have a more explicit description whereby the theory on a stack of $N$ $M5$ branes can be seen (at $\delta = 0$) as the $GL(N)$ invariants in $N(k+1)$ pairs of chiral free fermions. In this subsection we will relate these two descriptions.  

The free fermions live in the sum of the bifundamental representation of $GL(N) \times GL(k+1)$ with its dual.  In this way we can couple the free fermion systme to the non-commutative $GL(k+1)$ gauge theory. 

Let $\Sigma\subset \C^2$ be the surface on which the free fermion system lives, and $A$ be the connection of the $5$-dimensional gauge theory. If we turn off the non-commutativity by setting $\eps = 0$,  then we can incorporate the field $A$ into the action for the free-fermion theory by replacing the $\dbar$ operator by the operator $\dbar_A$, giving us the action
$$
\int_{\Sigma} \psi \dbar_A \psi' + O(\eps)
$$
where 
\begin{align*} 
 \psi &\in \cinfty(\Sigma) \otimes \op{Hom}(\C^N, \C^{k+1} ) \\
\psi' &\in \Omega^{1,0}(\Sigma) \otimes \op{Hom}(\C^{k+1}, \C^{N} ). 
\end{align*}

The $\eps$-dependent terms are a bit more complicated.  Suppose, for simplicity, that the surface $\Sigma$ is simply where $w = 0$, if $z,w$ are Darboux coordinates on $\C^2$. Then, the action for the theory on the surface, to all orders in $\eps$, is 
$$
\int \psi \dbar \psi'+  \sum_{n \ge 0}\eps^n \frac{1}{ n!} \psi \left( \dpa{w}^n A \right) \dpa{z}^n \psi'.  
$$  
For a general $\Sigma$, to write down the action coupling the system, one has to first make the structure sheaf $\Oo_\Sigma$ into a module for the algebra $\Oo_\eps(\C^2)$ of holomorphic functions on the non-commutative $\C^2$.   We have written down the action in the case that $\Sigma$ is the locus $w = 0$, and $\C[z,w]$ acts on $\C[z]$ by making $w^n$ act by $\eps^n \partial_z^n$.

How can the insertion of this free fermion system be interpreted as a t'Hooft operator? It is certainly not true in general that one can realize t'Hooft operators in such a simple way. A surface (or line) operator obtained by coupling to a free fermion system would normally be a Wilson surface (or line). 

However, we will find an unusual phenomemon in non-commutative gauge theory whereby the free-fermion surface operator can be viewed as a source for a field with the singularities of a t'Hooft operator.

\subsection{}
The starting point is the observation that there is an anomaly that appears when one tries to couple the $GL(k+1)$ gauge theory to the free fermion surface operator. This anomaly arises from the diagram in figure \ref{fig:chiral_anomaly}. 
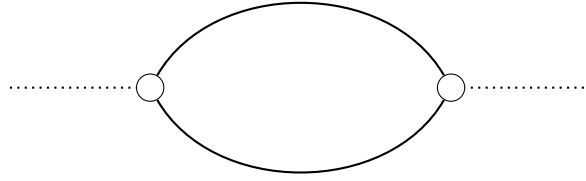
\begin{figure}

\label{fig:chiral_anomaly}
\begin{tikzpicture}
\node (D) at (0,2){};
  \node[circle, draw](A) at (2,2){};
  \node[circle, draw](B) at (6,2){};
\node (E) at (8,2){};
  \draw[thick](A) to [out=60,in=120](B);
  \draw[thick] (A) to [out=-60, in = -120] (B);
  \draw[thick,dotted](D)--(A); 
  \draw[thick,dotted](E)--(B); 
\end{tikzpicture}
\caption{The diagram giving the chiral anomaly that arises when coupling a surface operator. The solid lines are in the $\beta-\gamma$ system on the surface, and a $5$-dimensional gauge field is placed on the dotted lines.} 
\end{figure}

As a function of the fields of the gauge theory, at $\eps = 0$, the anomaly is given by
$$
\frac{1}{2 \pi i}\int_{\Sigma} A \partial \chi
$$
where $\chi$ is the ghost field, living in $\Omega^{0,0}(\R \times \C^2) \otimes \mf{gl}(k+1)$.   For $\eps \neq 0$, the expression is more complicated.   

However, a cohomological analysis detailed in appendix \ref{appendix_quantization} tells us that the $\eps$-dependent terms are completely fixed. The point is that there are very few single-trace Lagrangians we can build from the fields of the gauge theory.  All single-trace Lagrangians we can construct are built by a descent procedure from single-trace local operators. There is a unique single-trace local operator of ghost number $3$, as well as one of ghost number $5$, $7$, $9$, etc. Those of ghost number greater than $3$ will not be relevant for the present discussion.  

The operator of ghost number $3$ is of the form
$$ 
\mscr{O}_3 (\chi) = \op{Tr} (\chi^3) + c \eps \op{Tr} (\chi \dpa{z_1} \chi \dpa{z_2} \chi) + O(\eps^2), 
$$ 
for some combinatorial constant $c$ determined by the requirement that this operator is BRST closed. (As above, $\chi$ is the ghost field). 
 
A descent procedure allows us to view this operator as being a function, again of ghost number $3$, which maps fields of the theory to forms on $\R \times \C^2$.  We will call the form-valued version of this operator $\mscr{O}_3^{\Omega^\ast}$.  

The reason one can apply such a descent procedure is that, since infinitesimal translation in $\R \times \C^2$ is BRST exact, each derivative of the operator $\mscr{O}_3(\chi)$ is BRST exact.  The object making this derivative BRST exact provides the $1$-form component of $\Oo^{\Omega^\ast}_3$. One proceeds in a similar way to get the higher-form components.  

Note that we need $\eps \neq 0$ for the translations $\dpa{z_i}$ to be BRST exact when the act on local operators.  As, when $\eps \neq 0$, these translations act as gauge symmetries, and infinitesimal gauge symmetries always act in a BRST exact way on the space of local operators.  One can explicitly calculate that, for example, the component of $\Oo_3^{\Omega^\ast}(\chi)$ which is a $(1,1)$-form on $\C^2$ is of the form
$$
c \op{Tr} (A \partial \chi) + O(\eps). 
$$ 
If we integrate this over the surface $\Sigma$ we will indeed produce a non-zero multiple of the anomaly, to leading order in $\eps$.

One might worry that the higher-order terms in $\eps$ might differ from those in the anomaly by some terms arisng from the integral of $\eps^k \Oo_3^{\Omega^\ast}$. However, this is not possible for symmetry reasons, at least when we specialize to the case when the surface $\Sigma$ is a linearly embedded copy of $\C$ in $\C^2$.  In this case, there's a $\C^\times$ action on $\C^2$ which scales the normal direction to $\Sigma$. Since this action scales the Poisson tensor on $\C^2$,the parameter $\eps$ has weight $1$ with respect to this action.  Both the anomaly and the local operator $\Oo_3$ are fixed by this $\C^\times$ action. Therefore, if we can identify the anomaly with the integral of $\Oo_3^{\Omega^\ast}$ over $\Sigma$ to leading order in $\eps$, they must be the same to all orders in $\eps$.      

The end result of this discussion is to note that, in the case $\Sigma = \C \subset \C^2$,  there is a one-dimensional space of single-trace Lagrangians preserving the $\C^\times$ symmetries which we can integrate over the surface $\Sigma$ to produce an anomaly.  We will invoke this kind of argument at various points in our analysis.  

\subsection{}
To relate the anomaly to the t'Hooft operator, we will give a different description of the anomaly. 

Given any open subset $U$ in $\C^2$, we can consider the non-commutative deformation $\Oo_\eps(U)$ of holomorphic functions on $U$. There is a map
$$
\Oo_\eps(U) \to \op{Der} (\Oo_\eps(U) )
$$
which sends a function $f$ to the inner derivation given by the commutator with $f$.  The derivation $[f,-]$ is unchanged if we add a constant to $f$.  Thus, it only depends on the holomorphic $1$-form $\partial f$; and since every closed holomorphic $1$-form is locally exact, we have a map
$$
\Omega^1_{cl,hol}(U) \to \op{Der} (\Oo_\eps(U) ).   
$$ 

In the same way, every element
$$
f \in \Omega^\ast(\R) \what{\otimes} \Omega^{0,\ast}(U) 
$$
gives rise to an inner derivation of the non-commutative algebra $\Omega^\ast(R) \what{\otimes} \Omega^{0,\ast}(U)$ equipped with its Moyal product.  The cohomological degree of this derivation is the same as that of $f$.  

This derivation only depends on $\partial f$. Therefore, we can associate a derivation $D(\alpha)$  to any $\alpha \in \Omega^\ast(\R) \what{\otimes} \Omega^{1,\ast}(U)$ satisfying $\partial \alpha = 0$. If $(\d_{dR}^{\R} + \dbar) \alpha = 0$ also then this derivation will commute with the differential on $\Omega^\ast(\R) \what{\otimes} \Omega^{0,\ast}(U)$. 

Recall that that the Lagrangian of our theory is of the open-string field theory type, build from a differential graded algebra with a trace.  The dga is 
$$\mc{A} = \Omega^\ast(R) \what{\otimes} \Omega^{0,\ast}(\C^2) \otimes \mf{gl}(k+1),$$ equipped with the differential $\d_{\mc{A}} = \d_{dR}^{\R} + \dbar^{\C^2}$ and the non-commutative Moyal product.  Trace is given by integration against the holomorphic symplectic form on $\C^2$.   

We can generate first-order deformations of our theory by producing first-order deformations of the data $(\mc{A}, \d_{\mc{A}}, \op{Tr})$.  One way to produce such a deformation is to perturb the differential $\d_{\mc{A}}$ to $\d_{\mc{A}} + \gamma D$, where $\gamma$ is a parameter of square zero and $D$ is a derivation of degree $1$ of $\mc{A}$.  In order to produce a consistent deformation of the open-string algebra, we also need that $[D,\d_{\mc{A}}] = 0 $ and $\op{Tr} D A= 0$ for all $A \in \mc{A}$. This second condition is automatic for the derivations we will consider.  

Suppose we have some element 
$$
\alpha \in \oplus_{i+j = 1} \Omega^i(\R) \what{\otimes} \Omega^{1,j}(\C^2)
$$
satisfying the equations
\begin{align*} 
(\d_{dR}^{\R} + \dbar^{\C^2}) \alpha &= 0 \\
\partial^{\C^2} \alpha &= 0. 
\end{align*}
Then, the derivation $D(\alpha)$ of $\mc{A}$ gives rise to a first-order deformation of the Lagrangian of our $5$-dimensional gauge theory, in which we add on the quadratic term
$$
\frac{1}{\delta} \int \d z \d w \op{Tr}  \left(\tfrac{1}{2} A \d A + \tfrac{1}{3} A^3 \right) +  \frac{\gamma}{2 \delta}  \int\op{Tr}  A D(\alpha)(A)  
$$
where as before $\gamma$ is a parameter of square zero.  

If we can  write $\alpha = \partial f$, then this deformation can be written
$$
\frac{\gamma}{2 \delta} \int \op{Tr} \left( A [f\otimes\op{Id} ,A]\right). 
$$
In this situation, the deformation of the Lagrangian is that obtained by doing perturbation theory around the background $A = \gamma f\otimes \op{Id}$ instead of $A = 0$. 

\subsection{}
It turns out that the anomaly obtained by the insertion of a surface operator can be written in this way.
 
If we take the square-zero parameter $\gamma$ to be of cohomological degree $-1$ instead of $0$, the analysis above tells us that that an element
$$
\alpha \in \oplus_{i+j = 2} \Omega^i(\R) \what{\otimes} \Omega^{1,j}(\C^2)
$$
satisfying 
\begin{align*} 
(\d_{dR}^{\R} + \dbar^{\C^2}) \alpha &= 0 \\
\partial^{\C^2} \alpha &= 0. 
\end{align*}
gives rise to a Lagrangian of ghost number $1$, that is, a potential anomaly.  

The anomaly associated to the the insertion of a surface operator living on $0 \times \C$ in $\R \times \C^2$ is associated, in this way, to the form
$$
\frac{\delta}{ 2 \pi i \eps} \Delta_{t = 0, w = 0} \in \Omega^1(\R) \what{\otimes} \Omega^{1,1}(\C^2).  
$$
Here we are using coordinates $t,z,w$ on $\R \times \C^2$, where $t$ is the coordinate on $\R$ and the surface operator lives on the $z$-plane.  Of course, this form is singular, but the analysis above still applies.

To see why this is the case, note that 
$$
\Delta_{t = 0, w = 0} = \frac{-1}{2 \pi i} \partial \left(  \Delta_{t = 0} \frac{ \d \wbar}{\wbar}  \right).  
$$
This identity allows us to write the derivation associated to $\Delta_{t = 0, w = 0}$ as the commutator with the field on the right hand side of the equation. We can calculate this commutator as follows:  
\begin{align*} 
\frac{-1}{2 \pi i} \left[ \Delta_{t = 0} \wbar^{-1} \d \wbar, A  \right] &= \frac{-1}{2 \pi i} \sum_{n \ge 1}\eps^{n} \frac{1}{2^n n!} \frac{\partial^n}{\partial^n w} \left( \wbar^{-1} \d \wbar \right) \frac{\partial^n}{\partial^n z} A\\
&= \eps \Delta_{t = 0}\dpa{z} A_{w = 0} + O(\eps^2). 
\end{align*}
If, as before, we write $D(\Delta_{t = 0, w = 0})$ for the derivation associated to this singular form, we have 
\begin{align*} 
\frac{1}{\delta \eps} \int_{\R \times \C^2}\op{Tr} A D(\Delta_{t = 0, w = 0} ) A \d z \d w &=\frac{-1}{2 \pi i\delta  \eps}  \int_{\C^2}\op{Tr} A [\wbar^{-1} \d \wbar, A] \d z \d w \\
&= \frac{1}{\delta} \int_{\C} \op{Tr} A \dpa{z} A \d z + O(\eps).  
\end{align*}
It follows, by the uniqueness arguments discussed above, that the anomaly is given by the derivation associated to $\frac{\delta}{ \eps} \Delta_{t = 0, w = 0}$ to all orders in $\eps$.   

\subsection{}
Finally, we can see how to relate the surface operator to a t'Hooft operator. The anomaly arising from the surface operator arises from the derivation coming from $\tfrac{\delta}{ \eps} \Delta_{t=0, w = 0}$.  To cancel this anomaly we need to consider an element
$$
F \in \oplus_{i+j = 1} \Omega^i(\R) \what{\otimes} \Omega^{1,j}(\C^2)
$$
satisfying
\begin{equation}
\label{tHooft}
\begin{aligned}
\partial F &= 0\\
\left(\d_{dR}^{\R} + \dbar^{\C^2} \right) F + \tfrac{1}{2} [D(F), D(F) ]  &= \tfrac{\delta} { \eps} \Delta_{t = 0, w = 0} .
\end{aligned}
\end{equation}
As above, $D(F)$ is the derivation associated to $F$. We will typically consider $F$ such that $\dpa{z} F = 0$, in which case $[D(F), D(F)] = 0$. 

Equations \ref{tHooft} are those satisfied by the field strength when we have a t'Hooft surface operator of non-integer charge $\tfrac{\delta}{\eps}$.  I will explain shortly how to think of t'Hooft operators of non-integer charge.   

Given an $F$ satisfying these equations, then adding the term $\frac{1}{\delta} \int A D(F) A \d z \d w$ to the Lagrangian will cancel the anomaly. (Because the action functional is accompanied by a factor of $\delta^{-1}$, it is convenient to have such a factor next to the term involving $F$, which is why $F$ involves a $\delta$).  

If, locally, we have a field $A_0$ such that $F = \partial A_0$ satisfies the equations above, then working in perturbation theory around $A_0$ has the effect of adding the term $\int A D(F) A$ to the Lagrangian. (The factor of $\delta$ appears to cancel the $\delta^{-1}$ in front of the action).   The equations on $A_0$ are those satisfied  by a field in the presence of a t'Hooft surface operator of charge $\delta/\eps$.  

We thus conclude that, to cancel the anomaly, the free-fermion surface defect forces the fields away from the defect to behave like fields in the presence of a t'Hooft surface defect.  Thus, the free fermion surface defect can be thought of as a t'Hooft defect of charge $\delta/\eps$.  This confirms our previous analysis.  

\subsection{}
There are many ways that one can solve equations \ref{tHooft}.  One solution, up to factors of $2$ and $\pi$, is obtained by setting 
$$
F = \delta \frac{\wbar \d t \d w + 2 t \d w \d \wbar   }{\eps (t^2 + w \wbar)^{3/2}  } . 
$$
(This satisfies the equation $[D(F), D(F)] = 0$ as it is independent of $z$).

Another way to proceed, which produces a more singular field,  is as follows.  

Since 
$$
\Delta_{w = 0} = \frac{1}{2 \pi i} \dbar w^{-1}\d w 
$$
we find that we can take 
$$
F = \frac{\delta }{\eps (2 \pi i^2} \Delta_{t = 0} w^{-1} \d w.  
$$
Choosing a branch cut, we can write
$$
F = \partial \left( \frac{\delta }{\eps (2 \pi i)^2} \Delta_{t = 0} \log w\right)
$$ 
so that a gauge field sourced by the t'Hooft operator is
$$
A_0 = \frac{\delta }{\eps (2 \pi i)^2} \Delta_{t = 0} \log w.
$$
This is only well-defined away from the branch cut of $\log w$.  
 
The field $A_0$ is gauge equivalent to a less singular field Let us introduce a bump function $\phi(t)$ which is supported on a neighbourhood of $t = 0$ and whose integral is $1$.  Then, the $1$-form $\Delta_{t = 0}$ on $\R$ is cohomologous to $\phi(t) \d t$.  Therefore $A_0$ is gauge equivalent to
$$
A_0' = \frac{\delta }{\eps (2 \pi i)^2} \d t \phi(t) \log w.
$$

\subsection{t'Hooft operators with non-integer charge}
An $M5$ brane leads to a t'Hooft operator in the $5$-dimensional gauge theory of charge $N\delta /\eps$.  What does it mean to have a t'Hooft operator of non-integral charge? The answer is that it is not really an operator, but rather a deformation.

Suppose that $F$ is a form on $\R \times \C^2$ satisfying the equations \ref{tHooft} for the field-strength of a t'Hooft operator of charge $N\delta /\eps$. Locally, we can write $F = \partial A_0$ where $A_0$ is a field of the gauge theory.  Then, one can write the action functional $S$ in the background $A_0$,
$$
S_F (A) = S(A_0 + A). 
$$
Then, one can check that $S_F$ only depends on $F = \partial A_0$, and not on $A_0$ itself.

This is a special feature of the t'Hooft operators we consider, which are t'Hooft operators associated to the Abelian group $\mf{gl}_1$ in a non-commutative gauge theory.  The point is that the commutator $[A_0,-]$ that appears in the action in the background $A_0$ only depends on $\dpa{z_i}$-derivatives of $A_0$, and so only on $F$.  

For example, if we choose $F$ to  be a delta-current on a $3$-manifold which bounds the surface where the t'Hooft operator lives, then this deformation $S_F$ of the original action leads to the insertion of a defect on the $3$-manifold supporting $F$. In this case, the non-integral t'Hooft surface operator lives on the boundary of a $3$-dimensional operator.

Now, if the t'Hooft operator is of integral charge, it really can be thought of as an operator by itself and not as a surface operator living on the boundary of a $3$-manifold operator. The point is that in this case, the $3$-manifold operator is trivial. Indeed, away from the location of the t'Hooft operator, we can choose a $\C^\times$-bundle with non-trivial charge on the $S^2$ linking the surface of the t'Hooft operator, and whose curvature is $F$.  This means that away from the t'Hooft operator the deformation $S_F$ really arises considering the action $S$ in a non-trivial background, where the $GL(k+1)$ bundle is topologically non-trivial. 

To sum up, we have now found by two different methods that the $M5$ brane gives rise to a t'Hooft operator of charge $\delta/\eps$ in the $5$-dimensional gauge theory. 

\section{Relating the $5d$ gauge theory to the topological string}
In \cite{CosLi16}, we argued that a twist of type IIA string theory can be described by a topological string theory which is a mix of the topological $A$ and $B$-models.  Here, we have argued that $M$-theory, when placed in an $\Omega$-background, can be modelled by a $5$-dimensional gauge theory.  

In this section we will verify that these two statements are consistent with each other, by relating the compactification of our $5$-dimensional gauge theory to $4$ dimensions with a topological string theory obtained from deforming the topological $B$-model in $2$ complex dimensions.

Let us first recall the string-field theory of the topological $B$-model on a complex symplectic surface $X$.  The fields of the space-time theory (BCOV theory or Kodaira-Spencer theory \cite{BerCecOog94}) consists of polyvector fields $\PV^{\ast,\ast}(X)$ which are in the kernel of the operator 
$$\partial : \PV^{i,j} (X) \to \PV^{i-1,j}(X).$$ 
 The holomorphic volume form gives an isomorphism   
$$
\PV^{2,\ast}(X) \iso \Omega^{0,\ast}(X) \iso \Omega^{2,\ast}(X)
$$
and so defines an integration map
$$
\int : \PV^{2,2}(X) \to \C. 
$$
The action functional of the theory, in this formulation, is non-local in takes the form
$$
S(\alpha) = \tfrac{1}{2} \int \alpha \dbar \partial^{-1} \alpha + \tfrac{1}{3} \int \alpha^3 
$$
where we use the natural wedge product on polyvector fields. 

The fields in $\PV^{2,\ast}(X)$ do not propagate.  The only propagating fields are in $\PV^{1,\ast}(X)$ and $\PV^{0,\ast}(X)$.  The field in $\alpha^{1,\ast} \in \PV^{1,\ast}(X)$ is in the kernel of the operator $\partial$,  so that (at least locally) it can be written as 
$$
\alpha^{1,\ast} = \partial \left( A^{0,\ast}\pi\right)
$$
where we have introduce a new field $A \in \Omega^{0,\ast}(X)$, and $\pi \in \PV^{2,0}$ is the holomorphic Poisson tensor.

If we denote the field in $\PV^{0,\ast}(X)$ by $B$, we find that in this new set of fields the action is now local and reads
$$
S = \int \Omega B \dbar A  + \tfrac{1}{2} \int B \partial A \partial A. 
$$
Here, we have written everything in the language of differential forms, and $\Omega$ is the holomorphic volume form on $X$. 

The grading is such that $A^{0,1}$ and $B^{0,0}$ are bosonic. Thus, we should view $A^{0,1}$ as the $(0,1)$ component of a connection for $\mf{gl}(1)$ and $B^{0,0}$ as a Lagrange multiplier field.   The remaining fields are ghosts and antifields. 

If we only consider the quadratic term in the action $S$, the theory is holomorphic BF theory, which \cite{Cos13} is the holomorphic twist of $N=1$ supersymmetric gauge theory in $4$ dimensions.  

We thus find that BCOV theory on $X$ is a deformation of holomorphic BF theory, and so a deformation of the holomorphic twist of Abelian $N=1$ supersymmetric gauge theory.

\subsection{}
Let us now compare this theory to the reduction of our $5$-dimensional non-commutative gauge theory ormation of the holomorphic twist of Abelian $N=1$ supersymmetric gauge theory. To start with, we will consider compactifying the topological direction $\R$ to a circle.  The fields of the theory then become 
$$
\Omega^\ast(S^1) \otimes \Omega^{0,\ast}(X) [1].
$$ 
If we replace the de Rham complex of the circle by its cohomology, we find 
$$
 \Omega^{0,\ast}(X) [\theta] [1].
$$
where the odd parameter $\theta$ comes from $H^1(S^1)$.  We can identify this with the fields of BCOV theory on $X$, by saying
\begin{align*} 
A &\in \Omega^{0,\ast}(X)[1] \\
B & \in \theta \Omega^{0,\ast}(X).  
\end{align*}
Therefore, the $4$-dimensional partial connection $A \in \Omega^{0,1}(X)$ just comes from the components of the $5$-dimensional connection transverse to the $S^1$, whereas the Lagrange multiplier scalar $B \in \Omega^{0,0}(X)$ is the component of the $5$-dimensional connection along the $S^1$. 

The reduction to $4$ dimensions of the $5$-dimensional action is, when $X = \C^2$, 
$$
S_{5d} = \tfrac{1}{\delta} \int B\dbar A + \tfrac{\eps}{2 \delta} B \partial A \partial A + \tfrac{\eps^2}{2\cdot 2^2 \cdot 2! \delta} \int  \varepsilon_{ij} \varepsilon_{kl} B \left(  \tfrac{\partial}{\partial z_i} \tfrac{\partial}{\partial z_k} A\right)\left( \tfrac{\partial}{\partial z_j} A \tfrac{\partial}{\partial z_l} A \right) + O(\eps^3). 
$$ 
The first two terms in the action match the action for Kodaira-Spencer theory.  The remaining terms give a deformation of this theory.  

This deformation is non-commutative holomorphic BF theory in $4$ dimensions. It is natural to guess that this deformation arises as the holomorphic twist of $4$-dimensional non-commutative $N=1$ gauge theory.

We thus have a hierarchy of three theories which are deformations of each other: holomorphic BF theory, Kodaira-Spencer theory, and non-commutative holomorphic BF theory. We can rewrite the Kodaira-Spencer action as the action
$$
\int B \dbar A + \tfrac{1}{2} \int B \{A,A\}
$$
where $\{-,-\}$ indicates the  holomorphic Poisson bracket. This formulation makes the relationship between the reduction of the $5d$ gauge theory and the Kodaira-Spencer theory clear: the reduction of the $5d$ gauge  theory is obtained by replacing the Poisson bracket $\{A,A\}$ by the Moyal commutator.

 \section{Holography for $M5$ branes: overview} 

Si Li and I have proposed a twisted version of the AdS/CFT correspondence using our notion of twisted supergravity. Here I will check this proposal for the case of $M5$ branes in the $\Omega$-background.  Our proposal in general is phrased in the abstract language of Koszul duality, but here I will avoid that language in the interests of keeping things concrete.

We will work with the specialization where $\delta = 0$. This means that the $5$-dimensional gauge theory can be treated classically, and that the equivariant parameters corresponding to the rotation of the two planes in the $M5$ brane world-volume sum to zero. 

%I will start by considering the situation where we have compactified $M$-theory on a Taub-NUT with no singularity at the origin. Thus, the $5$-dimensional gauge theory has gauge group $\mf{gl}_1$.  The reason for the specialization is that the algebra of operators living on a stack of $N$ $M5$ branes in this case is better understood: as we have seen earlier it is the $W_{N}$ algebra. In the large $N$ limit, with the specialization $\delta = 0$, this becomes the $W_{1+\infty}$ algebra. The case of an $M5$ brane wrapping an $A_k$ singularity will be discussed later, and related to the work of Maulik-Okounkov \cite{MauOko12}.

%In the seminal paper \cite{BeeRasvan15}, it is shown that $M$-theory that the $W_{1+\infty}$ algebra arises in the AdS dual description of the theory on an $M5$ brane. Thus, the fact that we will find the $W_{1+\infty}$ algebra coming out of our $5$-dimensional  gauge theory can be seen as further justification of the relationship between the $5$-dimensional gauge theory and $M$-theory.  

\subsection{}
There are a number of aspects of this story we need to discuss.  First, I need to explain precisely what we mean by the large $N$ limit of the algebra of operators on the $M5$ brane. I also need to explain something about the approach developed by Si Li and myself to interpret this large $N$ limit in terms of the gravity theory, which in this case is our $5$-dimensional non-commutative gauge theory.   I also want to explain the connection of the large-$N$ results derived here with the important results in the math literature due to Maulik-Okounkov \cite{MauOko12} and Schechtman-Vasserot \cite{SchVas13} stating that $W_{k+1+\infty}$ algebras act on the cohomology of the moduli of instantons (of arbitrary rank and charge) on an $A_k$ singularity.

\section{The large $N$ limit and $W_{k+1+\infty}$ algebras}
Let us start  by discussing how to take the large $N$ limit. Consider a gauge theory with gauge group $\mf{gl}(N)$.   At the classical level one can map the algebra of operators with gauge gruop $\mf{gl}(N+M)$ to the operators with gauge group $\mf{gl}(N)$.  A classical local operator is a gauge-invariant function of some $\mf{gl}(N+M)$-valued fields, such as scalar fields or the curvature of a connection.  One can use the block-diagonal embedding $\mf{gl}(N) \into \mf{gl}(N+M)$ to turn a function of the $\mf{gl}(N+M)$-valued fields into one of the $\mf{gl}(N)$-valued fields.

However, this does not work at the quantum level, as this map will not preserve the OPE of local operators (or indeed their anomalous dimensions).  The point is that the loop-level diagrams do not just involve traces of products of matrices in $\mf{gl}(N+M)$, but may also involve a trace of the identity matrix.  In the language of ribbon graphs, traces of the identity matrix arise from a face of a ribbon graph with no external lines.

Because there is no map of the operator product algebras, it is hard to make sense of taking the large-$N$ limit.  

However, one can solve this issue by using super-groups.  If we use the supergroup $\mf{gl}(N+M \mid M)$ as our gauge group, then if $M > M'$, there is a map at the quantum level from the space of operators of the $\mf{gl}(N+M \mid M)$ gauge theory to those of the $\mf{gl}(N+M'\mid M')$ gauge theory.  This map arises again by the block diagonal embeding $\mf{gl}(N+M'\mid M') \into \mf{gl}(N+M\mid M)$.    This map preserves all the structure present on local operators even at the quantum level. The point is that the trace of the identity in $\mf{gl}(N+M\mid M)$ is $N$, which is the same as the trace of the identity in $\mf{gl}(N+M'\mid M')$. Since we now have a map preserving all the structure, we can take the large $M$-limit of the algebra of operators, to yield the algebra of operators for a gauge group one could call $\mf{gl}(N+\infty \mid \infty)$.

The only effect that taking $M$ to infinity in this way has is that it removes all the trace relations from the definition of the space of local operators.   The $N$-dependence of the OPEs is unchanged. 

\subsection{}
In our discussion, we are attempting to take the large $N$ limit of the algebra of operators on a surface defect in Chern-Simons theory. In the case, there are fields living in the fundamental and anti-fundamental representations as well as in the adjoint representation.

Recall that on the surface defect, we have fields $\psi \in \Hom(\C^N, \C^{k+1})$ and $\psi' \in \Hom(C^{k+1}, \C^N)$, with action functional $\int \d z \op{Tr}_{\C^{k+1}} \psi \dbar \psi'$.  This is coupled to $GL(N)$ Chern-Simons theory. 

The supergroup version of this is that we have fields 
\begin{align*} 
\psi_{N+M\mid M} & \in  \Hom(\C^{N+M\mid M}, \C^{k+1})\\
\psi'_{N+M\mid M}& \in \Hom(\C^{k+1}, \C^{N+M\mid M}). 
\end{align*}
We then couple to $GL(N+M \mid M)$ Chern-Simons  theory.  Thus, the fields $\psi,\psi'$ are no longer purely fermionic, but contain $(N+M)k$ pairs of chiral free fermions and $Mk$ copies of a free $\beta-\gamma$ system.

As in the case when all the fields are adjoint valued, there are maps between the algebras of operators on the surface defect from the theory based on $(N+M\mid M)$ to that based on $(N+M' \mid M')$, whenever $M > M'$.  The map is defined using the decomposition 
$$
\C^{N+M\mid M} = \C^{N+M' \mid M'} \oplus \C^{M \mid M'}.
$$ 
This decomposition allows us to restrict a $GL(N+M\mid M)$-invariant function of the fields $\psi_{N+M\mid M},\psi_{N+M\mid M}'$ to a $GL(N+M'\mid M')$-invariant function of the fields $\psi_{N+M'\mid M'}, \psi'_{N+M'\mid M'}$.   These maps respect the operator product expansion between local operators. This allows us to take the large $M$ limit.

\subsection{}
Let us now see a bit more formally how this works.  We will only consider the case when $\delta = 0$, which means that the Chern-Simons theory in which the surface defect lives is treated classically. This means that the operators on the surface defect are just the $GL(N+M\mid M)$-invariant operators of the free theory with fields $\psi_{N+M\mid M}, \psi'_{N+M\mid M}$. We let $\mc{F}_{k+1}(N+M\mid M)$ be the vector space of $GL(N+M\mid M)$ invariant local operators. This space  has the structure of a vertex algebra.  There are maps
$$
\rho^{M}_{M'} : \mc{F}_{k+1}(N+M \mid M) \to \mc{F}_{k+1}(N+M' \mid M') 
$$
if $M > M'$, which preserve the vertex algebra structure.   We can take the large $M$ limit by considering sequences of operators $\alpha_M \in \mc{F}_{k+1}(N+M\mid M)$ satisfying
$$
\rho^{M}_{M'} \alpha_M = \alpha_{M'}.
$$
We let $\mc{F}_{k+1}(N+\infty \mid \infty)$ denote the large $M$ limit defined in this way.

We can write down explicitly a basis of local operators in the large $M$ limit. For every matrix $A \in \mf{gl}_{k+1}$, and every pair of non-negative integer $r, s$,  there is a $GL(N+M\mid M)$-invariant operator 
$$
\Oo(A,r,s)(z)  = \op{Tr}_{\C^{k+1}}  A \partial_z^r \psi(z) \partial_z^s \psi' (z).   
$$
These operators are not all primary, as 
$$
\partial_z \Oo(A,r,s) = \Oo(A,r+1,s) + \Oo(A,r,s+1). 
$$ 
We can then take normally ordered products of these operators and their derivatives to get a collection of operators of the form
$$
: \Oo(A_1,r_1,s_1)(z) \dots \Oo(A_m,r_m,s_m)(z) :  
$$
This normally ordered product is symmetrized so that it does not depend on the order in which the operators are written.  By invariant theory for $GL(N+M\mid M)$, operators of this form exhaust all possible $GL(N+M\mid M)$ invariant operators, and as $M \to \infty$ there are no trace relations among these operators.  

We thus find that there is an isomorphism from the vector space $\mc{F}_{k+1}(N+\infty \mid \infty)$ of local operators to the symmetric algebra $\Sym^\ast \mf{gl}_{k+1}[u,v]$ on the space of matrices whose entries are polynomials in variables $u$ and $v$.  The generators of the symmetric algebra are mapped to operators via
$$
A u^r v^s \mapsto \Oo(A,r,s)
$$ 
and the product of generators in this symmetric algebra get sent to normally ordered products of local operators. 
 
We will see later that the variables $u,v$ will be related by holography to the holomorphic coordinates on the space-time $\R \times \C^2$ of our $5$-dimensional gauge theory. The space $\mf{gl}_{k+1} [u,v]$ will be related to the fields of the gauge theory.

\subsection{}
One can calculate explicitly the OPE satisfied by these operators, and we will do so shortly.  We will find that it is a $W_{k+1+\infty}$ algebra.  Let us describe the $W_{k+1+\infty}$ algebra explicitly. 

% Before we do so, however, let us detour a little bit into a discussion of results in the math literature relating $W_{k+1+\infty}$-algebras and the moduli of instantons on $A_k$ singularities. This discussion will motivate the guess that the vertex algebra $\mc{F}_{k+1}(N+\infty \mid \infty)$ is the $W_{k+1+\infty}$-algebra with central charge $N$. We will then verify this guess explicitly, before finally explaining how one can predict this from holography.  

For any vertex algebra, one can construct an associative algebra generated from the contour integrals of local operators over a circle in a cylinder.    For a vertex algebra $V$, we will denote this associative algebra by $\oint V$.   For example, the associative algebra constructed in this way from the chiral WZW vertex algebra is the universal enveloping algebra of the Kac-Moody vertex algebra with a fixed value of the central extension.  In our approach to AdS/CFT, we will find it easier to compute the commutators in this associative algebra than to compute the OPEs. Of course, these two objects are very closely related, so there is minimal loss of information.

%Let us take our $5$-dimensional space-time to be $\R \times \C^\times \C$, with coordinates $t$ on $\R$, $z$ on $\C^\times$ and $w$ on $\C$. Let us give $\C^\times \times \C$ the holomorphic volume form $\d z / z \wedge \d w$.  Let's write $z = r e^{i \theta}$.  We will place a stack of $M5$ branes at $t = w = 0$, and consider the associative algebra built from integrating local operators on the $M5$ brane over the circles where $r$ is constant.  

The associative algebra $\oint W_{k+1+\infty}$ associated to the $W_{k +1+  \infty}$ vertex algebra has a simple description as the universal enveloping algebra of a certain Lie algebra.  Let us equip the algebra 
$$\Oo(\C^\times\times \C) = \C[z,z^{-1},w]$$
 of polynomial functions on $\C^\times \times \C$ with the non-commutative Moyal product $\ast_\eps$ we have used before.   If we want to emphasize the parameter $\eps$ we will write this algebra as $\Oo_\eps(\C^\times \times \C)$.  We will consider the Lie algebra with Lie bracket given by the commutator in the associative algebra $\Oo_{\eps}(\C^\times \times \C)\otimes \mf{gl}_{k+1}$.

This Lie algebra has a central extension, with cocycle given by an expression of the form 
$$
\omega(f,g) = \frac{1}{2 \pi i} \oint_{\abs{z} = 1, w = 0} \op{Tr}_{\C^{k+1}} f \partial g + O(\eps).
$$
The higher-order terms in $\eps$ are uniquely determined up to equivalence by the condition that they are $GL(k+1)$ invariant, involve a single trace in $\mf{gl}(k+1)$, and satisfy the cocycle condition. When restricted to the sub-Lie algebra $\Oo(\C^\times)\otimes \mf{gl}_{k+1}$, this is the usual Kac-Moody central extension. 
 
The associative algebra $\oint W_{k+1+\infty}(\eps,c)$ associated to the $W_{k+1+\infty}$ vertex algebra is 
$$
\oint W_{k+1+\infty}(\eps,c) = U_c ( \Oo_\eps(\C \times \C^\times) \otimes\mf{gl}_{k+1} ) 
$$
where $U_c$ means the quotient of the universal enveloping algebra of the central extension of $\Oo_\eps (\C \times \C^\times) \otimes \mf{gl}_{k+1})$ where the central parameter is set to $c$. Note that the dependence on $\eps$ in $W_{k+\infty}(\eps,c)$ can be removed by scaling the parameter $w$.  It is often convenient to keep the parameter $\eps$, however.  

When $k = 1$, the vertex algebra $W_{1 + \infty}(\eps,c = N)$ admits a quotient which is the $W_N$ algebra.   The vertex algebra $W_{1+\infty}(\eps,c)$ is a scaling limit of the $W_N$ algebras.

\begin{proposition}
As before, let $\mc{F}_{k+1}(N+\infty \mid \infty)$ be the vertex algebra obtained as the large $M$ limit of the $GL(N+M\mid M)$ invariants of a free $\beta-\gamma/b-c$ system built from $\Hom(\C^{k+1}, \C^{N+M \mid M})$.    Let $\oint \mc{F}_{k+1}(N+\infty \mid \infty)$ be the corresponding associative algebra. 

Then, there is an isomorphism of associative algebras
$$
\oint \mc{F}_{k+1}(N+\infty \mid \infty) \simeq \oint W_{k+1+\infty}(\eps = 1, c = N)
$$
to the $W_{k+1+\infty}$ algebra with central charge $N$.
\end{proposition}
In fact, this isomorphism of associative algebras arises from an isomorphism of vertex algebras. A small extension of the proof given here will show we find an isomorphism of vertex algebras. 
\begin{proof} 
As before, let 
\begin{align*} 
\psi &\in \cinfty(\C^\times, \Hom(\C^{k+1}, \C^{N+M\mid M}) \\
\psi'&\in \cinfty(\C^\times, \Hom( \C^{N+M\mid M}, \C^{k+1}).
\end{align*}
A set of generators for the algebra  $\oint \mc{F}_{k+1}(N+\infty \mid \infty)$ is given by the currents 
$$
\mc{C}\left(A z^m \partial_z^n \right) =  \oint_{\abs{z} = 1}\op{Tr}_{\C^{N+M\mid M}} \left( \psi  A  z^m \partial_{z}^n \psi'\right) \d z 
$$
for $A \in \mf{gl}(k+1)$.  

These currents give a map of vector spaces
$$
\mf{gl}_{k+1} \otimes \op{Diff}(\C^\times) = \mf{gl}_{k+1} \otimes \C[z,z^{-1}, w] \to \oint \mc{F}_{k+1}(N+\infty \mid \infty).  
$$
Here we can identify $w$ with $\dpa{z}$ since we are setting the parameter $\eps$ to $1$.

There is a PBW theorem stating that, if we choose an ordering for a basis of $\mf{gl}_{k+1} \otimes \op{Diff}(\C^\times)$, then ordered products of the currents form a basis for the algebra $\oint \mc{F}_{k+1}(N+\infty \mid \infty)$.  This follows from invariant theory for the supergroup $GL(N+M\mid M)$. 

We need to check that these currents satisfy the commutation relations of the central extension of the Lie algebra $\mf{gl}_{k+1} \otimes \op{Diff}(\C^\times)$.  It is immediate that there is some cocycle $\omega$ such that  
$$
\left[ \mc{C}\left(A z^m \partial_z^n \right), \mc{C}\left(B z^i \partial_z^j \right) \right] = \mc{C} \left(\left[ (A z^m \partial_z^n, B z^i \partial_z^j    \right]    \right) +  \omega \left( A z^m \partial_z^n,  B z^i \partial_z^j \right)\op{Id}  
$$
In this expression , we are using the commutator on the Lie algebra $\op{Diff}(\C^\times) \otimes \mf{gl}_{k+1})$, and $\omega$ denotes a cocycle on this Lie algebra.

The cocycle $\omega$ must involve a single trace over $\mf{gl}_{k+1}$, and we have seen that there is a unique up to scale cocycle of this form (modulo exact cocycles).  To determine the scale, we can calculate the cocycle on the currents of the form $\mc{C}(A z^i)$.  These currents form a copy of the affine Kac-Moody algebra, acting in the usual way on a $\beta-\gamma/b-c$ system built from the $GL(k+1)$ representation $\op{Hom}(\C^{k+1}, \C^{N+M\mid M})$.  The central charge of the free fermion system is $N$,  so that we find the $W_{k+2+\infty}$ algebra with central charge $N$.

\end{proof}

\subsection{}
This result implies that there are homomorphisms of associative algebras
$$
\oint W_{k+1+\infty}(\eps = 1, c = N) \to \oint \mc{F}_{k+1}(N) 
$$ 
for all $N$.  This map is always surjective, with kernel given by the trace relations that hold at finite $N$.  As $N \to \infty$, there are no trace relations, so that this map  becomes an isomorphism. 

We can thus view $ W_{k+1+\infty}(\eps,N)$ as the large $N$ limit of the algebra of operators on an $M5$ brane compactified on the Taub-NUT manifold $TN_k$ with an $A_k$ singularity at the origin, and placed in the $\Omega$ background.

There are some results in the mathematics literature \cite{MauOko12,SchVas13} which justify the idea that this large $N$ limit should be the $W_{k+1+\infty}$ algebra.

% More generally, Maulik-Okounkov show that this $W_{k+1+\infty}$ associative  algebra acts on the equivariant  cohomology of the moduli of instantons on an $A_k $ singularity, in a particular limit of the equivariant parameters which in our notation is when $\delta=0$.   This associative algebra is the chiral homology of the vertex algebra on the punctured disc.   

Recall that a class of local operators (called instanton operators) in a topological twist \footnote{ This twist is topological in the weak sense that all translations are $Q$-exact for the supercharge we use to twist. } of $5d $ maximally supersymmetric gauge theory act on the cohomology of instanton moduli. These instanton operators descend from operators in $6d $ $(2,0)$ theory. 

The algebra $\oint \mc{F}_{k+1}(N)$ is built the algebra of operators on  $N$ $M5$ branes in the $\Omega$ background, and so gives the algebra of local operators in the $5$-dimensional maximally supersymmetric gauge theory, in the $\Omega$-background.  An analysis of the supersymmetry tells us that the operators $\oint \mc{F}_{k+1}(N)$ live in the twist of the $5$-dimensional gauge theory related to the cohomology of instanton moduli space.  We conclude that the algebra $\oint \mc{F}_{k+1}(N)$ should act on the cohomology of the moduli space of instantons of rank $N$ on an $A_k$ singularity, and of arbitrary charge.
  
It turns out that the large-$N$ version of this statement has already been proven in the math literature,  by Maulik-Okounkov \cite{MauOko12} and Schiffmann-Vasserot \cite{SchVas13}. 
\begin{theorem*} 
The algebra $\oint W_{k+1+\infty}(\eps,N)$ acts on the equivariant cohomology of the moduli of instantons of rank $N$ on an $A_{k}$ singuularity.  Here $\eps$ is an equivariant parameter for the $\C^\times$ action on $\C^2 / \Z_{k}$ which preserves the holomorphic symplectic form.
\end{theorem*}
Since $\oint \mc{F}_{k+1}(N)$ becomes $\oint W_{k+1+\infty}(\eps,N)$ as $N \to \infty$, this fits with our story.  These considerations lead to the following conjecture.
\begin{conjecture}
The action of $\oint W_{k+1+\infty}(\eps,N)$ on the cohomology of the moduli of instantons on $\C^2/\Z_{k}$ factors through an action of $\oint \mc{F}_{k+1}(N)$, using the homomorphism
$$
\oint W_{k+1+\infty}(\eps,N) \to \oint \mc{F}_{k+1}(N)
$$
discussed above.  
\end{conjecture}

\section{ $W_{k+1+\infty}$ algebras at level $0$ from the $5$-dimensional gauge theory }
So far, we have argued that the $W_{k+1+\infty}$ algebra arises as the large $N$ limit of the algebra of operators on a stack of $M5$ branes in the $\Omega$ background.  We would like to see the same $W_{k+1+\infty}$ algebra from the gravity side, that is, from the $5$-dimensional non-commutative gauge theory.  It is certainly very plausible that one can do this, because both the $W_{k+1+\infty}$ algebra and the $5$-dimensional gauge theory are closely related to the algebra $\C[z,w]\otimes \mf{gl}_{k+1}$ of matrices whose entries are non-commutative functions on $\C^2$.

The approach to holography we will take is an example of a general method developed by Si Li and the author.  Before explaining the general approach, however, we will start by seeing how some very simple physical considerations allow us to derive the $W_{k+1+\infty}$ algebra from the gauge theory in the special  case  when $N = 0$.    This is the case when htere are the same number of $M5$ branes and anti-$M5$ branes.  Thus, the algebra of operators $\mc{F}_{k+1}(M \mid M)$ on the $M5$ branes is the $GL(M \mid M)$ invariants of the $\beta-\gamma/b-c$ system built from the super vector spaces $\Hom(\C^k, \C^{M \mid M})$. In this special case, the central parameter in the $W_{k+1+\infty}$ algebra is set to zero.  

In general, a stack of $M5$ branes gives rise to a t'Hooft operator in the $5$-dimensional gauge theory of charge $N\delta /\eps$.  In the simplified setting when $N=0$, the $(M\mid M)$ $M5$ branes do not introduce a singularity in the fields of the $5$-dimensional gauge theory.

Suppose that our $5$-dimensional gauge theory lives on $\R \times \C \times \C^\times$, with coordinates $t,w,z$, and with holomorphic symplectic form $\d w \d z$.  Let us place our $(M \mid M)$ $M5$ branes at $t = w = 0$.  This involves coupling the $\beta-\gamma/b-c$ system living on $\C^\times$ to the $5$-dimensional gauge field, and then taking the $GL(M\mid M)$ invariants.  

Suppose that we have a gauge field of the $5$-dimensional theory of the form
$$
A = f(\abs{z}) z^m w^n L \d \zbar
$$
where $n \ge 0$, $L \in \mf{gl}(k+1)$, and $f(\abs{z})$ is a bump function which is $0$ outside of a small neighourhood of the circle where $\abs{z} = 1$.  We will normalize $f$ so that $\int f(\abs{z}) \d \zbar \d z/z  = 2 \pi i$.   

 Let us suppose that the field $A$ satisfies the linearized equations of motion of the theory, namely
$$
(\d_{dR}^{\R} + \dbar^{\C \times \C^\times} ) A = 0. 
$$
Then, if $\alpha$ is a parameter of square zero, we can put the theory on the $M5$ branes in the background of the field $\alpha A$.  This introduces a first-order deformation of the theory on the $M5$ branes,  localized near the circle $\abs{z} = 1$.  Explicitly, this first-order deformation is given by the Lagrangian
$$
\mc{C}(A) = \int f(\abs{z})\op{Tr}_{\C^{M \mid M} } \left(  \psi L z^m \partial_z^n \psi' \right) \d z \d \zbar.   
$$
As before
 \begin{align*} 
\psi & \in  \Hom(\C^{k+1}, \C{M\mid M})\\
\psi'& \in \Hom(\C^{M\mid M}, \C^{k+1}). 
\end{align*}
are the fields of the $\beta-\gamma/b-c$ system on the $M5$ branes.

Now, the Lagrangian $\mc{C}(A)$ is changed to a BRST equivalent Lagrangian if we change $A$ by an infinitesimal gauge transformation. We can make such a gauge transformation which changes $A$ into the singular gauge field
$$
A' = \Delta_{\abs{z} = 1}  z^m w^n L \d \zbar.  
$$
The current $\mc{C}(A')$ associated to this singular gauge field is
$$
\mc{C}(A') = \oint_{\abs{z} = 1} \op{Tr}_{\C^{M \mid M} } \left(  \psi L z^m \partial_z^n \psi' \right) \d z \d \zbar.   
$$ 
Thus, $\mc{C}(A)$ is BRST equivalent to this current. This  current is one we wrote down earlier, which we saw generated the associative algebra $\oint \mc{F}(M \mid M)$. 

\subsection{}
Any gauge field $A$ of the $5$-dimensional gauge theory which satisfies the linearized equations of motion, and is supported on a small neighbourhood of the set $\R \times \C \times\{\abs{z} = 1\}$, gives rise to a current in the theory on the $M5$ branes.  Up to gauge equivalence, the only such gauge fields are the one we wrote down.  To see this, note that we can always apply a gauge transformation to make any such gauge field independent of both $t$ and $\d t$ (where $t$ is the coordinate on $\R$).  Once we do this, we find that such gauge fields modulo infinitesimal gauge transformations can be described by a certain relative Dolbeault cohomology group: they are elements of
$$
H^1_{\dbar} (\C^\times \times \C, U)\otimes\mf{gl}_{k+1}
$$ 
where $U \subset \C^\times \times \C$ is the open subset which is the complement of a neighbourhood of $\R \times \C \times \{\abs{z} = 1\}$.  (For example, we can take $U$ to be the region where $\abs{ \log \abs{z} } > 1$).  Relative Dolbeault cohomology is simply the cohomology of the subcomplex of the Dolbeault complex consisting of Dolbeault forms vanishing on the specified open subset. 

A simple exact sequence calculation tells us that a (topological) basis for this relative Dolbeault cohomology group is provided by $f(\abs{z}) z^m w^n \d \zbar$, for $m\in \Z$, $n \ge 0$, and where $f(\abs{z})$ as before is a bump function supported near $\abs{z} = 1$ and normalized so that $\int f(\abs{z}) \d \zbar \d \log z = 2 \pi i$.  

To sum up what we have found so far, we see that the currents generating the algebra $\oint \mc{F}(M \mid M)$ on the $(M\mid M)$ $M5$ branes are all realized by putting the $M5$ branes in a background gauge field of the non-commutative gauge theory which is supported near the set $\R \times \C \times\{\abs{z} = 1\}$.  Further, we can realize \emph{only} these currents in this way: there is a bijection between the currents of the form $\oint \psi L z^m \partial^n_z \psi'$ which generate the algebra $\oint \mc{F}(M\mid M)$, and $5$-dimensional gauge fields satisfying the equations of motion and which are supported near $\abs{z} = 1$, modulo gauge equivalence.   

\subsection{}
A similar analysis gives us a description of local operators in the theory on the $(M \mid M)$ $M5$ branes.  For this, we consider the $5$-dimensional gauge theory on $\R \times \C \times \C$.  We then consider solutions to the linearized equations of motion of the theory which are supported near the origin in the $z$-plane.  

To describe the space of such gauge fields, fix a bump function $f(\abs{z})$ which is $0$ when $\abs{z} > r$, for some small number $r$.  Let us normalize $f$ so that $\int f(\abs{z}) \d z \d \zbar = 1$.  Then, for any matrix $L \in \mf{gl}_{k+1}$, a solution to the linearized equations of motion is given by
$$
A = w^m \partial_{z}^n f(\abs{z}) L \d \zbar.
$$   
$A$ is gauge equivalent to the singular gauge field
$$
\til{A} = w^m \partial_z^n \Delta_{z = 0} L \d \zbar.
$$
As $L$, $m$ and $n$ vary, these gauge fields form a topological  basis for the space of solutions to the linearized equations of motion of the theory which are localized near $z = 0$, taken modulo gauge equivalences.  (One proves this by a calculation in relative Dolbeault cohomology).

The response of the theory on the $M5$ branes to the field $A$ is given by the Lagrangian
$$
\int \op{Tr}_{\C^{M \mid M}} \psi  L\partial_z^m \psi'  \dpa{z}^n f(\abs{z}) \d \zbar. 
$$
This is BRST equivalent to the local operator
$$
(-1)^n \partial_z^n \left(\op{Tr}_{\C^{M \mid M}} \psi L \partial_z^m \psi' \right)(z = 0). 
$$
As we have seen, normally ordered products of such local operators provide a basis for the space $\mc{F}_{k+1}(M \mid M)$ of the space of local operators, in the $M \to \infty$ limit.

Let's interpret the operators which have a single $\psi$ and a single $\psi'$ as being the analog of single-trace operators in a gauge theory where all the fields are adjoint valued. If we do so, we find that we have proved an analog of a basic statement in the AdS/CFT correspondence. We find that single-trace local operators in the theory on $(\infty \mid \infty)$ $M5$ branes are in bijection with solutions to the equations of motion of the gauge theory supported near a fixed value of $z$ (modulo gauge equivalence).

\subsection{}
So far, we have seen how we can realize  both local operators and currents in the theory  on $(M \mid M)$ $M5$ branes as a response to a gauge field.  Can we also see the commutation relations among currents (and the OPE between local operators) from the gauge theory? 

The answer is yes, and we will find precisely the $W_{k+1+\infty}$ algebra we saw before. We will start by analyzing the commutation relations between currents.

In this subsection, I will allow myself to use singular $5$-dimensional gauge fields which involve $\delta$-functions, to make the exposition easier. These fields are all are gauge equivalent to smooth gauge fields in which the $\delta$-function is replaced by a bump function, and the analysis can be done with these smooth gauge fields instead.

Consider the gauge fields
\begin{align*} 
A_{\abs{z} = c} &= \Delta_{\abs{z} = c} z^m w^n X \d \zbar\\
B_{\abs{z} = c} &=  \Delta_{\abs{z} = c} z^r w^s Y \d \zbar
\end{align*}
for matrices $X,Y \in \mf{gl}_{k+1}$, integers $m,r \in \Z$, and non-negative integers $n,s$. Also, $c$ is any positive real number.

There is a linearized gauge transformation relating $A_{\abs{z} = c}$ and $A_{\abs{z} = c'}$ for any $c < c'$, since 
$$
\left(\d_{dR}^{\R} + \dbar^{\C \times \C^\times} \right) \Delta_{c \le \abs{z} \le c'} z^m w^n X = A_{\abs{z} = c'} - A_{\abs{z} = c}.   
$$

Let's introduce parameters $\alpha,\beta$ of square zero.  By putting the theory on the $M5$ brane in the background with gauge fields $\alpha A_{\abs{z} = 1} + \beta B_{\abs{z} = 2}$, we get a two-parameter deformation of the theory on the $M5$ branes.  Since the parameters $\alpha,\beta$ are square zero, the only terms we see in this two-parameter deformation are the coefficients of $\alpha$, $\beta$, and $\alpha \beta$. 

We can view the Hilbert space of the theory on the $M5$ branes as being the local operators at $z = 0$.  Putting the theory in the background given by the gauge field $\alpha A_{\abs{z} = 1}$ has the effect of acting on the Hilbert space by $\exp(\alpha \mc{C}(A))$, where $\mc{C}(A)$ is the current associated to $A$.  Similarly, putting the theory in the background given by $\alpha A_{\abs{z} = 1} + \beta B_{\abs{z} = 2}$ gives rise to the operator $\exp (\beta \mc{C}(B)) \exp (\alpha \mc{C}(A))$.  Putting the theory in the background given by $\alpha A_{\abs{z} = 1} + \beta B_{\abs{z} = 1/2}$ gives rise to  $\exp (\alpha \mc{C}(A))\exp (\beta \mc{C}(B))$. 

Now, since we are working modulo $\alpha^2, \beta^2$, we have
$$
\exp (\alpha \mc{C}(A))\exp (\beta \mc{C}(B)) - \exp (\beta \mc{C}(B)) \exp (\alpha \mc{C}(A)) = \alpha \beta [\mc{C}(A), \mc{C}(B)].  
$$ 
Thus, to understand the commutator of the currents $\mc{C}(A)$ and $\mc{C}(B)$, we need to understand whether the deformation of the theory on the $M5$ branes obtained by putting them in the background  $\alpha A_{\abs{z} = 1} + \beta B_{\abs{z} = 2}$ is equivalent to that obtained by putting it in the background $\alpha A_{\abs{z} = 1} + \beta B_{\abs{z} = 1/2}$.  

Now, two gauge fields in $5$-dimensional gauge theory which are gauge equivalent give rise to equivalent deformations of the theory on the $M5$ brane. Thus, to understand the commutator of the currents, we can analyze whether the gauge fields    $\alpha A_{\abs{z} = 1} + \beta B_{\abs{z} = 1/2}$ and  $\alpha A_{\abs{z} = 1} + \beta B_{\abs{z} = 2}$ are gauge equivalent. 

We can try to write down a gauge transformation relating these two gauge fields by using the infinitesimal gauge transformation
$$
\gamma = \beta \Delta_{1/2 \le \abs{z} \le 2}   z^r w^s Y \d \zbar
$$
(where recall that $Y \in \mf{gl}_{k+1}$).  If we apply this gauge transformation to the gauge field  $ \alpha A_{\abs{z} = 1} + \beta B_{\abs{z} = 1/2}$, we find
\begin{align*} 
 \alpha A_{\abs{z} = 1} + \beta B_{\abs{z} = 1/2}& \mapsto \left(\dbar^{\C^{\times} \times \C} + \d_{dR}^{\R} \right)\beta \gamma + \alpha \beta [\gamma, A_{\abs{z} = 1} ] \\
&=   \alpha A_{\abs{z} = 1} + \beta B_{\abs{z} = 2} +  \alpha \beta [\gamma, A_{\abs{z} = 1} ] \\
&=  \alpha A_{\abs{z} = 1} + \beta B_{\abs{z} = 2} +  \alpha \beta \Delta_{\abs{z} = 1} \left[z^r,w^s Y, z^m w^n X \right]  
\end{align*}
where on the last line we have used the commutator in the associative algebra $\C[z,w] \otimes \mf{gl}_{k+1}$ which involves the Moyal product on $\C[z,w]$. 

This implies that the commutator of the  currents $\mc{C}(A)$ and $\mc{C}(B)$ is given by the commutator in the Lie algebra $\C[z,w] \otimes \mf{gl}_{k+1}$.  Therefore, we have derived, just by thinking about the $5$-dimensional non-commutative gauge theory, the current commutator in the algebra of currents living on $(M \mid M)$ $M5$ branes.

\subsection{}
We can phrase this result as an example of a general conjecture describing current algebras in the theory living on a brane in terms of the dual gravity theory.  This conjecture is a special case of a formulation of AdS/CFT developed by Si Li and the author.  

Suppose we are in a string or $M$-theory set-up involving a brane which does not source any fields in the ambient supergravity theory. Typically this happens if the theory on the brane involves a $GL(M\mid M)$ supergroup, so that we have $M$ branes and $M$ anti-branes \cite{DijHeiJefVaf16}.  Suppose the space-time manifold for the supergravity theory is of the form $\R \times S^k \times \R^l$, and that the brane lives on $\R \times S^k$.  Then, by compactifying on $S^k$ while including all KK modes, we can treat the theory on the brane as a quantum-mechanical system with some algebra of operators $\mc{A}(M \mid M)$.  

We want to express this associative algebra $\mc{A}(M \mid M)$ in terms of the ambient supergravity theory, in the limit when $M \to \infty$ and when we further take the planar limit.  (In the case of the theory on $(M \mid M)$ $M5$ branes, taking the planar limit amounted to setting $\delta = 0$). 

To do this, we observe that the equations of motion for the supergravity theory can be written as the Maurer-Cartan equation for some homotopy Lie (or $L_\infty$) algebra $\mc{L}$.  This is a standard part of the BV formalism: the space $\mc{L}^i$  of degree $i$ elements of the homotopy Lie algebra $\mc{L}$ is the collection of fields of the theory of ghost number $i-1$.  Thus, for instance, $\mc{L}^1$ consists of ordinary fields, $\mc{L}^0$ of ghosts, $\mc{L}^2$ of anti-fields, etc.  The $L_\infty$ operations encode the structure of the theory.  For instance, the $L_\infty$ maps
$$
\mc{L}^0 \times (\mc{L}^1)^{\otimes n} \to \mc{L}^1
$$
describe how the infinitesimal gauge symmetries in $\mc{L}^0$ act, in a non-linear way, on the fields in $\mc{L}^1$.   The $L_\infty$ operations
$$
(\mc{L}^1)^{\otimes n} \to \mc{L}^2
$$
encode the equations of motion of the theory. An element $\alpha \in \mc{L}^1$ satisfies the Maurer-Cartan equation
$$
\sum_{n \ge 1} \frac{1}{n!} l_n(\alpha,\dots,\alpha) = 0 
$$
if and only if it satisfies the equations of motion. 

In the  case of the $5$-dimensional non-commutative gauge theory on $\R \times \C^\times \times \C$, the $L_\infty$ algebra is the differential graded Lie algebra 
$$
\mc{L} = \Omega^\ast(\R) \what{\otimes} \Omega^{0,\ast}(\C^\times \times \C) \otimes \mf{gl}_{k+1}. 
$$
The differential is $\d_{\mc{L}} = \d_{dR}^{\R} + \dbar^{\C^\times \times \C}$, and the Lie bracket arises as the commutator for the associative product obtained by combining wedge product of forms, the Moyal product on $\C \times \C^\times$, and the product of matrices in$\mf{gl}_{k+1}$.  

This dgla is quasi-isomorphic to the Lie algebra
$$
\op{Hol}(\C^\times \times \C) \otimes \mf{gl}_{k+1}
$$
where the algebra of holomorphic functions on $\C^\times \times \C$ is equipped with the Moyal product.  Up to completion, we can identify this with $\C[z,z^{-1},w] \otimes \mf{gl}_{k+1}$ where $[z,w] = \eps$. 

In this case, the branes we are considering consist of $(M \mid M)$ $M5$ branes wrapping $\C^\times$.  We have seen that as $M \to \infty$ the algebra of operators on the quantum-mechanical system obtained by reducing the $M5$ brane theory on a circle is the universal enveloping algebra
$$
U( \C[z,z^{-1},w] \otimes \mf{gl}_{k+1} ). 
$$
(This is the $W_{k+1+\infty}$ algebra at central charge zero.)

Now, let us state the general conjecture of which this is a special case.
\begin{conjecture}
Suppose that we are in string or $M$-theory (after some partially topological twist, and maybe in an $\Omega$-background) on $\R \times S^k \times \R^{l}$.  Suppose we have a stack of branes wrapping $\R \times S^k$, and suppose that we have set things up so that the brane does not source any fields in the supergravity theory (for instance by considering the same number of branes and anti-branes). 

Let $\mc{L}$ be the $L_\infty$ algebra discussed above with the feature that Maurer-Cartan elements in $\mc{L}$ are solutions to the equations of motion of the supergravity theory on $\R \times S^k \times \R^{l}$.

Let $\mc{A}(M \mid M)$ be the associative algebra of operators on the quantum mechanical system obtained by compactifying the theory on the branes on $S^k$, including all KK modes.  Let $\mc{A}_{planar}$ be obtained from $\mc{A}(M \mid M)$  by sending $M \to \infty$ and taking the planar limit.

Then, there is an isomorphism (up to completion) 
$$
H^\ast ( \mc{A}_{planar} ) \iso U (H^\ast(\mc{L}) )  
$$
between the cohomology of $\mc{A}_{planar}$ (with respect to the BRST differential) and the universal enveloping algebra of the cohomology of the $L_\infty$ algebra $\mc{L}$.  
\end{conjecture}

\section{$W_{k+1+\infty}$ algebras from the $3$-dimensional reduction}
\label{section_dimensional_reduction_3d}
So far, we have explained how to calculate the algebra of operators on $(M \mid M)$ $M5$ branes, in the limit  $M \to \infty$ and when $\delta = 0$, in terms of the $5$-dimensional gauge  theory.  To go further in our analysis of holography, we need to generalize in two ways.  First, we need to understand in terms of the $5$-dimensional gauge theory why we find a central extension of the $W_{k+1+\infty}$ algebra when we have $(N+M \mid M)$ $M5$ branes. Secondly, we need to understand how quantum corrections to the gauge theory when $\delta \neq 0$ relate to the Yangian studied in \cite{MauOko12} which deforms the $W_{k+1+\infty}$ algebra. 

To see something about these richer aspects of the story, we need to take a somewhat different point of view.   The strategy is as follows.  Consider, as before, our theory on $\R \times \C \times \C$ with $M5$ branes living on $0 \times 0 \times \C$. Let us use coordinates $t,z,w$ where the $M5$ branes are at $t = w = 0$. We will choose a certain boundary condition when $t^2 + w\br{w} \to \infty$, which is defined by asking that all fields go to zero as $t^2 + w \br{w} \to \infty$. 

Then, the proposal is the following.  Consider the theory on $\R \times \C  \times \C$, with this boundary condition when $t^2 + w \wbar \to \infty$, and after removing the locus where $t^2 + w \wbar = 0$.  Let us also incorporate the flux associated to $N$ $M5$ branes.   Then, let us reduce this to a theory on $\C$ by projecting to the $w$ coordinate. 

\begin{conjecture}
The resulting theory on $\C$ is equivalent to the large $M$ limit of the theory of the theory on $(N+M\mid M)$ $M5$ branes.
\end{conjecture}
(Note that this theory has infinitely many fields). I have stated the conjecture slightly vaguely for now, because a precise statement (including, for example, details of the boundary condition) is a little lengthy.  The precise formulation is given in section \ref{section:quantum_holography}.  

Let me try to motivate this conjecture. Consider the $5$-dimensional gauge theory on $\R \times \C \times \C$, coupled to the theory on $(N+M\mid M)$ $M5$ branes.  Let's compactify this theory to the $z$-plane, using the boundary condition we have chosen at infinity in the $t-w$ directions.   We will now use an important feature of the boundary condition we consider: when we compactify in this way, all the fields of the $5$-dimensional gauge theory become infinitely massive, so that we can remove them. We are  thus left with the theory on $(N+M\mid M)$ $M5$ branes.  

When $M \to \infty$, we expect that the effect of coupling the gravitational theory to $(N+M \mid M)$ $M5$ branes has the effect of removing the location of the $M5$ branes (as well as introducing the field $F$ sourced by the branes).  This is a standard argument in AdS/CFT.   

Now, we can consider first, compactifying to $\C$ and then sending $M \to \infty$; or first, sending $M \to \infty$ and then compactifying to $\C$.  We expect that it doesn't matter in which order we perform these operations.  

This leads to the statement in the conjecture.  

\subsection{}
A first check of this proposal will be obtained by analyzing a compactification to $3$ dimensions.

Let $r^2 = t^2 + w \wbar$.   Let us identify
$$
\R \times \C \times \C \setminus \C \simeq S^2 \times \R_{> 0} \times \C, 
$$
where the two-sphere is those $(t,w)$ with $t^2 + w \wbar$ fixed, and the coordinates on $\R_{> 0} \times \C$ are $r,w$.

We can compactify our $5$-dimensional theory to $\R_{>0} \times \C$ along $S^2$, while retaining all KK modes.

The boundary condition we have chosen in $5$-dimensions at $r \to \infty$ gives rise to a boundary condition in $3$ dimensions.  Then, the conjecture we stated above leads  to the following statement.
\begin{conjecture}
    The algebra of operators on the boundary at $r = \infty$ of the $3$-dimensional gauge theory is the same as the algebra of operators of the theory on $(N+M\mid M)$ $M5$ branes in the $\Omega$-background as $M \to \infty$.  
    \end{conjecture}
    
\subsection{}
To get a feeling for this conjecture, let us write down this $3$-dimensional  theory explicitly.  Recall that the $3$-dimensional theory is  obtained from the $5$-dimensional theory on $(\R_t \times \C_w \setminus 0) \times \C_z$ by projecting along the map
\begin{align*} 
 (\R_t \times \C_w \setminus 0) \times \C_z & \mapsto \R \times \C_z \\
(t,w,z) & \mapsto ( (t^2 + w \wbar)^{1/2}, z). 
\end{align*} 

\begin{proposition}
\label{prop:compactification}
The field theory on $\R^+_r \times \C_z$ obtained  by compactifying the $5$-dimensional field theory on $S^2$ has space of fields consisting of a series of partial $1$-forms valued in $\mf{gl}_{k+1}$,  
$$
A^i = A^i_r \d r + A^i_{\zbar} \d \zbar
$$
for $i \ge 0$. These are conveniently arranged into a series
\begin{align*} 
A &= \sum_{i \ge 0} w^i A^i \\
A^i &= A^i_r \d r + A^i_{\zbar} \d \zbar   
\end{align*}
which  defines a one-parameter family of partial connections, or equivalently, a  partial connection for the Lie algebra $\mf{gl}_{k+1}[w]$.

In addition, we have a sequence of adjoint-valued scalar fields  $B^{i}$ for $i < 0$, which  can be  arranged into a series
$$
B = \sum_{i > 0} B^{-i} w^{-i}. 
$$ 
The action functional is (when we turn off the flux $F$ induced by the $M5$ brane) 
$$
\int_{r,z} \oint_{w} \d z \d w  B F(A) 
$$
where the curvature $F(A)$ is defined in the non-commutative sense
$$
F(A) = \d A + \tfrac{1}{2}[A,A]
$$
and $[A,A]$ is defined in the Moyal algebra where $z,w$ commute to $\eps$.  In this action, the contour integral over $w$ is treated as formal operation: we are not now treating $w$ as a coordinate on the space-time.

There is gauge symmetry generated by
\begin{align*} 
X &= \sum_{i \ge 0} X^i w_i \\
X^i  &\in \mf{gl}_{k+1} 
\end{align*} 
whose infinitesimal action on the space  of fields is given by
$$
A \mapsto A + \eta \left( \d \zbar \dpa{\zbar} X + \d r \dpa{r} X +[X,A] \right)  
$$
where $\eta$ is a parameter of square  zero, and the commutator $[X,A]$ is defined using the Moyal product.
 
Turning on the flux associated to $N$ $M5$ branes, according to the prescription given in section \ref{section:quantum_holography}, leads to the deformation of the action by adding a term 
$$
 \sum_{\substack{n \ge 1 \\ n \text{ odd} } } \frac{N  \eps^{n-1}}{  2^{n+1}n } \int_{z,r}\oint_w w^{-n}   \left( \partial_z^n A \right)  A \d z \d w.  
$$
This deformation also changes the gauge transformations of the field $B$.  The deformed infinitesemal gauge transformation is
$$
B \mapsto B + \eta [X,B] + N\delta  \sum_{\substack{n \ge 1 \\ n \text{ odd}}  } \frac{ \eps^{n-1} } {2^n n } \partial_z^n X \cdot w^{-n}. 
$$
(Here we are using that action of the ring $\C[w]$ on the space $w^{-1} \C[w^{-1}]$. Recall that $X$ is a function of $w$ as well as of the space-time coordinates).  
\end{proposition}
The proof of this result is placed in appendix \ref{appendix:3d_reduction}. It's useful to understand where the fields which accompany $w^n$ and $w^{-m}$ come from in terms of the origina $5d$ gauge theory. Since  we are compacfifying on $S^2$, the Kaluza-Klein modes can  be  described in terms of the eigenfunctions of the Laplacian on $S^2$.  The space  of smooth functions on $S^2$ decomposes into a sum of irreducible representations of $SO(3)$, where each irreducible representation appears exactly once.  Fields which are accompanied by $w^n$ come from the lowest-weight vector in the spin $n$ irreducible representation of $SO(3)$, living inside the space  of  functions on $S^2$; and fields accompanied by $w^{-n-1}$ corresponds to the highest weight vector in the same spin $n$ representations. Fields associated to other functions on $S^2$ do not  contribute, either because they are infinitely massive in an appropriate gauge, or because they can be gauged away.

This theory is related to Chern-Simons theory. To see this , consider the lowest-lying KK modes of the three-dimensional theory. These are $A_0$ and $B_{-1}$.  The action for these fields is
$$
\frac{1}{\delta} \int_{z,r} B_{-1} F(A_0) \d z + \frac{N}{2 }\int_{z,r} A_0 \dpa{z} A_0 \d z.  
$$
If we form a $3$-component connection
$$
\til{A} = A_0 + N^{-1}\delta^{-1} \d z B_{-1} 
$$
we find the action for the field $\til{A}$ is
$$
N \int CS(\til{A}). 
$$
Thus, if we only consider the fields $A_0,B_{-1}$, the theory is Chern-Simons theory for $\mf{gl}_{k+1}$ at level $N$. 

In addition, we have an infinite tower of adjoint-valued fields with an action including arbitrarily high number of derivatives.

\subsection{}
The boundary condition at $r = \infty$ is the one where the field $A$ is trivial.   Our goal is to relate the operators on the boundary to the $W_{k+1+\infty}$ algebra.  What we will find is that if we consider the tree-level contributions to the boundary OPE, we will find precisely the $W_{k+1+\infty}$ algebra which is isomorphic to the large $M$ limit of the theory on $(N+M\mid M)$ $M5$ branes.

As a warm-up, let's see why this is true  at the level of the lowest-lying fields of the $3d$ theory, namely $A_0$ and $B_{-1}$. Earlier we saw  that these combine into a connection
$$
\til{A} = \frac{1}{N\delta } B_{-1} \d z + A_0
$$

The boundary condition we are using is the one where, on the boundary, only the $\d z$ component of the connection $\til{A}$ survives, and we have broken all gauge symmetry on the boundary.  This is the boundary condition where the operators on the boundary are those of the chiral WZW model. This algebra of operators is the Kac-Moody algebra at level $N$. (Because we are considering only the contibution of tree-level diagrams, we don't see the shift of the level by the critical  level).  This is indeed part of the $W_{k+1+\infty}$ algebra.  

The claim is that the familiar argument that the Kac-Moody algebra lives at the boundary of Chern-Simons theory extends to the case we are considering, to show that the algebra of operators at the boundary is the full $W_{k+1+\infty}$ algebra.  We will verify this by explicit computation of boundary OPEs.

On the boundary only the field $B$ survives. For every matrix $M \in \mf{gl}_{k+1}$ and every  integer $n \ge 0$, one defines an operator living at a point $z_0$ on the boundary by 
$$
\mc{O}(M,n)(z_0) = - \oint_w \op{Tr}_{gl_{k+1}} w^n M B(z = z_0, r = \infty). 
$$
Thus, there are the correct number of operators to generate the $W_{k+1+\infty}$ algebra.  

To compute the OPE between these operators we'll use Feynman diagrams. The first step is to fix a gauge for the bulk theory, and then compute the fields sourced by the boundary operators in this gauge.  Let's change coordinates and write $t = 1/r$, so that the boundary is at $t = 0$. Since the theory is topological in the $r$ direction, the action has the same form in the $t$ coordinate as it did in the $r$ coordinate.  The gauge we will use on the bulk is the one where 
$$
\dpa{z} A_{\zbar} + \dpa{t} A_t = 0. 
$$
(Here $A_{\zbar}$, $A_t$ are the coefficients of $\d \zbar$ and $\d t$ in $A$).   

Given a boundary operator $\mc{O}(M,r)(z_0)$, we can solve the linearized equations of motion of the theory where the action has been deformed by the boundary operator. The quadratic and linear  part of the deformed action is
$$
\int_{z,t \ge 0} \oint_w \d z B \dbar A - \oint_w   \op{Tr}_{gl_{k+1}} w^r M B(z = z_0, t =0 ). 
$$
Varying $B$, we find that $A$ must satisfy the equation
$$
\d z \dbar A =  w^n M \Delta_{z = z_0, t = 0} . 
$$
$A$ must also satisfy $A\mid_{t = 0} = 0$, and the gauge-fixing condition we have chosen.  

The unique solution to these equations is 
$$
A(M,r) = \frac{1}{4 \pi} M w^r \left( \d \zbar \dpa{t} - \d t \dpa{z} \right) \left(t^2 + \abs{z - z_0}^2\right)^{-1/2}. 
$$
We will drop the normalization factor of $\frac{1}{4 \pi}$ here and in what follows.  It will not contribute to the result of the calculation, because the $W_{k+1+\infty}$ algebra only has a single parameter, the value of the central parameter. The value of the central parameter is determined by the level of the Kac-Moody algebra for $\mf{gl}_{k+1}$ we find, and we will verify this level by an analysis of the  copy of Chern-Simons theory that is the lowest lying fields of our theory. 

The field $A$ sourced by the boundary operator is the result of evaluating the boundary operator on the bulk-boundary propagator.  We can thus plug this field into Feynman diagram calculations for operator products of boundary operators.

The operator product between two boundary operators can be calculated in terms of scattering of bulk fields. The scattering process one needs to consider is between two fields $A$ sourced by currents, and some number of fields $B$.   The result will be an operator on the boundary which is a polynomial in the field $B$.  

Consider a general diagram that can appear in such a scattering process. It will have some number of external lines, two of which are labelled by the fields sourced by the boundary operator, and the remainder are labelled by the field $B$.  

\begin{figure}

%\begin{subfigure}

\begin{tikzpicture}
  \node[circle, draw](A) at (0,0) {$\mscr{O}_1$};
  \node[circle, draw](B) at (4,0) {$\mscr{O}_2$};
  \node[circle, draw](C) at (2,2) {$I$};
  \node (D) at (2,4){$A$};  
  \draw[thick](A) to [out=90,in=180](C);
  \draw[thick] (C) to [out=0, in = 90] (B);
  \draw[thick](C)--(D); 
\end{tikzpicture}
%\caption{\label{figure:walgebra_diagram}}
%\end{subfigure}

\vspace{10pt}

%\begin{subfigure}
\begin{tikzpicture}
  \node[circle, draw](A) at (0,0) {$\mscr{O}_1$};
  \node[circle, draw](B) at (4,0) {$\mscr{O}_2$};
  \node[circle, draw](C) at (2,2) {$N F$};
  
  \draw[thick](A) to [out=90,in=180](C);
  \draw[thick] (C) to [out=0, in = 90] (B);
 
\end{tikzpicture}
%\caption{\label{figure:kac_moody_extension}}
%\end{subfigure}
\vspace{10pt}

%\begin{subfigure}

\begin{tikzpicture}
  \node[circle, draw](A) at (0,0) {$\mscr{O}_1$};
  \node[circle, draw](B) at (6,0) {$\mscr{O}_2$};
  \node[circle, draw](C) at (2,2) {$I$};
  \node[circle, draw](D) at (4,2) {$I$};
  \draw[thick](A) to [out=90,in=180](C);
  \draw[thick] (C) to [out=45, in = 135] (D);
\draw[thick] (C) to [out=-45, in = -135] (D);
  \draw[thick](D) to [out=0, in =90](B); 
\end{tikzpicture}

%\end{subfigure}

\caption{ The three diagrams which can contribute to the boundary OPE. In each diagram $\mscr{O}_i$ are the operators we insert on the boundary. The second and third diagram contribute to the central extension, and together show that the value of the central charge is $N$ plus the critical level.} 

\label{fig:diagrams} 
\end{figure}
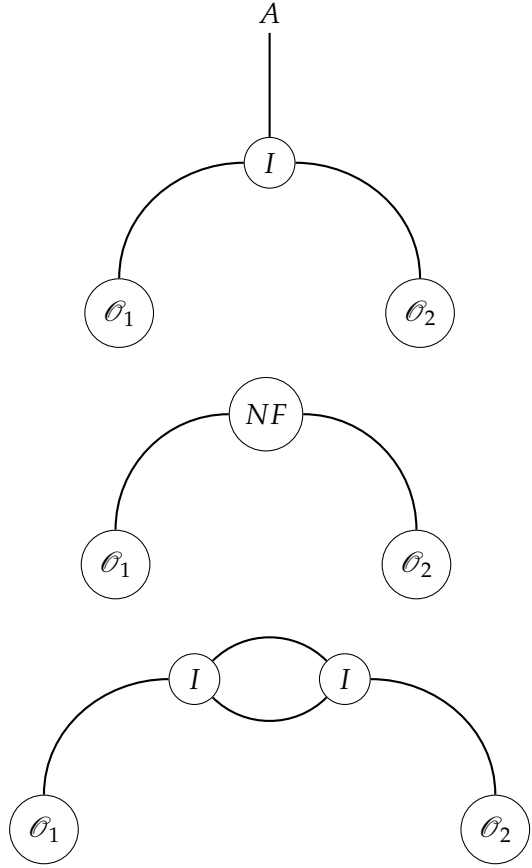

Three such scattering diagrams are illustrated in figure \ref{fig:diagrams}. 

We are interested in a certain limit of the OPE whereby only tree level diagrams contribute.  If we do this, it turns out only two diagrams can appear: these are the two tree-level diagrams illustrated in figure \ref{fig:diagrams}.  The reason for this is simple.  Every vertex in the diagram either takes two copies of the field $A$ and one copy of the field $B$, or is a bivalent vertex (associated to the flux $F$) which takes two copies of the field $A$.  Except for the two special external lines associated to the boundary operators, the external lines must be labelled by the field $B$.  The propagator connects the field $A$ to the field $B$. A simple combinatorial argument now shows that there are no possible tree-level diagrams except the two illustrated. 

In our conventions, the first diagram in figure \ref{fig:diagrams} is accompanied by $\delta$, and  the second diagram is accompanied by a factor of $N\delta^2$.  By rescaling the operators $\Oo(M,r)$ by $\delta$, we change the factors of $\delta$ so that the coefficient of each diagram does not involve $\delta$.

The following proposition explains the result of calculating the OPE in this way. 
\begin{proposition}
The tree-level limit of the algebra of operators on the boundary is the $W_{k+1+\infty}$ algebra with central charge $N$.
\label{proposition:tree_level_holography} 
\end{proposition}
\begin{proof}

Let us compute the contribution of each diagram to the OPE. Suppose the operators we insert are
\begin{align*} 
\mscr{O}_1 &= \mscr{O}(M_1,r_1)(z_1) \\
\mscr{O}_2 &= \mscr{O}(M_2,r_2)(z_2) 
\end{align*}
using the terminology above. These operators source fields
$$
A(M_i,r_i) = M_i w^{r_i} \left( \d \zbar \dpa{t} - \d t \dpa{z} \right) \left(t^2 + \abs{z - z_i}^2\right)^{-1/2}. 
$$
We find that, in the limit we are taking, the first diagram contributes the expression
\begin{multline*}
\Oo(M_1,r_1)(z_1) \cdot \Oo(M_2,r_2)(z_2) = \int_{z,t \ge 0} \oint_w B[A(M_1,r_1), A(M_2,r_2)] \d z \d w \\ 
+ \text{ contributions from the other diagram} 
\end{multline*}
where the commutator incorporates the Moyal product in the $z-w$ plane, as usual.   

The result of this integral will be a function of the field $B$ and the positions $z_1,z_2$.  Expanding in series in $(z_1 - z_2)^{-1}$ will yield the terms in the boundary OPE as functions of $B$.  We can assume that the field $B$ only depends on $z$ and not on $\zbar$ and $t$, because boundary operators only have derivatives in $z$. (In other words, we are analyzing scattering of on-shell states).  

Let us write this integral out explicitly, by writing $A(M_i,r_i) = M_i w^{r_i} \alpha(z_i)$ where $M_i$ is a matrix and $\alpha(z_i)$ the one-form
$$
\alpha(z_i) = \left( \d \zbar \dpa{t} - \d t \dpa{z} \right) \left(t^2 + \abs{z - z_i}^2\right)^{-1/2}. 
$$  
Note that the $z$-derivative of $\alpha(z_i)$ is the same as the $z_i$ derivative, so that we can write the Moyal commutator involving the $z_i$ derivative.  We find that 
\begin{multline*}
[A(M_1,r_1), A(M_2,r_2)] = [M_1,M_2]w^{r_1 + r_2} \alpha(z_1) \wedge \alpha(z_2) \\ + \eps (M_1 M_2 + M_2 M_1) w^{r_1 + r_2 - 1} \left( r_2 \dpa{z_1} \alpha(z_1) \wedge \alpha(z_2) - r_1 \alpha(z_1) \wedge \dpa{z_2} \alpha(z_2)   \right) + O(\eps^2).  
\end{multline*}
Continuing this pattern, we find that the Moyal commutator can be written as a sum of the form
\begin{multline*} 
 [A(M_1,r_1), A(M_2,r_2)] \\= \sum_{n,m \ge 0} \frac{\eps^{n+m} }{2^{n+m}}(-1)^m w^{r_1 + r_2 - n-m}\binom{r_1}{n} \binom{r_2}{m} \left( M_1 M_2 + (-1)^{n+m+1} M_2 M_1 \right) \partial_{z_1}^{m} \partial_{z_2}^n  \alpha(z_1) \wedge \alpha(z_2).   
\end{multline*}
Therefore, if $B = N f(z,w)$, the integral computing the OPE can be written as 
\begin{multline*} 
 \int_{z,t \ge 0} \oint_w B[A(M_1,r_1), A(M_2,r_2)] \d z \d w \\
=  \sum_{n,m \ge 0} \frac{\eps^{n+m} }{2^{n+m}}\binom{r_1}{n} \binom{r_2}{m} \left(\op{Tr} N M_1 M_2 + (-1)^{n+m+1}\op{Tr} N M_2 M_1 \right) \\
 \partial_{z_1}^m \partial_{z_2}^n \oint_w \int_{z,t \ge 0}  w^{r_1 + r_2 - n-m}  f(z,w)  \alpha(z_1) \wedge \alpha(z_2) \d z \d w .  
\end{multline*}
Thus, to compute the OPE, it suffices to compute the integral
$$
\int g(z) \alpha(z_1) \wedge \alpha(z_2) \d z .  
$$
for a holomorphic function $g$ of $z$.  We will show that result of this integral is
$$
c (z_1 - z_2)^{-1} +  F(z_1 - z_2)  
$$
for some non-zero constant $c$ and holomorphic function $F$ of $z_1 - z_2$.  

Explicitly, the integral is  \begin{multline*} 
     \int_{z,t \ge 0}  g(z) \left( \d \zbar \dpa{t} - \d t \dpa{z} \right) \left(t^2 + \abs{z - z_1}^2\right)^{-1/2}   \left( \d \zbar \dpa{t} - \d t \dpa{z} \right) \left(t^2 + \abs{z - z_2}^2\right)^{-1/2}  \d z \\
     =   \int_{z,t \ge 0} g(z)   \dpa{t}  \left(t^2 + \abs{z - z_1}^2\right)^{-1/2}   \dpa{z} \left(t^2 + \abs{z - z_2}^2\right)^{-1/2} \d t\d \zbar \d z  \\
     -  \int_{z,t \ge 0}  g (z) \dpa{z}  \left(t^2 + \abs{z - z_1}^2\right)^{-1/2}    \dpa{t} \left(t^2 + \abs{z - z_2}^2\right)^{-1/2} \d t\d \zbar \d z  
    \end{multline*}
We can assume, by translation invariance, that $z_1 = 0$. Performing the change of coordinates
\begin{align*} 
z &\mapsto z_2 z \\
t & \mapsto t \abs{z_2}  
\end{align*}
puts the integral in the form
\begin{multline*}
  \frac{1}{z_2} \int_{z,t \ge 0} g(z z_2)   \dpa{t}  \left(t^2 + \abs{z}^2\right)^{-1/2}   \dpa{z} \left(t^2 + \abs{z - 1}^2\right)^{-1/2} \d t\d \zbar \d z  \\
     - \frac{1}{z_2} \int_{z,t \ge 0}  g (z z_2 ) \dpa{z}  \left(t^2 + \abs{z }^2\right)^{-1/2}    \dpa{t} \left(t^2 + \abs{z - 1}^2\right)^{-1/2} \d t\d \zbar \d z  
    \end{multline*}
Expanding $g(z z_2)$ as a series in $z z_2$, we see that singular part in $z_2$ of the integral arises from replacing $g(z z_2)$ by the constant $g(0)$.  The non-singular part is some holomorphic function of $z_2$. Thus, we find the integral is of the form $c z_2^{-1} + F(z_2)$ where $F$ is holomorphic, as desired. The constant $c$ is   
\begin{multline*}
c =   \int_{z,t \ge 0} g(0)   \dpa{t}  \left(t^2 + \abs{z}^2\right)^{-1/2}   \dpa{z} \left(t^2 + \abs{z - 1}^2\right)^{-1/2} \d t\d \zbar \d z  \\
     -  \int_{z,t \ge 0}  g (0) \dpa{z}  \left(t^2 + \abs{z }^2\right)^{-1/2}    \dpa{t} \left(t^2 + \abs{z - 1}^2\right)^{-1/2} \d t\d \zbar \d z  
    \end{multline*}
One can check that this constant is a non-zero multiple of $g(0)$.  

This calculation shows that the contribution to the boundary of OPE of the first diagram in figure \ref{fig:diagrams} gives the OPE of the $W_{k+1+\infty}$ algebra with no central extension. The other diagram gives the  central extension.  The central extension is determined (up to equivalence) by its value on the copy of the Kac-Moody vertex algebra for $\mf{gl}_{k+1}$ living inside the $W_{k+1+\infty}$ algebra.  Therefore we need only calculate the central extension there. However, the answer here is immediate, and we don't even need to calculate. As, the Kac-Moody algebra comes from the operators on the boundary of the copy of Chern-Simons theory living inside our $3$-dimensional field theory, as the lowest-lying KK modes.  The action functional is that of Chern-Simons at level $N$.  Because we are  only considering tree-level diagrams when we compute the OPE, we do not find the shift by the critical level that appears in the relationship between the Kac-Moody algebra and Chern-Simons theory. This shift arises from the $1$-loop diagram in figure \ref{fig:diagrams}.  We therefore find the Kac-Moody algebra at level $N$.  

\end{proof}

\section{Holography beyond the planar limit}

\label{section:quantum_holography}
    To formulate holography at the quantum level, the description in terms of the $3$-dimensional theory obtained from $5$ dimensions by reducing along a sphere is not sufficient. Indeed, this $3$-dimensional theory has an infinite number of fields, and perturbation theory beyond the tree level is ill-defined in such a context.  To understand what happens at the quantum level, we need to work directly in $5$ dimensions. 

    To make the $5$-dimensional statement precise, we need to start by writing down the boundary condition in $5$ dimensions precisely. 

    Consider our $5$-dimensional gauge theory on $S^1 \times \mbb{P}^1\times \C$, where the $w$ plane has been compactified to $\mbb{P}^1$.    Let us modify our theory so that all of our fields are trivial at $\infty$. That is, both the gauge field $A$ and the infinitesimal gauge transformations are now twisted by the bundle $\Oo(-1)$ of functions on $\mbb{P}^1$ which vanish at $\infty$. One can check that, although the $2$-form $\d z \d w$ has a quadratic pole at $w = \infty$, the Lagrangian density 
    $$
    \tfrac{1}{2}\d z \d w \op{Tr} A \left(\dbar^{\mbb{P}^1 \times \C} + \d_{dR}^{\R} \right)  A  + \tfrac{1}{3}\d z \d w \op{Tr} A \ast_\eps A \ast_\eps A
    $$
    is regular at $w = \infty$.

    When we trivialize the fields at $\infty$ in this way, the space of fields with its linearized BRST operator is the cochain complex
    $$
    \Omega^\ast(S^1) \what{\otimes} \Omega^{0,\ast}(\mbb{P}^1 \times \C, \Oo(-1) ) \otimes\mf{gl}_{k+1}  
    $$ 
    with differential $\d_{dR}^{S^1} + \dbar^{\mbb{P}^1 \times \C}$.    This complex has no cohomology. Thus, the trivial solution to the equations of motion is rigid, and also has no gauge symmetries. If we work in perturbation theory, when we compactify the $5$-dimensional theory to $\C$ along $S^1 \times \mbb{P}^1$, we find the trivial two-dimensional theory. 

    Asking that the fields on $\R \times \C$ extend to $S^1 \times \mbb{P}^1$ in this way is the boundary condition I mentioned earlier.   

    Let's now remove the surface
    $$
    0 \times 0 \times \C \subset S^1 \times \mbb{P}^1 \times \C
    $$
    from the $5$-dimensional space-time.  The complement of this surface is a manifold
    $$
    X = S^1 \times \mbb{P}^1 \times \C \setminus 0 \times 0 \times \C. 
    $$ 
    There is a projection map $X \to \C$, and we can ``compactify\footnote{While compactifying along non-compact manifolds is not an operation often considered in the physics literature, it is a perfectly well-behaved operation in the language of factorization algebras \cite{CosGwi11}.  We will use the language of factorization algebras to implement these constructions precisely.}'' the theory on $X$ to produce a $2$-dimensional theory on $\C$ (with an infinite-dimensional space of fields).   

    The statement of holography we will make will relate this two-dimensional theory on $\C$ to the large $M$ limit of the theory on $(N+M\mid M)$ $M5$ branes, in the limit as $M \to\infty$.  To make a precise statement, we need to introduce the field on $X$ sourced by the t'Hooft operator living on the plane we have removed. 

    Let 
    $$
    F \in \Omega^\ast (S^1 ) \what{\otimes} \Omega^{1,\ast}( \mbb{P}^1, \Oo(-1) ) 
    $$  
    be the distributional $2$-form 
    $$
    F = \left(\d_{dR}^{S^1} + \dbar^{\mbb{P}^1}\right)^{-1} \Delta_{0 \times 0}.$$
    There is a unique such $F$ up to the addition of forms exact for $\d_{dR} + \dbar$,  just because the complex has no cohomology.  If we replaced $S^1$ by $\R$, we would find that $F$ has the expression we wrote down earlier  
    $$
    F =  \frac{\wbar \d t \d w + 2 t \d w \d \wbar   }{\eps (t^2 + w \wbar)^{3/2}  }  
    $$
    where $t,w$ are the coordinates on $\R$ and $\mbb{P}^1$.  The $2$-form we find on $S^1 \times \mbb{P}^1$ is a periodized version of this. 

    As we saw earlier, a field $A$ on $S^1 \times \mbb{P}^1 \times \C$ sourced by the t'Hooft operator will satisfy $\partial A = \frac{\delta}{\eps} F$. This is an equation we can solve locally.  However, we can mimic the effect of working in the background of such a field using only $F$. This is because for a field $A$ in the Abelian factor $\mf{gl}_1$ of the gauge Lie algebra $\mf{gl}_{k+1}$, the commutator $[A,-]$ only depends on $\partial A$.  Using the notation we introduced earlier, let $D(F)$ denote the derivation of the non-commutative algebra $\Omega^\ast(S^1) \what{\otimes} \Omega^{0,\ast}(\C \times \C)$  arising from $F$.  Then, we can add a term
    $$
    \int_{S^1 \times \mbb{P}^1 \times \C} \d z \d w \frac{N}{\eps} A D(F)(A) 
    $$  
    to the Lagrangian of our $5$-dimensional gauge theory. This term cancels the anomaly arising from the insertion of $N$ $M5$ branes at $0 \times 0 \times \C$.  

    Now, we are in a position to state the conjecture.  
    \begin{conjecture}\label{conjecture:holography}
    Consider the $5$-dimensional gauge theory on 
    $$
    X \subset S^1 \times \mbb{P}^1 \times \C
    $$
    with action
    $$
    \frac{1}{\delta} \int \d z \d w \op{Tr} \left(\tfrac{1}{2}  A \left( \d_{dR} + \dbar \right) A + \tfrac{N\delta }{2\eps} A D(F)(A) + \tfrac{1}{3} A \ast_\eps A \ast_\eps A   \right). 
    $$

    Let us view this as a theory on $\C$, by compactifying along the map $X \to \C$.   The algebra of operators of this theory on $\C$ are a vertex algebra, which we denote $\mc{F}_{grav}(\eps,\delta,N)$.  

    Let $\mc{F}_{M5}(N+M\mid M)(\eps,\delta)$ denote the vertex algebra of operators of the theory on $(N+M \mid M)$ $M5$ branes, with $\Omega$-background parameters $\eps$ and $\delta$.  
    Then, I conjecture the following statements hold:
    \begin{enumerate} 
    \item For $\delta \neq 0$, there is an isomorphism of vertex algebras
    $$
    \lim_{M \to \infty} \mc{F}_{M5} (N+M\mid M)(\eps,\delta) \iso \mc{F}_{grav} (\eps,\delta,N). 
    $$
    \item The current algebra
    $$
    \oint \mc{F}_{grav}(\eps,\delta, N)  
    $$  
    is isomorphic to the Yangian for affine $\mf{gl}_{k+1}$ studied by Maulik and Okounkov\cite{MauOko12}, which acts on the equivariant cohomology of the moduli of instantons of rank $N$ on an $A_k$ singularity.  
    \end{enumerate} 
    \end{conjecture}  
    \begin{remark}
    Implicit in this conjecture is the statement that the $5$-dimensional gauge theory on $S^1 \times \mbb{P}^1 \times \C$ exists at the quantum level.   The proof in the appendix that the theory quantizes does not extend to the case we consider here, where the fields are modified so that they are trivial at $\infty$ on $\mbb{P}^1$. 
    \end{remark}
   This conjecture is clearly a quantization of the statement we proved earlier \ref{proposition:tree_level_holography}. Indeed, the boundary condition at $\infty$ we have chosen in $5$-dimensions is a lift of the one we considered before in the $3$-dimensional compactification. Therefore the algebra of operators we considered in proposition \ref{proposition:tree_level_holography}, living at the boundary of the $3$-dimensional compactification, is simply the one obtained from the algebra denoted by by $\mc{F}_{grav}(\eps,\delta,N)$ in conjecture \ref{conjecture:holography} by only considering tree-level contributions to the OPE.    

I hope to return to  an explicit investigation of the algebra $\mc{F}_{grav}(\eps,\delta,N)$ in future.

\subsection{}
    Let me try to motivate this conjecture in terms of more standard holographic ideas.  Consider the $5$-dimensional gauge theory on $S^1 \times \mbb{P}^1 \times \C$, coupled to the theory on $(N+M\mid M)$ $M5$ branes.  Now, when we compactify the theory to $\C$, the $5$-dimensional gauge theory does not contribute at all, because this theory has no zero modes.  We are  thus left with the theory on $(N+M\mid M)$ $M5$ branes.  

    When $M \to \infty$, we expect that the effect of coupling the gravitational theory to $(N+M \mid M)$ $M5$ branes has the effect of removing the location of the $M5$ branes (as well as introducing the field $F$ sourced by the branes).  

    Now, we can consider first, compactifying to $\C$ and then sending $M \to \infty$; or first, sending $M \to \infty$ and then compactifying to $\C$.  We expect that it doesn't matter in which order we perform these operations.  

    This leads to the statement in the conjecture.  

    \appendix

    \section{An $11$-dimensional supergravity background}
\label{appendix:sugra}
    In this appendix we will prove the following theorem, stated in section \ref{section:11dsugra_omega}.  Suppose, as explained there, that we have an $11$-manifold of the form $M \times Y$ where $M$ is a manifold with a $G2$ structure and $Y$ is a $4$-dimensional hyper-K\"ahler manifold.  Let $\psi_M$ denote the covariant constant spinor on $M$ coming from the $G2$ structure. Choose one of the complex structures on $Y$, and let $\psi_Y$ denote a spinor corresponding to this complex structure. This has the feature that for every complex $1$ form $\lambda$ on $Y$ which is $(1,0)$ for our given complex structure, we have $\lambda \cdot \psi_Y = 0$ where $\cdot$ indicates Clifford multiplication. Let $\omega^{2,0}_Y$ and $\omega^{0,2}_Y$ denote the holomorphic and anti-holomorphic volume forms on $Y$ coming from the hyper-K\"ahler structure.

    Let $V$ be a Killing vector field on $M$ which also preserves the $G2$ structure (meaning that Lie derivative with respect to $V$ preserves the $3$-form on $M$ defining the $G2$ structure).   Let $V^{\flat}$ denote the dual $1$-form, and let $\omega_V = \d V^{\flat}$. Note that we can also  view $\omega_V$ as being the covariant derivative of $V$ or of $V^{\flat}$, because the fact that $V$ is a Killing vector field means that its covariant derivative is anti-symmetric. 
    \begin{theorem*}
    Consider the $11$-manifold $M \times Y$ equipped with the $4$-form
    $$
    F = \omega_V \wedge \omega^{0,2}_Y
    $$
    where $\omega^{0,2}_Y$ is the complex conjugate of the holomorphic symplectic form on $Y$. We will view $F$ as a flux defining an $M$-theory background. Then,
    \begin{enumerate} 
    \item $M \times Y$ with the flux $-\eps F$ satisfies the equations of motion of $11$-dimensional supergravity (where $\eps$ is a constant). 
    \item The spinor 
    $$\Psi(\eps)=  \psi_M \otimes \psi_Y + \eps (X \cdot \psi_M) \otimes (\omega^{0,2}_Y \cdot \psi_Y)$$
    is a generalized Killing spinor for the supergravity background with flux $-\eps F$.  Here, $\cdot$ indicates Clifford multiplication. 
    \item The square of $\Psi(\eps)$ is the Killing vector field $\eps V$ on $M \times Y$.  
    \end{enumerate}
    \end{theorem*}
    \begin{proof}
    The equations of motion of $11$-dimensional supergravity with closed $4$-form flux $F$ are the following. We require that \begin{align*} 
     \d \ast F + \tfrac{1}{2} F \wedge F & = 0\\ 
    \op{Ric}(\alpha,\beta) &= \tfrac{1}{2}\ip{\alpha\vee F, \beta \vee F  } -\tfrac{1}{6} g(\alpha,\beta) \ip{F,F}  
    \end{align*}
    for all vector fields $\alpha,\beta$ on $M \times Y$. 

    In our situation, most of the terms that appear in these equations vanish. We have
    \begin{align*} 
     \d \ast F &= 0 \\
    F \wedge F &= 0 \\
    F \wedge \ast F &= 0 \\
    \op{Ric}(\alpha,\beta) &= 0.
    \end{align*}
    Note that we are using the complex-linear extension of the Hodge star operator to forms with complex coefficients.  With this Hodge star operator, the form $\omega^{0,2}_Y$ on $Y$ is self-dual, which is why $F \wedge \ast F = 0$ and $\d \ast F = 0$.  This reduces the field equations to the equation
    $$
    (\alpha \vee F) \wedge \ast (\beta \vee F) = 0
    $$ 
    for all vector fields $\alpha,\beta$.  If both $\alpha,\beta$ are parallel to $M$, then this is clear from the fact that $\ast \omega^{0,2}_Y = \omega^{0,2}_Y$ and that $\left( \omega^{0,2}_Y\right)^2 = 0$.  Next suppose that one of the vector fields, say $\alpha$, is parallel to $Y$, and the other is parallel to $M$. Then $\alpha \vee F$ is a $(0,1)$ form on $Y$ tensored with a $2$-form on $M$.  For any $(0,1)$ form $\lambda$ on $Y$, we have $\lambda \wedge \ast \omega^{0,2}_Y = 0$.  Finally, if both $\alpha,\beta$ are parallel to $Y$, then the statement follows from the fact that for any two $(0,1)$ forms $\lambda,\lambda'$ on $Y$ we have $\lambda \wedge \ast \lambda' = 0$.

    Next let us check that $\Psi(\eps)$ is a a generalized Killing spinor for the $M$-theory background with flux $-\eps F$. We need to check that for all vector fields $\alpha$ on $M \times Y$, we have
    $$
    \nabla_\alpha \Psi(\eps) - \eps (\alpha \vee F) \cdot \Psi(\eps) = 0
    $$
    where $\cdot$ indicates Clifford multiplication.  

    If $\alpha$ is parallel to $Y$, then since $\psi_Y$ and $\omega^{0,2}_Y \cdot \psi_Y$ are both covariant constant on $Y$, we have $\nabla_\alpha \Psi(\eps) = 0$.  We also have in this case that 
    $$(\alpha \vee F )\cdot (X \cdot \psi_M) \otimes (\omega^{0,2}_Y \cdot \psi_Y) = 0$$
    simply because for any $(0,1)$ form $\lambda$ on $Y$, $\lambda \cdot \omega^{0,2}_Y \cdot \psi_Y = 0$.     

    It is less obvious that, if $\alpha$ is parallel to $Y$, we have
    $$
    (\alpha \vee F) \cdot (\psi_M \otimes \psi_Y) = 0.
    $$    
    To see this, note that $\alpha \vee F = \omega_V \otimes \lambda$ where $\lambda = \alpha \vee \omega^{0,2}_Y$.  Thus, 
    $$
    (\alpha \vee F) \cdot (\psi_M \otimes \psi_Y) = (\omega_V \cdot \psi_M) \otimes (\lambda \cdot \psi_Y).  
    $$
    Now, we have the following lemma.
    \begin{lemma}
     $\omega_V \cdot \psi_M = 0$.  
    \end{lemma}
    \begin{proof}
    We can see this as follows.  Recall that we require that the isometry $V$ of $M$ preserves the $3$-form defining the $G2$ structure on $M$. Recall also that the associated bundle to the frame bundle of $M$ for the adjoint representation of $\op{Spin}(7)$ is the bundle of $2$-forms on $M$.  Further, the Clifford multiplication action of $2$-forms on spinors arises from the action of the Lie algebra $\op{spin}(7)$ on spinors by passing to the associated bundles to the frame bundle. 

    Since we have a reduction of structure group to $G2$, the adjoint representation of $G2$ on the Lie algebra $\mf{g}_2$ gives us a subbundle of the bundle of $2$-forms on $M$.  We will call a $2$-form which is a section of this bundle a $G2$ $2$-form.    Since, in a local frame, the spinor $\psi_M$ is preserved by the Lie algebra $\mf{g}_2$, it is also preserved by Clifford multiplication by any $G2$ $2$-form. 

    To show that $\omega_V \cdot \psi_M = 0$, it suffices to show that $\omega_V$ is a $G2$ $2$-form.  Now, $G2$ $2$-forms are those that, when viewed as endomorphisms of the tangent bundle, preserve the $3$-form $A$ defining the $G2$ structure on $M$.  Equivalently, this means that the $2$-form $\omega_V$ is $G2$ if the $3$-form defined by 
    $$
    \omega_V \ast A = \sum g_{rs} \omega_V^{ir} A^{sjk}
    $$ 
    (i.e.\ the contraction of $\omega_V$ with $A$) is zero.  

    We will check this explicitly. By the definition of $\omega_V$, the $3$-form $\omega_V \ast A$ is the antisymmetrization of the  multilinear function which sends a triplet $\alpha,\beta,\gamma$ of vector fields to 
    $$
    (\nabla_\alpha V) \vee \beta \vee \gamma \vee  A 
    $$
    where $\vee$ indicates contraction. 

    Equivalently, we can view the $3$-form $\omega_V \ast A$ as the anti-symmetrization of 
    $$
    (\nabla V) \vee A
    $$ 
    where $\nabla V$ is viewed as a $1$-form valued in vector fields, and we contract the vector field part with the $3$-form $A$. 

    Now, since $A$ is parallel, 
    $$ 
     \nabla (V \vee A) = (\nabla V) \vee A.  
    $$
    Since the antisymmetrization of the covariant derivative of a form is the exterior derivative, we find that 
    $$
    \omega_V \ast A = \d (V \vee A).
    $$
    Finally, the Cartan homotopy formula together with the fact that $A$ is closed tells us that
    $$
    \d (V \vee A) = \mc{L}_V A
    $$
    where on the right hand side $\mc{L}_V A$ is the Lie derivative.  By assumption, this is zero, so that $\omega_V \ast A = 0$ and $\omega_V$ is a $G2$ $2$-form as desired.
    \end{proof}

    So far, we have proved that if $\alpha$ is parallel to $Y$, then 
    $$
    \nabla_\alpha \Psi(\eps) - \eps (\alpha\vee F) \cdot \Psi(\eps) = 0.
    $$ 
    Next, let us check that this is the case if $\alpha$ is parallel to $M$.  Recall that
    $$
    \Psi(\eps) = \psi_M \otimes \psi_Y + \eps (V \cdot \psi_M) \otimes (\omega^{0,2}_Y \cdot \psi_Y). 
    $$
    Now, $\nabla_\alpha \psi_M = 0$ if $\alpha$ is a vector field on $M$, but
    $$
    \nabla_\alpha (V \cdot \psi_M) = (\nabla_\alpha V) \cdot \psi_M
    $$
    since $\psi_M$ is covariant constant.  Further, by the definition of $\omega_V$, we have
    $$
    (\alpha \vee \omega_V) \cdot \psi_M = (\nabla_\alpha V) \cdot \psi_M.
    $$
    Putting this together, we find that if $\alpha$ is parallel to $M$, the following identities hold.
    \begin{align*} 
     \nabla_\alpha (\psi_M \otimes \psi_Y)  &= 0 \\
    (\alpha \vee F) \cdot (\psi_M \otimes \psi_Y)&= \nabla_\alpha ( V \cdot \psi_M) \otimes (\omega^{0,2}_Y \cdot \psi_Y)\\
    (\alpha \vee F) \cdot \left\{(V \cdot \psi_M) \otimes (\omega^{0,2}_Y \cdot \psi_Y)\right\} &= 0  
    \end{align*}
    These tell us that 
    $$
    \nabla_\alpha \Psi(\eps) - \eps (\alpha\vee F) \cdot \Psi(\eps) = 0
    $$
    as desired.

    To complete the proof, we need to compute the square of $\Psi(\eps)$.    This is a calculation we can do pointwise in a local frame.  Thus, let us discuss the calculation in terms of representations of the group $G2 \times SU(2)$.  As before, let $S_7$ denote the spin representation of $\op{Spin}(7)$, $S_4$ that of $\op{Spin}(4)$, which decomposes as $S_4 = S_+ \oplus S_-$.  The spinors $\psi_M$ on $M$ and $\psi_Y$ on $Y$ come from elements $\psi_7 \in S_7$ and $\psi_4 \in S_-$ preserved by $G2$ and $SU(2)$ respectively.  Further, the spinor $\omega^{0,2}_Y \cdot \psi_Y$ comes from a linearly independent element $\psi_4' \in S_-$.  The representation $S_-$ of $\op{Spin}(4)$ is equipped with a symplectic form, and we have $\omega(\psi_4,\psi'_4) = 1$. 

    Let us now consider the Lie bracket on the $11$-dimensional supersymmetry algebra. This is a map
    $$
    \Sym^2 ( S_7 \otimes S_4) \to \C^{11}. 
    $$   
    which takes two spinors and produces a complexified translation.  We are interested in the bracket of $\psi_7 \otimes \psi_4$ with $V \cdot \psi_7 \otimes \psi_4'$, for a vector $V \in \C^{11}$. Now, 
    $$
    \Sym^2 (S_7 \otimes S_-) =  (\Sym^2 S_7) \otimes (\Sym^2 S_-) \oplus (\wedge^2 S_7) \otimes (\wedge^2 S_-).
    $$
    Now, there is no $\op{Spin}(4)$-equivariant map from $\Sym^2 S_-$ to the $4$-dimensional vector representation, so the term involving $\Sym^2 S_-$ can not contribute to the  bracket.

    However, $\wedge^2 S_-$ is canonically trivialized by the symplectic form on $S_-$, and the $7$-dimensional $\Gamma$-matrices give a (unique up to scale) $\op{Spin}(7)$-equivariant map
    $$
    \wedge^2 S_7 \to \C^{7}.
    $$
    Thus we find a unique $\op{Spin}(7) \times \op{Spin}(4)$ equivariant map
    $$
    \Sym^2 (S_7 \otimes S_-) \to \C^7 \subset \C^{11}.
    $$
    This map therefore must, up to a non-zero constant, be the bracket on the $11$-dimensional supersymmetry algebra. 

    It follows that, for a vector $V \in \C^{7}$, we have
    $$
    [\psi_7 \otimes \psi_4, (V \cdot \psi_7) \otimes (\psi'_4) ] = \Gamma_7 ( \psi_7 , V \cdot \psi_7 ) \omega(\psi_4, \psi'_4)
    $$
    where $\Gamma_7 : \wedge^2 S_7 \to \C^7$ is the $7$-dimensional $\Gamma$-matrices map, and $\omega$ is the symplectic pairing on $S_-$.  

    The representation $S_7$ of $\op{Spin}(7)$ has a symmetric non-degenerate inner product.  Since the $7$-dimensional $\Gamma$-matrices map is unique up to scale, it follows that
    $$
    \Gamma(\psi_7, V \cdot \psi_7) = V \ip{\psi_7, \psi_7}
    $$
    up to a non-zero constant.  This proves the result, up to a non-zero constant which is convention dependent and can be absorbed into a redefinition of the parameter $\eps$. 
\end{proof}
    \subsection{}
    Next, let us prove another theorem stated in section \ref{section:11dsugra_omega}, showing how in this situation certain variations of the metric and $C$-field are exact for the supercharge we are using.

    Let us consider the family Riemannian manifolds where the underlying manifold is $M \times Z$ and where the metric on $M$ is $r^2 g_M$.  For each value of $r$, we will write down a family of $C$-fields and generalized Killing spinors which square to the vector field generating the $S^1$ action.  To do this, we first need to introduce a little notation which will allow us to compare spinors on $M$ for different values of $r$.  Let $S_M$ be the $8$-dimensional bundle of spinors on $M$.  The definition of $S_M$ does not depend on the scale of the metric.  We will take the convention that the bundle map
    $$
    \Gamma : S_M \otimes S_M \to TM 
    $$  
    given by the $\Gamma$ matrices does not depend on the scale of the metric. Then, if we fix a metric on $M$, we get a metric on $S_M$ characterized by the property that if we take the tensor product metric on $S_M \otimes S_M$, and restrict to the sub-bundle $TM$, we find the Riemannian metric on $TM$.  If we scale the inner product on $TM$ by a factor of $r^2$, we scale that on $S_M$ by a factor of $r$.

    Note that this is not the only reasonable convention for how to identify the spin bundles for different values of $r$. One could also declare, for instance, that the Clifford multiplication map is independent of $r$, in which case the $\Gamma$ matrices will depend on $r$.   

    With the convention we choose, the Clifford multiplication map 
    $$
    TM \otimes S_M \to S_M 
    $$
    is defined so that 
    $$
    \Gamma(X \cdot \psi,\psi') = X \ip{\psi,\psi'}. 
    $$
    The Clifford multiplication map then depends on the scale $r$ of the metric.  If $\cdot_r$ indicates Clifford multiplication with the metric $r^2 g_M$, then we have
    $$
    X \cdot_r \psi = r X \cdot_1 \psi.  
    $$

    Let $\psi_M$ denote the covariant constant spinor present on $M$ by virtue of the $G2$ structure.  We choose $\psi_M$ so that it is of norm one in the metric coming from $g_M$.  Let $V$ be the vector field generating the $S^1$ action on $M$.  Then, $r^{-1/2} \psi_M$ is of norm $1$ in the metric coming from $r^2 g_M$.   The family of spinors
    $$
    \Psi(\eps,r) = r^{-1/2} \psi_M \otimes \psi_Z + \eps r^{1/2} (V \cdot_1 \psi_M) \otimes (\omega^{0,2} \cdot \psi_Z) 
    $$
    is a family of generalized Killing spinors for $11$-dimensional supergravity on $M \times Z$ where the metric on $M$ is $r^2 g_M$ and the $3$-form field is 
    $$
    C(\eps,r) = -\eps r^2 V^{\flat} \omega_Z^{0,2}.
    $$
    Here the dual $1$-form $V^{\flat}$ is defined using the metric $g_M$.

    A variant of this construction which will be important for the theorem is obtained by adding on to $C(\eps,r)$ a closed $3$-form.  Closed $3$-forms do not affect the equations of motion of the  theory, nor do they affect the generalized Killing spinor equation.  The particular closed $3$-form which will be important  for us is $A_{G2}$, the $G2$ $3$-form on the manifold $M$. As we scale the metric on $M$ by $r^2 g_M$, the $G2$ form scales as $r^3 A_{G2}$.   

    Let us now consider the $3$-parameter of supergravity backgrounds equipped with a generalized Killing spinor given by: 
    \begin{align*} 
     \Psi(\eps,r,\alpha) &= r^{-1/2}\alpha  \psi_M \otimes \psi_Z + \eps r^{1/2}\alpha^{-1} (V \cdot_1 \psi_M) \otimes (\omega^{0,2} \cdot \psi_Z)\\ 
    C(\eps,r,\alpha) &= -\eps r^2 \alpha^{-2} V^{\flat} \omega_Z^{0,2} + c r^3 A_{G2}\\ 
    g &= r^2 g_M \boxtimes g_Z.  
    \end{align*}
    The spinors $\Psi(\eps,r,\alpha)$ are generalized Killing spinors which square to $\eps V$. We treat $\eps,r,\alpha$ as parameters and treat $c$ as a constant, whose value will be fixed by supersymmetry considerations.

    We now have a $3$-parameter family of $\Omega$-backgrounds for $11$-dimensional supergravity. The theorem states that these parameters are not independent: instead, there is a cohomological relation among them. 

    \begin{theorem}
    For some non-zero value of the constant $c$, the action of the vector field $r \dpa{r} + \tfrac{1}{2} \alpha \dpa{\alpha}$ on this family of supergravity solutions is $Q$-exact,  where $Q$ indicates the supercharge associated to the family of spinors $\Psi(\eps,r,\alpha)$. 
    \end{theorem}
    This theorem implies that if we move along a trajectory for this vector field in the parameter space, we do not change the twisted version of $M$-theory we construct.  In particular, if we set $\alpha = r^{1/2}$, so that the size of the supercharge is $r^{1/2}$, we find that after twisting, everything is independent of $r$.  
    \begin{proof}
    As a first step,  we will construct a gravitino field $\delta \lambda(r,\alpha)$ which  satisfies the equations of motion modulo $\delta^2$ and such that applying the supercharge we use  to $\delta \lambda(r,\alpha)$ yields the $r$-derivative of the family of metrics $r^2 g_M$.  Applying supersymmetry to the gravitino field $\lambda$ will also yield some $C$-field, which will show that we can trade scaling of the metric for variation of the $C$-field. 

    The metric $g_M$ is a section of $\Sym^2 T^\ast M$.  There is a map, given by Clifford multiplication, of bundles on $M$
    $$
    \mu_r : \Sym^2 T^\ast M \otimes S_M \to T^\ast M \otimes S_M  
    $$
    Since Clifford multiplication depends on $r$, this map depends on $r$. We will evaluate this map at $r = 1$, and just call it $\mu$ instead of $\mu_1$.

    In coordinates, this map sends
    $$
    \sum f_{ij}\d x_i \d x_j \psi \mapsto \sum f_{ij} \d x_i \gamma^j \psi
    $$
    where $\gamma^j$ is the $\gamma$-matrix generating the action of the Clifford algebra. 

    Recall that $\psi_M$ is a covariant constant spinor on $M$ associated to the $G2$ structure.  Define a spinor-valued $1$-form
    $$
    \lambda_M(r,\alpha)  = 2 \alpha^{-1}r^{1/2}  \mu ( g_M \otimes \psi_M)  \in \Omega^1(M, S_M). 
    $$
    Note that $\lambda_M$ is covariant constant, because $g_M$ and $\psi_M$ are, and the operation of Clifford multiplication respects the Levi-Civita connection.

    Then, define a $1$-form valued spinor on $M \times Z$ by
    $$
    \lambda(r,\alpha) = \lambda_M(r,\alpha) \otimes (\omega^{0,2} \cdot \psi_Z). 
    $$
    This is again covariant constant, because it is a tensor product of covariant constant objects on $M$ and on $Z$.

    In the absence of a $C$-field, the linearized equation of motion for a gravitino field states that a certain contraction of its covariant derivative vanishes. Thus, any covariant constant gravitino field satisfies the linearized equation of motion, and in particular $\lambda(r,\alpha)$ does.   

    If we turn on the $C$-field, there is an extra term in the equations of motion which is a certain contraction of $F = \d C$ and $\lambda$. With the $C$-field we use, this extra term is zero. Recall that $F = \omega_V \otimes \omega^{0,2}_Z$. In a local frame on $Z$, we can analyze how $F$ and $\lambda$ transform under $SU(2)_+$ (where we take the convention that the holonomy of $Z$ is in $SU(2)_-$). We see that $F$ has weight $+2$ under the Cartan of $SU(2)_+$ and $\lambda$ has weight $+1$.  Any objected constructed by contractiosn of $F$ and $\lambda$ will have weight $+3$. However, the highest weight under this Cartan that appears in the tensor product of the vector representation of $\op{Spin}(4)$ with the spin representation is $+2$, so the contraction of $F$ and $\lambda$ must be zero.  Therefore, $\lambda$ satisfies the equations of motion.

    Let 
    $$\Gamma_r : S \otimes S \to T^\ast M \times Z$$
    denote the map obtained from the $11$-dimensional $\gamma$-matrices.  The $r$-dependence issuch that $\Gamma_r = r^2 \Gamma_1$.   The action of local supersymmetry on the fields of supergravity sends a variation of the gravitino to a variation of the metric, by the map 
    $$
    \op{Id}\otimes \Gamma_r :  T^\ast M \otimes S \otimes S  \mapsto T^\ast M \otimes T^\ast M  
    $$
    followed by symmetrization. (In the physics literature, this formula would be written as something like
    $$
    \partial g_{ij} = \sum \lambda_{i}^{\alpha} \gamma^{i}_{\alpha \beta} \eps^{\beta}
    $$
    where Roman letters are space time indices, and Greek letters are spinor indices and $\eps$ is the generator of the supersymmetry).

    It is clear from this that the action of the supercharge 
    $$r^{-1/2} \alpha \psi_M \otimes \psi_Z$$
     on $\lambda(r,\alpha)$ produces the metric $2 r^2 g_M$ on $M$, as desired.  Further, the action of 
    $$\Psi_1 = (V \cdot \psi_M) \otimes (\omega^{0,2} \cdot \psi_Z)$$
    on $\lambda$ is zero.    
      
    We conclude that the operation of applying $r \dpa{r}$ to the family of metrics $r^2 g_M$ is made $Q$-exact by the gravitino field $\lambda$.

    We are not quite done, however, because the action of supersymmetry on a gravitino can also produce a $3$-form field.  In the physicists notation, the formula for this would be something like
    $$
    \delta C_{ijk} = \lambda_{i}^{\alpha}\eps^{\beta} \Gamma^{jk}_{\alpha \beta}
    $$
    where $\Gamma^{jk}$ is a certain product of $\gamma$-matrices. The invariant meaning of $\Gamma^{jk}$ is that, if $S_{11d}$ denotes the the spin representation of $\op{Spin}(11)$ and $V_{11d}$ denotes the vector representation, we have a decomposition into irreducible representations
    $$
    \Sym^2 S_{11d} = V_{11d} \oplus \wedge^2 V_{11d} \oplus \wedge^5 V_{11d}.
    $$
    The matrix $\Gamma^{jk}$ arises from the map 
    $$\Gamma^{(2)}:  S_{11d}^{\otimes  2} \to \wedge^2 V_{11d}$$
    which projects onto the second summand of the decomposition of $\Sym^2 S_{11d}$. 
    This map also gives rise to the central extension corresponding to the $M2$ brane.  

    The $3$-form field associated to a spinor $\psi$ and a gravitino $\lambda$ comes from applying the map
    $$
    \op{Id} \otimes \Gamma^{(2)}: V_{11d} \otimes S_{11d}^{\otimes 2} \to V_{11d} \otimes \wedge^2 V_{11d} 
    $$ 
    followed by the antisymmetrization map. 

    Our task is now to understand this $3$-form field.  Let us local coordinates $x_1,\dots,x_7$ on the $G2$ manifold $M$. The $3$-form field is 
    \begin{multline*}
    \sum g_{ij}  \d x_i \Gamma^{(2)}\left( (\d x_j \cdot \psi_M)\otimes (\omega^{0,2}_Z \cdot \psi_Z), \psi_M \otimes \psi_Z  \right) \\
    +  \eps \sum g_{ij}  \d x_i \Gamma^{(2)}\left( (\d x_j \cdot \psi_M)\otimes (\omega^{0,2}_Z \cdot \psi_Z), (V \cdot \psi_M) \otimes (\omega^{0,2} \cdot \psi_Z)  \right) . 
    \end{multline*} 
    Here, as above, $V$ is the vector field on $M$ generating the $S^1$ action. 

    Let $A_{G2}\in \Omega^3(M)$ denote the $3$-form associated to the $G2$ structure.  The symmetries of the problem tell us that for any vector field $W$ on $M$ we have
    \begin{align*} 
    \Gamma^{(2)} \left(  ( W  \cdot \psi_M) \otimes (\omega^{0,2}_Z \cdot \psi_Z), \psi_M \otimes \psi_Z  \right) &= c_1 W \vee A_{G2}\\
     \Gamma^{(2)}\left( ( W \cdot \psi_M) \otimes (\omega^{0,2}_Z \cdot \psi_Z), (V \cdot \psi_M) \otimes (\omega^{0,2} \cdot \psi_Z)  \right) &= c_2 g(W,V) \omega^{0,2}_Z 
    \end{align*}
    for some non-zero constants $c_1$, $c_2$.   
     
    From this, and calculations of how everything depends on the scale $r$, we see that applying the supercharge
    $$
    \Psi(\eps,r,\alpha) = r^{-1/2}\alpha  \psi_M \otimes \psi_Z + \eps r^{1/2}\alpha^{-1} (V \cdot_1 \psi_M) \otimes (\omega^{0,2} \cdot \psi_Z) 
    $$ 
    to 
    $$
    \lambda(r,\alpha) = 2 \alpha^{-1}r^{1/2}  \mu ( g_M \otimes \psi_M) \otimes (\omega^{0,2}_Z \cdot \psi_Z) 
    $$
    yields, as well as the metric, the $3$-form
    $$
    2 c_2 \alpha^{-2} \eps r^2 V^{\flat}  \omega_Z^{0,2} + 2 c_1 r^3 A_{G2} = 2c_2 C(r,\eps,\alpha) + 2 c_1 r^3 A_{G2}.  
    $$
    Here $A_{G2}$ refers to the $G2$ $3$-form for the metric at $r = 1$. 

    Since the $3$-form $A_{G2}$ is closed, it does not effect the supergravity equations of motion.  Let us consider, momentarily, our supergravity background as being specified by the closed $4$-form $F = \d C$, instead of by the $3$-form field.  If we work in this way, we find that the application of the supersymmetry to $\lambda(r,\alpha)$ produces the $4$-form $2 c_1 \d C(r,\eps,\alpha)$ as well as the variation of the metric we discussed above.

    Let us write down again explicitly our $3$-parameter family of supergravity backgrounds, but in terms of $F$ rather than $C$:
    \begin{align*} 
     \Psi(\eps,r,\alpha) &= r^{-1/2}\alpha  \psi_M \otimes \psi_Z + \eps r^{1/2}\alpha^{-1} (V \cdot_1 \psi_M) \otimes (\omega^{0,2} \cdot \psi_Z)\\ 
    F(\eps,r,\alpha) &= -\eps r^2 \alpha^{-2}\d  V^{\flat} \omega_Z^{0,2} \\
    g &= r^2 g_M \boxtimes g_Z.  
    \end{align*}
    Applying the vector field $r \dpa{r} + \tfrac{1}{2} \alpha \dpa{\alpha}$ will produce a first order variation of  a given supergravity background equipped with a generalized Killing spinor.    This vector field preserves the spinor $\Psi(\eps,r,\alpha)$.  Further, the variation of the metric under this vector field is $Q$-exact: it is made exact by the gravitino $\lambda(r,\alpha)$.    Supersymmetry applied to the gravitino also produces a variation of $F$.    

    The variation of $F$ under this vector field must be $Q$-exact. If not, then by adding the $Q$-exact term generated by $\lambda(r,\alpha)$, we would find the first-order variation of the supergravity solution is equivalent to one where we only vary $F$, and not $g$ or $\Psi$.  This is not possible, because $\Psi$ remains a generalized Killing spinor when we  vary the supergravity background.    

    These considerations fix the value of the constant $c_2$ so that the action of vector field $r \dpa{r} + \tfrac{1}{2} \alpha \dpa{\alpha}$ on our family of supergravity backgrounds is $Q$-exact.  

    Now let's again consider the supergravity background as being defined by a $3$-form instead of by its field-strength.  Consider the family of supergravity backgrounds  
    \begin{align*} 
     \Psi(\eps,r,\alpha) &= r^{-1/2}\alpha  \psi_M \otimes \psi_Z + \eps r^{1/2}\alpha^{-1} (V \cdot_1 \psi_M) \otimes (\omega^{0,2} \cdot \psi_Z)\\ 
    \til{C}(\eps,r,\alpha) &= -\eps r^2 \alpha^{-2}\d  V^{\flat} \omega_Z^{0,2} + \frac{c_1}{3} r^3 A_{G2} \\ 
    g &= r^2 g_M \boxtimes g_Z.  
    \end{align*}
    We find that the action of the vector field $r \dpa{r} + \tfrac{1}{2} \alpha \dpa{\alpha}$ on this family of supergravity backgrounds is $Q$-exact, as desired.  
    \end{proof}

    \section{Quantization of the theory}
    \label{appendix_quantization}
In this appendix we will prove the theorems stated in section \ref{section:quantization} concerning quantization of the $5d$ gauge theory. 
    \begin{theorem*}
    Let $X$ be a conical complex symplectic surface, and consider our gauge theory on $\R \times X$ with gauge group $\mf{gl}_N$. Let us consider quantizing the theory in perturbation theory in the loop expansion parameter $\delta$, and let us work ``uniformly in $N$'' as discussed in detail in \cite{CosLi15}.  Let us also ask that the quantization is compatible with the $\C^\times$-action on $X$. Finally, let us allow negative powers of $\eps$ as long as they are accompanied by positive powers of $\delta$. Heuristically, this means that we should treat $\eps$ and $\delta$ as both being small but with $\delta$ much smaller than any positive power $\eps^n$ of $\eps$.  

    Then, there are no obstructions (i.e.\ anomalies) to producing a consistent quantization.  At each order in $\delta$, we are free to add $\op{dim} H^2(X) + 1$ independent terms to the action functional, leading to an ambiguity in the quantization.   Explicitly, the deformation of the action functional we can incorporate at $k$ loops is
     \begin{align*} 
     S & \mapsto S  + \delta^{k}\eps^{-k}\left( \tfrac{1}{2} \int \alpha \d \alpha + \tfrac{1}{3} \int \alpha \ast_\eps \alpha \ast_\eps \alpha\right)\\
     S & \mapsto S +   \delta^k \eps^{-k}\int \alpha \ast_{\eps} \alpha \ast_{\eps,w} \alpha 
    \end{align*}
    where $\ast_{\eps,w}$ indicates the first-order deformation of the $\ast$-product on $X$ associated to a class $w \in H^2(X)$. 

    \end{theorem*}

    \begin{proof}
    Throughout the proof we will use the notation introduced in section \ref{section_5d_gauge_theory}. Using the techniques of \cite{CosLi15}, we see that the obstruction-deformation complex describing quantizations of this theory is built from the cyclic cohomology of the algebra $J \mc{A}$ of jets of sections of $\mc{A}$. This is an associative algebra in $D$-modules on $\R  \times X$.   Let $J \Oo_{X}$ denote the $D$-module on $X$ of jets of holomorphic functions.  Then, the $D$-module $J \mc{A}$ is, at $\eps = 0$, quasi-isomorphic to $\pi^\ast J \Oo_{X}$, where $\pi : \R  \times X\to X$ is the projection map.  Let us first analyze the obstruction-deformation complex at $\eps = 0$, and then analyze the $\eps$-dependent terms in the differential. 

    The cyclic chain complex of $J \Oo_{X}$ is quasi-isomorphic to the dg  $D$-module 
    $$J \Omega^{-\ast}_{hol}[t^{-1}][1] = \oplus_{k, i \ge 0 } t^{-k} J \Omega^i_{hol} $$ 
    where $t^{-k} J \Omega^i_{hol}$ is situated in degree $-2k-i-1$. The differential is $t \partial$ where $\partial$ is the de Rham differential. 

    At $\eps = 0$, the obstruction-deformation complex to quantizations at $k$-loops is 
    $$
    \omega_{\R  \times X} \otimes_{D} \Oo ( \pi^\ast J \Omega^{-\ast}_{hol} [t^{-1}][1] )[N] [-6(k-1)]$$ 
    where in this equation $\Oo(M)$ means the $D$-module of functions on a $D$-module $M$ modulo constants, that is, the completed symmetric algebra of the dual with the zeroth symmetric power removed. Also the symbol $[N]$ indicates polynomials in a parameter of degree $0$ called $N$.  The point is that a Lagrangian in a $\gl_N$ gauge theory can involve polynomials in $N$ while still being defined ``uniformly in $N$'' in the sense of \cite{CosLi15}.   

    Next, let us introduce the parameter $\eps$, and as we explained in the statement of the theorem, allow $\eps$ to be negative as long as the negative powers are accompanied by positive powers of $\delta$.  This implies that when analyzing the obstruction-deformation complex at $k$ loops we find instead the cyclic cohomology of the algebra $J \Oo_{X} ((\eps))$, which is viewed as an associative algebra over $\C((\eps))$ in the category of $D$-modules on $X$, and where the product is the $\ast$-product on holomorphic functions on $X$.   We take cyclic cohomology linearly over the base field $\C((\eps))$.  

    The cyclic homology of $J \Oo_{X} ((\eps))$ is quasi-isomorphic to the complex of $D$-modules
    $$
    J \Omega^{-\ast}_{hol} [t^{-1}] ((\eps)) [1]
    $$ 
    which has two terms in the differential, $t \partial$ and $\eps \mc{L}_{\pi}$, where
    $$
    \mc{L}_{\pi} : \Omega^i_{hol} \to \Omega^{i-1}_{hol}
    $$
    is the operator defined by Lie derivative with repect to the holomorphic Poisson tensor $\pi$ on $X$. (More precisely, if $\iota_{\pi}$ indicates contraction with $\pi$, then $\mc{L}_{\pi} = [\partial, \iota_{\pi}]$.) The operator $\mc{L}_{\pi}$ is sometimes called the differential on Poisson cohomology.

    Now, the cohomology of $J \Omega^\ast_{hol}$ with respect to the differential $\mc{L}_{\pi}$ is the trivial $D$-module $\cinfty_{X}$, which is realized as the kernel of $\mc{L}_{\pi}$ inside $J \Omega^2_{hol}$.  It follows that the cyclic homology of $J \Oo_{X}((\eps))$ is quasi-isomorphic to
    $$
    \cinfty_{X} [t^{-1} ] ((\eps)) [3]. 
    $$  
    Therefore, the cyclic cohomology is the $\C((\eps))$-linear dual of this, and the obstruction-deformation complex of the theory, at $k$ loops, is described in terms of the symmetric algebra of dual.  We find that the obstruction-deformation complex at $k$ loops is is 
    $$
    \omega_{\R  \times X} \otimes^{\mbb L}_{D} \left( \cinfty_{\R  \times X} \otimes \left(\C[ \eta_3,\eta_5,\dots]\right)/\C  \right)[N] ((\eps)) 
    $$
    where the parameters $\eta_k$ are in degree $k$, and we quotient the algebra $\C[\eta_3,\eta_5,\dots]$ by the constants $\C$.  

    Note that 
    $$
    \omega_{\R \times X} \otimes^{\mbb L}_{D} \cinfty_{\R \times X} \simeq \Omega^\ast(\R \times X) [5] 
    $$
    (one can see this by using a standard resolution of the right $D$-module $\omega_{R \times X})$).  It follows that the obstruction-deformation complex is quasi-isomorphic to
    $$
    \Omega^\ast(\R \times X) \otimes \left(\C [\eta_3,\eta_5,\dots] / \C\right) [N]((\eps))[5] 
    $$
    Thus, without imposing any further symmetries, the cohomology of the obstruction-deformation is the following:
    $$
    \begin{array}{l l} 
    \text{ degree } 0 & \eta_5 H^0(X)[N]((\eps)) \oplus \eta_3 H^2(X)[N]((\eps)) \\
    \text{ degree} 1 & \eta_5 H^1(X)[N]((\eps)) \oplus \eta_3 H^3(X)[N]((\eps)) \\
    \text{ degree} -1 & \eta_3 H^1(X)[N]((\eps)). 
    \end{array}
    $$
    Order by order in $\delta$, degree $0$ elements correspond to deformations, that is, to terms we can add on to the Lagrangian, degree $1$ elements to obstructions, and degree $-1$ elements to symmetries.

    Now let us require that the quantization be compatible with the $\C^\times$ action on $X$ which scales the symplectic form. If this $\C^\times$ action scales the symplectic form with weight $k$, then to make the action functional $\delta^{-1} S$ scale invariant, we need to scale the parameters $\eps$ and $\delta$ with weight $k$ too.  We also find that the action on the obstruction complex scales $H^\ast(X)$ with weight $k$ as well.  Thus, at $l$ loops, where we are considering the coefficient of $\delta^{l-1}$, we find that the $\C^\times$-invariant terms are accompanied by $\eps^{-l}$.  Thus, imposing $\C^\times$-invariance reduces the obstruction-deformation complex by replacing each occurence of $H^i(X)[N]((\eps))$ by $H^i(X)[N]\eps^{-l}$.

    It follows that, assuming $H^1(X) = H^3(X) = 0$, the obstruction to quantization vanishes, and at each order in $\delta$, we can add on at most $\op{dim} H^{ev}(X)$ $N$-dependent terms to the Lagrangian.  We will explicitly write the possible terms later. 

    Next, let us discuss the case that $X = T^\ast \Sigma$ where $\Sigma$ is a Riemann surface.  In this case there is a potential obstruction, coming from $H^1(\Sigma) \eta_5$.      We can use a general formal argument to show that this obstruction vanishes.  Suppose it vanishes to order $n$ in $\delta$. Let $O_{n+1}(X) \in H^1(X)$ be the obstruction at order $n+1$.  Since obstructions are local, for any open subset $U \subset X$,  the class $O_{n+1}(U) \in H^1(U)$ is the restriction of $O_{n+1}(X)$ to $U$.  To show that $O_{n+1}(X)$ vanishes, it suffices to show that it does on every open set of $X$ of the form $T^\ast A$ where $A \subset \Sigma$ is an annulus.  Thus, it suffices to show that $O_{n+1}(T^\ast A)$ vanishes, for any choice of quantization to order $n$ of the theory on the annulus.  

    Let $D$ be a disc containing the annulus $A$.  We will show that every quantization of the theory on $T^\ast A$ extends to order $n$ extends to one on $T^\ast D$.  Choices in quantization are, at each order in $\delta$, parametrized by $H^0(X)[N]$.   Since the map $H^0(T^\ast D) \to H^0(T^\ast A)$ is surjective any quantization on $T^\ast A$ extends to one on $T^\ast D$.

    Finally, since $H^1(T^\ast D) = 0$, there are no obstructions to quantizing on $T^\ast D$. Again using the fact that obstructions are local, we find that there can never be obstructions to quantizing on $T^\ast A$, and so that there can't be any obstruction to quantizing on $T^\ast \Sigma$ for any $\Sigma$.

    Let us discuss the final case of the theorem, where we consider  holomorphically translation-invariant quantizations on $\R \times \C^2$. In this case, the obstruction-deformation complex becomes
    $$
    \C[\eta_3,\eta_5,\dots]/ \C \otimes \C[N] \otimes \C((\eps)) \otimes \C[\d z_1, \d z_2] [5] 
    $$
    where $\d z_i$ have degree $1$.  There are no $sl_2$ invariant elements in degrees $-1$ or $1$ so there is no possible obstruction. As before, $\C^\times$-invariance fixes the power of $\eps$ that can appear. Therefore, we find that the obstruction-deformation complex in degree $0$ consists of 
    $$\eta_5 \C[N] \oplus \eta_3 \d z_1 \d z_2 \C[N].$$
    This proves the theorem in this case.

    Next, we need to write explicitly  the possible Lagrangians we can write down at each term in the expansion.  To do this, in the case of a general conical $X$ with $H^3(X) = 0$, we simply need to write down deformations of the classical Lagrangian corresponding to $H^2(X)$ and to the one-dimensional space $H^0(X)$.  

    As is well known, a class in $H^2(X)$ gives rise to a first-order deformation of the $\ast$-product on $\Oo_X$. We can write this first-order deformation corresponding to a cohomology class $w \in H^2(X)$ as 
    $$
    \alpha \ast_\eps \beta + \gamma \alpha \ast_{\eps,w} \beta
    $$
    where the parameter $\gamma$ has square zero. Then, we find a corresponding first-order  deformation of the Lagrangian at $k$ loops by adding the term
    $$
    \gamma \delta^k \eps^{-k}\int \alpha \ast_{\eps} \alpha \ast_{\eps,w} \alpha. 
    $$
    to the Lagrangian of our theory.

    The class in $H^0(X)$ corresponds to adding on the term 
    $$
     \delta^{k}\eps^{-k}\left( \tfrac{1}{2} \int \alpha \d \alpha + \tfrac{1}{3} \int \alpha \ast_\eps \alpha \ast_\eps \alpha\right) 
    $$
    at $k$ loops.

    Note that this second term can be reabsorbed by a change of coordinates on the parameter $\delta$. If we perform the change of coordinates 
    $$
    \delta \mapsto \delta - \gamma \delta^{k} \eps^{-k} 
    $$
    (where $\gamma^2 = 0$) then (if $S$ denotes the action functional of our theory) we find
    $$
    \delta^{-1} S \mapsto \delta^{-1} S + \gamma \delta^{k-1} \eps^{-k} S
    $$
    generating the second term we write down at $k$ loops. 

    In fact, the first term can in a certain sense be thought of as being part of a choice of coordinates.  If, instead of our original non-commutative product $\ast_\eps$ on $\Oo_X$, we chose the product
    $$
    \ast_\eps +  \delta^{k} \eps^{-k} \eps \ast_{\eps,w} 
    $$
    we would generate the first term we wrote down.  

    The case of $\R \times \C^2$ is slightly different, as one of the deformations is associated to the constant $(2,0)$ form $\d z_1 \d z_2$.  In this case the two extra terms that can appear in the Lagrangian at $k$ loops are
    \begin{align*} 
    &\delta^{k}\eps^{-k}\tfrac{1}{2} \int \alpha \d \alpha + \tfrac{1}{3} \int \alpha \ast_\eps \alpha \ast_\eps \alpha\\ 
    \delta^k \eps^{-k} \eps \dpa{\eps} \int \alpha \ast_\eps \alpha \ast_\eps \alpha.
    \end{align*}
    As before, the first term can be absorbed into a change of coordinates on the parameter $\delta$. The second term can be absorbed into a change of coordinates of the parameter $\eps$. The required coordinate change is
    $$
    \eps \mapsto \eps + \eps \delta^k \eps^{1-k}
    $$

    \end{proof}

    \begin{lemma}
     Every holomorphically translation invariant quantization of the theory on $\C^2$ can be realized in way only involving the parameters $\eps$, $\delta$ and $\eps^{-1} \delta$. 
    \end{lemma}
    \begin{proof}
    Every holomorphically translation invariant Lagrangian is equivalent to one that only involves differentiation with respect to $z_1$ and $z_2$, and not with respect to the remaining parameters $t,\br{z}_i$, or to the odd parameters $\d t$, $\d \zbar_i$.  Every such is therefore a linear combination of Lagrangians of the form
    $$
    \alpha \mapsto \int \d z_1 \d z_2 \Phi\left( \alpha \wedge D_1 \alpha \dots \wedge D_n \alpha\right)
    $$
    where $D_i \in \C[\dpa{z_i}]$ and $\Phi$ is an invariant polynomial of degree $n+1$ on $\mf{gl}_N$.  If the differential operators $D_i$ are of degree $k_i$, then this Lagrangian is of weight $2 - \sum k_i$ under the $\C^\times$ action on $\C^2$ which scales the $z_i$.  Now, we are interested in the part of the obstruction-deformation complex at $k$ loops which is scale invariant.  For a Lagrangian $\mc{L}$ of weight $l$ under the $\C^\times$ action on $\C^2$, then $\eps^{2-l-k} \delta^{k} \mc{L}$ can be included in the action functional $S$ at $k$ loops (bearing in mind that it is $\delta^{-1} S$ that we need to be $\C^\times$ invariant, and that both $\delta$ and $\eps$ have weight $2$ under the $\C^\times$ action on $\C^2$).Since $l \le 2$, $2-l-k \ge -k$, so that we only find $\eps^{-r} \delta^k$ when $r \ge k$.  
    \end{proof} 

    \section{Reducing the $5d$ theory to $3$ dimensions}
\label{appendix:3d_reduction}
In this appendix I will give a proof of proposition \ref{prop:compactification}, which describes the theory obtained from the $5$-dimensional theory by reduction to $3$ dimensions.  Let me restate the proposition first.

Recall that the $3$-dimensional theory is  obtained from the $5$-dimensional theory on $(\R_t \times \C_w \setminus 0) \times \C_z$ by projecting along the map
\begin{align*} 
 (\R_t \times \C_w \setminus 0) \times \C_z & \mapsto \R \times \C_z \\
(t,w,z) & \mapsto ( (t^2 + w \wbar)^{1/2}, z). 
\end{align*} 

\begin{proposition}
The field theory on $\R^+_r \times \C_z$ obtained  by compactifying the $5$-dimensional field theory on $S^2$ has space of fields consisting of a series of partial $1$-forms valued in $\mf{gl}_{k+1}$,  
$$
A^i = A^i_r \d r + A^i_{\zbar} \d \zbar
$$
for $i \ge 0$. These are conveniently arranged into a series
\begin{align*} 
A &= \sum_{i \ge 0} w^i A^i \\
A^i &= A^i_r \d r + A^i_{\zbar} \d \zbar   
\end{align*}
which  defines a one-parameter family of partial connections, or equivalently, a  partial connection for the Lie algebra $\mf{gl}_{k+1}[w]$.

In addition, we have a sequence of adjoint-valued scalar fields  $B^{i}$ for $i < 0$, which  can be  arranged into a series
$$
B = \sum_{i > 0} B^{-i} w^{-i}. 
$$ 
The action functional is (when we turn off the flux $F$ induced by the $M5$ brane) 
$$
\int_{r,z} \oint_{w} \d z \d w  B F(A) 
$$
where the curvature $F(A)$ is defined in the non-commutative sense
$$
F(A) = \d A + \tfrac{1}{2}[A,A]
$$
and $[A,A]$ is defined in the Moyal algebra where $z,w$ commute to $\eps$.  In this action, the contour integral over $w$ is treated as formal operation: we are not now treating $w$ as a coordinate on the space-time.

There is gauge symmetry generated by
\begin{align*} 
X &= \sum_{i \ge 0} X^i w_i \\
X^i  &\in \mf{gl}_{k+1} 
\end{align*} 
whose infinitesimal action on the space  of fields is given by
$$
A \mapsto A + \eta \left( \d \zbar \dpa{\zbar} X + \d r \dpa{r} X +[X,A] \right)  
$$
where $\eta$ is a parameter of square  zero, and the commutator $[X,A]$ is defined using the Moyal product.
 
Turning on the flux associated to $N$ $M5$ branes, according to the prescription given in section \ref{section:quantum_holography}, leads to the deformation of the action by adding a term 
$$
- \sum_{\substack{n \ge 1 \\ n \text{ odd} } } \frac{N  \eps^{n}}{  2^{n+1}n } \int_{z,r}\oint_w  w^{-n}  \left( \partial_z^n A \right)  A \d z \d w.  
$$
This deformation also changes the gauge transformations of the field $B$.  The deformed infinitesemal gauge transformation is
$$
B \mapsto B + \eta [X,B] - N\delta  \sum_{\substack{n \ge 1 \\ n \text{ odd}}  } \frac{ \eps^{n-1} } {2^n n } \partial_z^n X \cdot w^{-n}. 
$$
(Here we are using that action of the ring $\C[w]$ on the space $w^{-1} \C[w^{-1}]$. Recall that $X$ is a function of $w$ as well as of the space-time coordinates).  
\end{proposition}
\begin{proof}
Let us identify
$$
(\R_t \times \C_w \setminus 0) \times \C_z \iso \R^+_r \times \C_z \times S^2
$$
via the map
$$
(t,w,z) \mapsto \left(r = (w \wbar + t^2)^{1/2}, z, \frac{(w,t)}{r} \right). 
$$
Here we are identifying $S^2$ with unit sphere $\{w \wbar + t^2 = 1\}$ inside $\C_w \times \R_t$. 

Let us give $(\R_t \times \C_w \setminus 0)\times \C_z$ a metric which, under the isomorphism above, becomes  a product of a metric on $\R^+_r \times \C_z$ with the standard metric on $S^2$. The metric is given by
$$
\d z \d\zbar + r^{-2}(\d w \d \wbar + \d t^2). 
$$
Note that the only role the metric plays is that it gives a choice of gauge for the $5$-dimensional gauge theory.

The manifold $\R_t  \times \C_w$ has a codimension $1$ complexified foliation, given by the subbundle of the tangent bundle spanned by $\partial_{t}$ and $\partial_{\wbar}$. The intersection of this with the complexified tangent bundle of $S^2$ is a rank $1$ complexified foliation,  spanned  by the vector field 
$$V =  -  w\partial_{t} + 2 t \partial_{\wbar}.$$ 
This vector field preserves the metric we have chosen, and also the functions $z,\zbar,r$.  We can thus view it as a map from $\R^+_r \times \C_z$ to $\mf{so}(3,\C)$, the Lie algebra of isometries of $S^2$.  One can check that, using the Killing form on $\mf{so}(3)$, the norm $\ip{V,\br{V}}$ is constant, so we should treat $V$ as being a constant map to $\mf{so}(3,\C)$.
  
Let us introduce a one-form $\theta$ defined away from the origin on $\R_t \times \C_w$, which is a linear combination of $\d t$ and $\d \wbar$ and which satisfies $\ip{\theta,V} = 1$. Explicitly,
$$
\theta =r^{-2}  \left(- \wbar \d t + \tfrac{1}{2} t \d \wbar   \right) .  
$$
Then, every field of the $5$-dimensional gauge theory on $(\R_t \times \C_w \setminus 0) \times \C_z$ can be written as a linear combination 
$$
A_\theta \theta + A_r \d r + A_{\zbar} \d \zbar
$$
where
$$
A_\theta, A_{\zbar}, \ A_r   \in \cinfty(\R^+_r \times \C_z) \otimes \cinfty(S^2) \otimes \mf{gl}_{k+1}.  
$$
From the point of view of $3$ dimensions, since $\d \theta$ is dual to a vector field which is tangent to $S^2$, we can view it as a $3$-dimensional scalar.  We let $B = A_\theta$, and $A = A_r \d r + A_{\zbar} \d \zbar$. We thus find that the space of fields in $3$ dimensions consists of a scalar $B$, and a partial connnection $A$, each valued in $\cinfty(S^2) \otimes \mf{gl}_{k+1}$. 

Let us now write down quadratic part of the action functional (when $N = 0$). It is 
$$
\int B \d A \theta \d z \d w - \int A \mc{L}_V A \theta \d z \d w.  
$$ 
In addition, there is a gauge symmetry with infinitesimal  generator $X$ which acts on the fields $B,A$ by
$$
(B,A) \mapsto (B + \eps \mc{L}_V X, A + \eps \d  X + \eps [X,A]). 
$$
Here $\mc{L}_V$ indicates the Lie derivative with respect to the vector field $V$. 

To proceed further, let us analyze the action of the vector field $V$ on the space $\cinfty(S^2)$.  The vector field $V$ is inside the Lie algebra $\mf{so}(3) = \mf{sl}(2)$, and is a nilpotent element.  We can choose a basis of $\mf{sl}(2)$ so that the standard elements $(e,f,h)$ are represented by $(V,-\br{V},-[V,\br{V}])$. We have
$$
h = - [V,\br{V}] =  2 w \dpa{w} - 2 \wbar \dpa{\wbar}. 
$$ 

The space $\cinfty(S^2)$ is, as a representation of $SU(2)$, the sum of all integer spin representations, where each such representation appears once. We can identify it with the space
$$
\oplus_{n \ge 0} \Sym^{2n}(\C^2)
$$
where $\C^2$ is the fundamental representation of $SU(2)$.

It follows that kernel of the vector field $V$ is spanned by the highest-weight vectors in each integer spin representation, and the cokernel is spanned by the highest-weight vectors.  Using the embedding $S^2 \into \R \times \C\setminus 0$, the highest-weight vector in the representation of spin $n$ is the restriction of the function $w^n$, whereas the lowest-weight vector in the representation is the restriction of the function $\wbar^n$.

The form $\d \theta \d w$, when restricted to $S^2$, is a volume form invariant under $SU(2)$, and integration against this form provides a non-degenerate pairing  between the highest and lowest weight vectors. 

Now, the components of the field $B$ which are not lowest weight vectors can be gauged away. Therefore we can choose a gauge where $\mc{L}_{\br V} B = 0$.  This reduces the gauge symmetry to those generators $X$ which satisfy $\mc{L}_V X = 0$. 

Once we choose this gauge, the components of the field $A$ which are not highest weight vectors are non-dynamical. This is because in the term $\int B \d A$ in the action, the components of $A$ which are not highest weight vectors paired with the components of $B$ which are not lowest weight vectors, which we have removed by gauge symmetry. 

Further, the term $\int A \mc{L}_V A$ in the action means that the components of $A$ which are not highest weight vectors, and so not in the kernel of $\mc{L}_V$, acquire a mass. These components of $A$ can be treated as auxiliary fields, and removed. 

We thus left with a basis  of the space of fields consisting of 
$$
A^i = A^i_r \d r + A^i_{\zbar} \d \zbar
$$
which is, as a function on $S^2$, the highest weight vector in the spin $i$ spin representation of $SU(2)$ appearing in $\cinfty(S^2)$; together with $B^{-i-1}$, which is the lowest weight vector in the spin $i$ representation. 

With appropriate normalization of the highest and lowest weight vectors, the action functional on these fields is the one stated in the proposition.

\end{proof}

\bibliography{Mtheoryomega}

\newcommand{\etalchar}[1]{$^{#1}$}
\def\cprime{$'$}
\begin{thebibliography}{AHBC{\etalchar{+}}10}




\bibitem[AGT10]{AldGaiTac10}
L.~F. Alday, D.~Gaiotto and Y.~Tachikawa, \textsl{ Liouville correlation
  functions from four-dimensional gauge theories},
\newblock Letters in Mathematical Physics \textbf{ 91}(2), 167--197 (2010).

\bibitem[BCOV94]{BerCecOog94}
M.~Bershadsky, S.~Cecotti, H.~Ooguri and C.~Vafa, \textsl{ Kodaira-{S}pencer
  theory of gravity and exact results for quantum string amplitudes},
\newblock Comm. Math. Phys. \textbf{ 165}(2), 311--427 (1994).

\bibitem[BD04]{BeiDri04}
A.~Beilinson and V.~Drinfeld,
\newblock \textsl{ Chiral algebras}, volume~51 of \textsl{ American
  Mathematical Society Colloquium Publications},
\newblock American Mathematical Society, Providence, RI, 2004.

\bibitem[BDG15]{BulDimGai15}
M.~Bullimore, T.~Dimofte and D.~Gaiotto, \textsl{ The Coulomb Branch of 3d N=4
  Theories},
\newblock arXiv preprint arXiv:1503.04817  (2015).

\bibitem[BFN14]{BraFinNak14}
A.~Braverman, M.~Finkelberg and H.~Nakajima, \textsl{ Instanton moduli spaces
  and W-algebras},
\newblock arXiv preprint arXiv:1406.2381  (2014).

\bibitem[BRvR15]{BeeRasvan15}
C.~Beem, L.~Rastelli and B.~C. van Rees, \textsl{ W symmetry in six
  dimensions},
\newblock Journal of High Energy Physics \textbf{ 2015}(5), 1--38 (2015).


\bibitem[CG16]{CosGwi11}
K.~Costello and O.~Gwilliam,
\newblock \textsl{ Factorization algebras in perturbative quantum field
  theory},
\newblock Cambridge {U}niversity {P}ress, 2016.

\bibitem[CL15]{CosLi15}
K.~Costello and S.~Li, \textsl{ Quantization of open-closed BCOV theory, I},
\newblock (2015), {arXiv:1505.06703}.

\bibitem[CL16]{CosLi16}
K.~Costello and S.~Li, \textsl{ Twisted supergravity and its quantization},
\newblock (2016).

\bibitem[Cos13]{Cos13}
K.~Costello, \textsl{ Supersymmetric gauge theory and the {Y}angian},
\newblock (2013), {arXiv:1303.2632}.

\bibitem[Cos16]{Cos16}
\textsl{ Holography and Koszul duality: the example of the $M2$ brane},
\newblock (2016).

\bibitem[DHSV08]{DijHolSulVaf08}
R.~Dijkgraaf, L.~Hollands, P.~Su{\l}kowski and C.~Vafa, \textsl{ Supersymmetric
  gauge theories, intersecting branes and free fermions},
\newblock Journal of High Energy Physics \textbf{ 2008}(02), 106 (2008).

\bibitem[DK97]{DonKro97}
S.~K. Donaldson and P.~B. Kronheimer,
\newblock \textsl{ The geometry of four-manifolds},
\newblock Oxford University Press, 1997.


\bibitem[GG12]{GabGop12}
M.~Gaberdiel and R.~Gopakumar, \textsl{ Minimal model holography},
\newblock (2012), {arXiv:1207.6697}.


\bibitem[Gro95]{Gro95}
I.~Grojnowski, \textsl{ Instantons and affine algebras I: the Hilbert scheme
  and vertex operators},
\newblock arXiv preprint alg-geom/9506020  (1995).


\bibitem[HOR12]{HelOrlRef12}
S.~Hellerman, D.~Orlando and S.~Reffert, \textsl{ The omega deformation from
  string and M-theory},
\newblock Journal of High Energy Physics \textbf{ 2012}(7), 1--36 (2012).


\bibitem[LQ84]{LodQui84}
J.-L. Loday and D.~Quillen, \textsl{ Cyclic homology and the Lie algebra
  homology of matrices},
\newblock Comment. Math. Helv. \textbf{ 59}(4), 569–591 (1984).


\bibitem[LTYZ14]{LuoTanYag14}
Y.~Luo, M.-C. Tan, J.~Yagi and Q.~Zhao, \textsl{ Omega-deformation
  of B-twisted gauge theories and the 3d-3d correspondence},
\newblock arXiv preprint arXiv:1410.1538  (2014).

\bibitem[MO12]{MauOko12}
D.~Maulik and A.~Okounkov, \textsl{ Quantum cohomology and quantum groups},
\newblock (2012), {arXiv:1211.1287}.

\bibitem[Nak97]{Nak97}
H.~Nakajima, \textsl{ Heisenberg algebra and Hilbert schemes of points on
  projective surfaces},
\newblock Annals of mathematics \textbf{ 145}(2), 379--388 (1997).

\bibitem[Nek03]{Nek03}
N.~Nekrasov, \textsl{ Seiberg-Witten prepotential from instanton counting},
\newblock Advances in Theoretical and Mathematical Physics \textbf{ 7}(5),
  831--864 (2003).

\bibitem[NT11]{NisTac11}
T.~Nishioka and Y.~Tachikawa, \textsl{ Central charges of para-Liouville and
  Toda theories from M5-branes},
\newblock Physical Review D \textbf{ 84}(4), 046009 (2011).

\bibitem[NW10]{NekWit10}
N.~Nekrasov and E.~Witten, \textsl{ The omega deformation, branes,
  integrability and Liouville theory},
\newblock Journal of High Energy Physics \textbf{ 2010}(9), 1--83 (2010).


\bibitem[DHJV16]{DijHeiJefVaf16}
  R.~Dijkgraaf, B.~Heidenreich, P.~Jefferson and C.~Vafa, \textsl{ Negative Branes,
  Supergroups and the Signature of Spacetime},
\newblock (2016), {1603.05665}.

\bibitem[SV13]{SchVas13}
O.~Schiffmann and E.~Vasserot, \textsl{ Cherednik algebras, W-algebras and the
  equivariant cohomology of the moduli space of instantons on A 2},
\newblock Publications math{\'e}matiques de l'IH{\'E}S \textbf{ 118}(1),
  213--342 (2013).

\bibitem[SW99]{SeiWit99}
N.~Seiberg and E.~Witten, \textsl{ String theory and noncommutative geometry},
\newblock J. High Energy Phys. (9), Paper 32, 93 pp. (electronic) (1999).


\bibitem[Tsy83]{Tsy83}
B.~Tsygan, \textsl{ Homology of matrix algebras over rings and the Hochschild
  homology},
\newblock Uspeki Math. Nauk. \textbf{ 38}, 217–218 (1983).

\bibitem[Wit86]{Wit86}
E.~Witten, \textsl{ Noncommutative geometry and string field theory},
\newblock Nuclear Phys. B \textbf{ 268}(2), 253--294 (1986).

\bibitem[Wit95]{Wit95a}
E.~Witten, \textsl{ String theory dynamics in various dimensions},
\newblock Nuclear Physics B \textbf{ 443}(1), 85--126 (1995).

\bibitem[Yag12]{Yag12}
J.~Yagi, \textsl{ Compactification on the Omega-background and the AGT
  correspondence},
\newblock Journal of High Energy Physics \textbf{ 2012}(9), 1--20 (2012),
  {arXiv:1205.6820}.

\bibitem[Yag14]{Yag14}
J.~Yagi, \textsl{ Omega-deformation and quantization},
\newblock arXiv preprint arXiv:1405.6714  (2014).

\end{thebibliography}

\end{document}